\documentclass[a4paper,
               DIV=10, 
               titlepage=off]{scrartcl}
\usepackage[margin=2.5cm]{geometry}

\usepackage[utf8]{inputenc}
\usepackage[english]{babel}
\usepackage{csquotes}

\usepackage{amsthm,amsmath,amsfonts,amssymb}
\usepackage{braket,bbm,nicefrac}
\usepackage{mathtools,stmaryrd} 

\usepackage[colorlinks=true,linkcolor=blue,urlcolor=blue,citecolor=blue,anchorcolor=green,pdfusetitle]{hyperref}
\usepackage[capitalize,noabbrev]{cleveref}
\usepackage{dsfont,xspace,xparse}
\usepackage{lmodern,newtxtt,microtype}
\usepackage{tabularx,xifthen,afterpage,booktabs}
\usepackage[shortlabels]{enumitem}
\usepackage{setspace,adjustbox,array} 
\newcolumntype{L}{>{\centering\arraybackslash}m{3cm}}
\usepackage{multirow} 
\usepackage{diagbox} 
\usepackage{svg}

\usepackage[affil-it]{authblk}

\usepackage{microtype,parskip} 

\usepackage{upgreek}

\usepackage{float}

\newtheorem{theorem}{Theorem}
\newtheorem{corollary}[theorem]{Corollary}

\newtheorem{lemma}[theorem]{Lemma}
\newtheorem{remark}[theorem]{Remark}
\newtheorem{claim}[theorem]{Claim}

\theoremstyle{definition}
\newtheorem{definition}[theorem]{Definition}
\newtheorem{example}[theorem]{Example}

\usepackage{thm-restate}
\newtheorem{mtheorem}{Theorem}
\newtheorem{mcorollary}[mtheorem]{Corollary}
\newtheorem{openproblem}{Open Problem}

\numberwithin{equation}{section}
\numberwithin{theorem}{section}

\newcommand{\C}{{\mathbb{C}}} 
\newcommand{\R}{{\mathbb{R}}}

\DeclareMathOperator*{\E}{\mathbb{E}}

\DeclareMathOperator{\gen}{gen} 
\DeclareMathOperator{\id}{id} 

\DeclareMathOperator*{\argmin}{argmin}
\DeclareMathOperator{\tr}{Tr}

\newcommand{\cH}{\mathcal{H}} 
\newcommand{\cB}{\mathcal{B}} 
\newcommand{\cT}{\mathcal{T}} 
\newcommand{\alg}{\mathcal{A}} 

\newcommand{\Norm}[1]{\left\|{#1}\right\|}

\newcommand{\ketbra}[2]{\ket{#1}\!\!\bra{#2}}

\crefname{lemma}{lemma}{lemmas}
\crefname{proposition}{proposition}{propositions}
\crefname{definition}{definition}{definitions}
\crefname{theorem}{theorem}{theorems}
\crefname{conjecture}{conjecture}{conjectures}
\crefname{corollary}{corollary}{corollaries}
\crefname{example}{example}{examples}
\crefname{section}{section}{sections}
\crefname{appendix}{appendix}{appendices}
\crefname{figure}{fig.}{figs.}
\crefname{equation}{eq.}{eqs.}
\crefname{table}{table}{tables}
\crefname{item}{property}{properties}
\crefname{remark}{remark}{remarks}

\usepackage[backend=biber,style=trad-alpha,sorting=nyt,doi=false,isbn=false,url=false,maxbibnames=20,maxcitenames=2,backref,backrefstyle=three]{biblatex}

\newbibmacro{string+doi+url}[1]{%
  \iffieldundef{doi}{%
    \iffieldundef{url}{%
      #1
    }{%
      \href{\thefield{url}}{#1}
    }%
  }{%
    \href{https://dx.doi.org/\thefield{doi}}{#1}
  }%
}

\DeclareFieldFormat{title}{\usebibmacro{string+doi+url}{\mkbibemph{#1}}}
\DeclareFieldFormat[article]{title}{\usebibmacro{string+doi+url}{\mkbibquote{#1}}}
\DeclareFieldFormat[incollection]{title}{\usebibmacro{string+doi+url}{\mkbibquote{#1}}}
\DeclareFieldFormat[inproceedings]{title}{\usebibmacro{string+doi+url}{\mkbibquote{#1}}}

\renewbibmacro{in:}{}

\title{Information-theoretic generalization bounds for learning from quantum data}

\author[1,2]{Matthias C.~Caro\thanks{matthias.caro@fu-berlin.de}}
\author[3]{Tom Gur\thanks{tom.gur@cl.cam.ac.uk}}
\author[4,5]{Cambyse Rouz\'{e}\thanks{cambyse.rouze@telecom-paris.fr}}
\author[6]{Daniel Stilck França\thanks{daniel.stilck\_franca@ens-lyon.fr}}
\author[3,7]{Sathyawageeswar Subramanian\thanks{ss2310@cam.ac.uk}}

\affil[1]{\normalsize Dahlem Center for Complex Quantum Systems, Freie Universit\"at Berlin, Berlin, Germany}
\affil[2]{\normalsize Institute for Quantum Information and Matter, Caltech, Pasadena, CA, USA}
\affil[3]{\normalsize Department of Computer Science and Technology, University of Cambridge, Cambridge, UK}
\affil[4]{Inria, T\'{e}l\'{e}com Paris - LTCI, Institut Polytechnique de Paris, Palaiseau, France}
\affil[5]{Zentrum Mathematik, Technische Universit\"{a}t M\"{u}nchen, Garching, Germany}
\affil[6]{\normalsize Univ Lyon, ENS Lyon, UCBL, CNRS, Inria, LIP, F-69342, Lyon Cedex 07, France}
\affil[7]{\normalsize Department of Computer Science, University of Warwick, Coventry, UK}
\date{}

\begin{document}
\maketitle

\begin{abstract}
    Learning tasks play an increasingly prominent role in quantum information and computation. They range from fundamental problems such as state discrimination and metrology over the framework of quantum probably approximately correct (PAC) learning, to the recently proposed shadow variants of state tomography. 
    However, the many directions of quantum learning theory have so far evolved separately.
    We propose a general mathematical formalism for describing quantum learning by training on classical-quantum data and then testing how well the learned hypothesis generalizes to new data.
    In this framework, we prove bounds on the expected generalization error of a quantum learner in terms of classical and quantum information-theoretic quantities measuring how strongly the learner's hypothesis depends on the specific data seen during training.
    To achieve this, we use tools from quantum optimal transport and quantum concentration inequalities to establish non-commutative versions of decoupling lemmas that underlie recent information-theoretic generalization bounds for classical machine learning.
    
    Our framework encompasses and gives intuitively accessible generalization bounds for a variety of quantum learning scenarios such as quantum state discrimination, PAC learning quantum states, quantum parameter estimation, and quantumly PAC learning classical functions.
    Thereby, our work lays a foundation for a unifying quantum information-theoretic perspective on quantum learning.
\end{abstract}

\clearpage
\tableofcontents
\clearpage

\section{Introduction}\label{section:introduction}

The intersection of machine learning and quantum physics has developed into a vibrant area of research.
On the one hand, along the lines of using (at least partially) quantum learners for classical data, there are proposals for machine learning models based on quantum circuits \cite{biamonte2017quantum, dunjko2018machine, havlivcek2019supervised}, such as the so-called variational quantum machine learning models and quantum kernel methods.
On the other hand, there has been significant progress in learning from quantum data. Inspired by ``pretty good tomography'' \cite{Aaronson2007}, viewing quantum experiments through the lens of learning from quantum data has given rise to `shadow' protocols \cite{aaronson2019shadow, huang2020predicting} that use few copies of an unknown quantum state to predict many of its properties. The learning perspective has also led to insights into the potential for quantum advantage of fully quantum over conventional experiments \cite{huang2021information, aharonov2022quantum, chen2022exponential, chen2022quantum, huang2022quantum-advantage, caro2022learning, chen2023unitarity, chen2023does}.
Moreover, from the viewpoint of computer science, quantum theory allows for new kinds of oracular access to an unknown object that is to be learned \cite{bshouty1998learning}, and thus potentially (though not always) for more efficient learning algorithms \cite{arunachalam2017survey}.
Even fundamental problems of quantum information theory, such as state or process tomography \cite{haah2016sample, odonnell2016efficient, haah2023query, zhao2023learning} or state discrimination \cite{helstrom1969quantum, holevo1974remarks, yuen1975optimum}, can be interpreted as tasks of learning from quantum data \cite{Guta2010, Sentis2019Unsupervised}.

As quantum machine learning and quantum learning theory have grown, so has the number of different quantum learning scenarios and mathematical descriptions thereof.
This is reminiscent of the plethora of approaches to generalization and sample complexity bounds in classical machine learning theory \cite{vapnik1971uniform, pollard1984convergence, littlestone1986relating, kearns1994efficient, dudley1999uniform, mcallester1999some, bousquet2002stability, bartlett2002rademacher, dwork2006calibrating}. Recently, information-theoretic generalization bounds \cite{hellstroem2023generalization}, going back to \cite{Xu2017, russo2019much}, have emerged as a promising approach towards unifying these varied results. Furthermore, they may help overcome the limitations of uniform generalization bounds \cite{zhang2017understanding, zhang2021understanding}, which have recently also been pointed out for quantum machine learning models \cite{gil-fuster2023understanding}.
However, a similarly unifying perspective on quantum learning has so far been lacking.

In the spirit of unification, we propose a general mathematical framework for quantum learning procedures that train on data composed of classical samples as well as (copies of) quantum data states, and then produce a classical and/or quantum hypothesis to be used for prediction on new classical-quantum data.
We prove that the generalization behavior of such quantum learners -- that is, how well they generalize from available training data to previously unseen data -- can be controlled through classical and quantum information-theoretic quantities, which quantify how much information the learner's hypothesis contains about the data, combined with concentration properties of the loss observables used for the training. 
We demonstrate several applications of this quantum version of the central insight from \cite{Xu2017, russo2019much}. To mention a few, it allows us to provide a new perspective on quantum state classification tasks \cite{Guta2010}, and recover the seminal result of \cite{Aaronson2007} on probably approximately correctly learning quantum states as well as the results of \cite{chung2021sample, caro2021binary, fanizza2022learning} on learning state preparation procedures.

\subsection{Main results}

Our first contribution is a unifying framework capable of capturing a wide variety of quantum learning problems. Having formulated the framework, we then use it to prove information-theoretic generalization error bounds for quantum learners and demonstrate applications of our bounds to learning quantum states, learning classical functions from entangled quantum data, and quantum state classification. 

\subsubsection{Unified information-theoretic framework}\label{subsection:unified-framework}

To give an impression of the varied landscape of quantum learning, let us briefly examine three influential learning tasks that deal with quantum information. The learner in \cite{Aaronson2007} takes as input 
classical data associated with an $n$-qubit mixed quantum state $\rho$ obtained through measurements. In particular, it receives a classical description of observables $E_1\ldots,E_m$ drawn i.i.d.~according to an unknown distribution $P$ over effect operators, and the corresponding expectation values $\tr[ E_i \rho]$. It addresses the task of Probably Approximately Correctly (PAC) learning a classical representation $\hat{\rho}$ of the unknown state, where the figure of merit is the quality of the hypothesis $\hat{\rho}$ in approximating the expectation values $\tr[O\rho ]$ of new observables $O\sim P$ drawn from the same distribution $P$. 

\cite{bshouty1998learning} proposed a quantum input model for learning classical functions: The quantum algorithm takes copies of the superposition example $\sum_x \sqrt{P(x)}\ket{x,f(x)}$.
This can be viewed as a quantum version of classical access to pairs $(x_i,f(x_i))$ where the input points $x_i\sim P$ are drawn i.i.d.~from the unknown data distribution $P$. The learner is tasked with producing a classical hypothesis $h$ that, with high probability, agrees with $f$ on a set of inputs that has large probability under $P$.

Quantum parameter estimation is a fundamental task in quantum metrology \cite{giovannetti2006quantum} and quantum sensing applications \cite{degen2017quantum}.
Here, given access to copies of a parameter-dependent quantum state $\rho (\theta)$ with an unknown parameter vector $\theta$, one aims to perform a measurement and observe an outcome $\hat{\theta}$ that serves as an accurate estimate for $\theta$. 
(The mapping $\theta\mapsto\rho(\theta)$ may be known beforehand.)
We can interpret the task of identifying such a measurement as one of learning a (probabilistic) mapping from quantum data to a classical parameter vector that approximates the true underlying parameter-state connection.

At first sight, the three tasks described above differ in important aspects: The learners handle different kinds of inputs -- classical data, pure states, or mixed states -- and produce different kinds of outputs -- a classically described quantum state, a classical function, or a parameter estimate. 
Indeed, the approaches and techniques employed in solving these problems vary widely. 
This raises the following question: Can we define and analyze quantum learners in a framework that simultaneously captures these (and more) different quantum learning tasks?
Taking inspiration from recent developments at the intersection of classical learning theory and information theory, we now propose such a framework.

\paragraph{Learners as maps.}
It is well established that classical randomized (supervised) learning algorithms can be modeled as channels. They take as input a random variable called the \emph{training data}, which is a set $S=(Z_1,\ldots,Z_m)$ of $|S|=m$ i.i.d.~data points drawn from a probability distribution $P$ over an \emph{instance space} $\mathsf{Z}$. 
The output of a learner is a random variable called the \emph{hypothesis} taking values in a \emph{hypothesis space} $\mathsf{W}$. We often think of the input domain $\mathsf{Z}$ as being a Cartesian product $\mathsf{Z}=\mathsf{X}\times\mathsf{Y}$, and the hypothesis space $\mathsf{W}$ as a subset of $\mathsf{Y}^{\mathsf{X}}$, so that a hypothesis is in fact a (randomized) \emph{function} $w:\mathsf{X}\to\mathsf{Y}$. 
The learner can then be identified with a conditional probability distribution $P(W|S)$ for the hypothesis given the data.

In analogy, we propose to think of \emph{quantum} learning algorithms $\alg$ as quantum procedures that take as input data represented by a quantum state $\rho$ coming from a quantum instance space $\mathcal{Z}$. The output of a quantum learner is a hypothesis state taking values in a space $\mathcal{W}$. In particular, without loss of generality, we can take $\mathcal{Z}$ to be a space of \emph{classical-quantum} ``CQ states'', that is, states of the form 
\begin{equation}
\label{eq:intro-cq}
    \rho 
    = \E_{S\sim P^m}\left[\ketbra{S}{S}\otimes \rho(S)^{\vphantom{make the brackets bigger}}\right], 
\end{equation}
where $\rho(S)$ is a quantum state on the Hilbert space $\mathcal{H}_{\mathrm{train}}$. 
Typically, we consider $\mathcal{H}_{\mathrm{train}}\cong \bigotimes_{i=1}^m\C^d$, where $d$ is the local dimension, and assume the quantum training data to factorize as $\rho(S)=\bigotimes_{i=1}^m \rho(Z_i)$ for $d$-dimensional states $\rho(Z_i)$. 
Similarly $\mathcal{W}$ consists of states of the form 
\begin{equation}
\label{eq:intro-hypstate}
    \sigma^{\alg}
    =\E_{(S,W)\sim P^{\alg}}\left[\ketbra{S,W}{S,W}\otimes \sigma^{\alg}(S,W)^{\vphantom{make the brackets bigger}}\right],
\end{equation}
where $P^\alg$ is a joint distribution over data and hypothesis induced by the learner, and $\sigma^{\alg}(S,W)$ is a quantum state on the Hilbert space $\mathcal{H}_{\mathrm{hyp}}$.
The learning procedure consists of two steps that can be iterated: measurement and post-processing. The measurements may be implemented by positive operator-valued measures (POVMs), and the associated instruments.
Here, a POVM maps states to classical probability distributions over outcomes, and the instruments give the corresponding mappings to post-measurement states.
Post-processing can include randomized classical post-processing of the measurement outcomes as well as quantum post-processing of the post-measurement states. Learners that use only a single iteration are \emph{non-adaptive}.

\paragraph{Risk for classical learners.} In classical learning theory, the performance of a hypothesis on a data point is evaluated by a loss function $\ell:\mathsf{W}\times\mathsf{Z}\to\R_{\geq 0}$. Accordingly, the \emph{true risk} of a hypothesis $w\in\mathsf{W}$ relative to the distribution $P$ is 
\begin{equation}
    R_P(w)
    = \E_{Z\sim P}[\ell (w,Z)]. 
\end{equation}
The goal of a learner is to output a randomized hypothesis $W$ that has small true risk $R_P(W)$, either in expectation or with high success probability. However, the data distribution $P$ is typically unknown, so the learner cannot directly evaluate $R_P(w)$ for a candidate hypothesis $w$.
Instead, the average loss of a hypothesis on available training data serves as a proxy for the true risk. For training data $S=(Z_1,\ldots,Z_m)$ and hypothesis $w\in\mathsf{W}$, the \emph{empirical risk} is defined by
\begin{equation}\label{eq:empirical-risk-classical}
    \hat{R}_S(w)
    = \frac{1}{m}\sum_{i=1}^m \ell(w,Z_i)\, .
\end{equation}
In contrast to $R_P(w)$, a classical learner with access to $S$ can in principle evaluate $\hat{R}_S(w)$ for any $w\in\mathsf{W}$.
When the focus is on the average performance of a learner, the quality of $\hat{R}_S(W)$ as a proxy for $R_P(W)$ may be quantified by the \emph{expected generalization error}
\begin{equation}
\label{eq:gen-err-def-xr17}
    \operatorname{gen}_P (\alg)
    = \E_{(S,W)\sim P^{\alg}} \left[R_P(W) - \hat{R}_S(W)\right] .
\end{equation}
In this work, we refer to bounds on $\operatorname{gen}_P (\alg)$ simply as generalization bounds\footnote{Concentration bounds for the generalization error are also often of interest, but we primarily consider bounds in expectation in this article.}. Such generalization bounds then give rise to guarantees on when successful training, as quantified by small empirical risk, also leads to small true risk.

\paragraph{Risk for quantum learners.} In translating the above classical recipe for evaluating the performance of a learner to the quantum scenario, we encounter a fundamental obstacle: In general, quantum data cannot be reused. Quantum training data that has been used for training may be irreversibly modified by measurements and post-processing, and cannot simply be reused for evaluating the empirical risk of a hypothesis obtained at the end of the training process. 

Therefore, we extend our quantum framework by introducing an additional quantum system to capture test data. That is, we now allow $\rho(S)$ in the quantum data state of \Cref{eq:intro-cq} to be states on a composite Hilbert space $\mathcal{H}_{\mathrm{data}}=\mathcal{H}_{\mathrm{test}}\otimes \mathcal{H}_{\mathrm{train}}$. 
Note that $\rho(S)$ can be correlated or even entangled across the test-train bipartition of the data Hilbert space.
As before, the action of the learner on the training data subsystem then leads to a hypothesis state as in \Cref{eq:intro-hypstate},
with the difference that $\sigma^{\alg}(S,W)$ is now a quantum state on the Hilbert space $\mathcal{H}_{\mathrm{test}}\otimes \mathcal{H}_{\mathrm{hyp}}$.

Lifting the notion of loss function to the stature of a quantum observable, we work with a family of (Hermitian and nonnegative) loss observables $\{L(S,W)\}\subset\mathcal{B}(\mathcal{H}_{\mathrm{test}}\otimes \mathcal{H}_{\mathrm{hyp}})$. We then define the \emph{expected empirical risk} of the quantum learner $\alg$ as the expectation value of the observable $L(S,W)$ on the hypothesis state $\sigma (S,W)$, further averaged over $P^{\alg}$. That is, 
\begin{equation}
    \hat{R}_\rho (\alg)
    = \E_{(S,W)\sim P^{\alg}}\left[\tr[L(S,W)\sigma^{\alg}(S,W)]^{\vphantom{make the brackets bigger}}\right]\,.
\end{equation}
In contrast, we define the \emph{expected true risk} of $\alg$ as
\begin{equation}\label{eq:informal-empirical-risk}
    R_\rho (\alg)
    = \E_{(\bar{S},\bar{W})\sim P^{\alg}_{\mathsf{Z}^m}\otimes P^{\alg}_{\mathsf{W}}}\left[\tr\left[L(\bar{S},\bar{W}) \left(\rho_{\mathrm{test}}(\bar{S})\otimes \sigma^{\alg}_{\mathrm{hyp}}(\bar{S},\bar{W})\right)\right]\right] ,
\end{equation}
where we have ``decoupled'' the quantum test and training data systems before letting the learner act, and we have also decoupled the classical training data and hypothesis random variables. 
Here, a state with a subscript denotes a reduced density matrix obtained by tracing out the other subsystems. Mathematically, this is achieved by a partial trace, for example, we have $\rho_{\mathrm{test}}(\bar{S}) = \tr_{\mathrm{train}}[\rho(\bar{S})]$ and $\sigma^{\alg}_{\mathrm{hyp}}(\bar{S},\bar{W}) = \tr_{\mathrm{test}}[\sigma^{\alg}(\bar{S},\bar{W})]$.
As in the classical case, we define the expected generalization error as the difference between expected true and empirical risks, 
\begin{equation}\label{eq:informal-true-risk}
    \operatorname{gen}_\rho (\alg)
    = R_\rho (\alg) - \hat{R}_\rho (\alg).
\end{equation}
Our main goal is to bound $\operatorname{gen}_\rho (\alg)$ in terms of properties of the CQ data $\rho$, the loss observables $L(S,W)$, and the learner $\alg$.
We visualize our framework for quantum learners in \Cref{fig:framework}.

In theory, we can consider alternative notions of decoupling, and indeed, alternative definitions for the quantum risks. These notions may also be tailored differently in order to capture the essence of what is important in the learning task at hand. In \Cref{subsection:comparison-frameworks}, we motivate our decoupling approach to the definition of true risk and generalization error by a comparison to the classical framework, and demonstrate how it extends established notions from classical learning theory.
In addition to reducing to the expected empirical and true risk in the classical case, these choices give rise to natural notions of risks and generalization error for a variety of quantum learning tasks (see \Cref{sec:applications}). 
Moreover, our definitions account for the desiderata that $\hat{R}_\rho (\alg)$ should incorporate all aspects in which the learner's actions ``contaminate'' the test data, whereas the test data in $R_\rho (\alg)$, both classical and quantum, must be completely untarnished by the learner.
This justifies \Cref{eq:informal-empirical-risk,eq:informal-true-risk} as the quantum extension of \cite{Xu2017}'s change-of-measure/decoupling perspective on classical generalization analysis.

\begin{figure}
    \centering
    \includegraphics[width=0.9\textwidth]{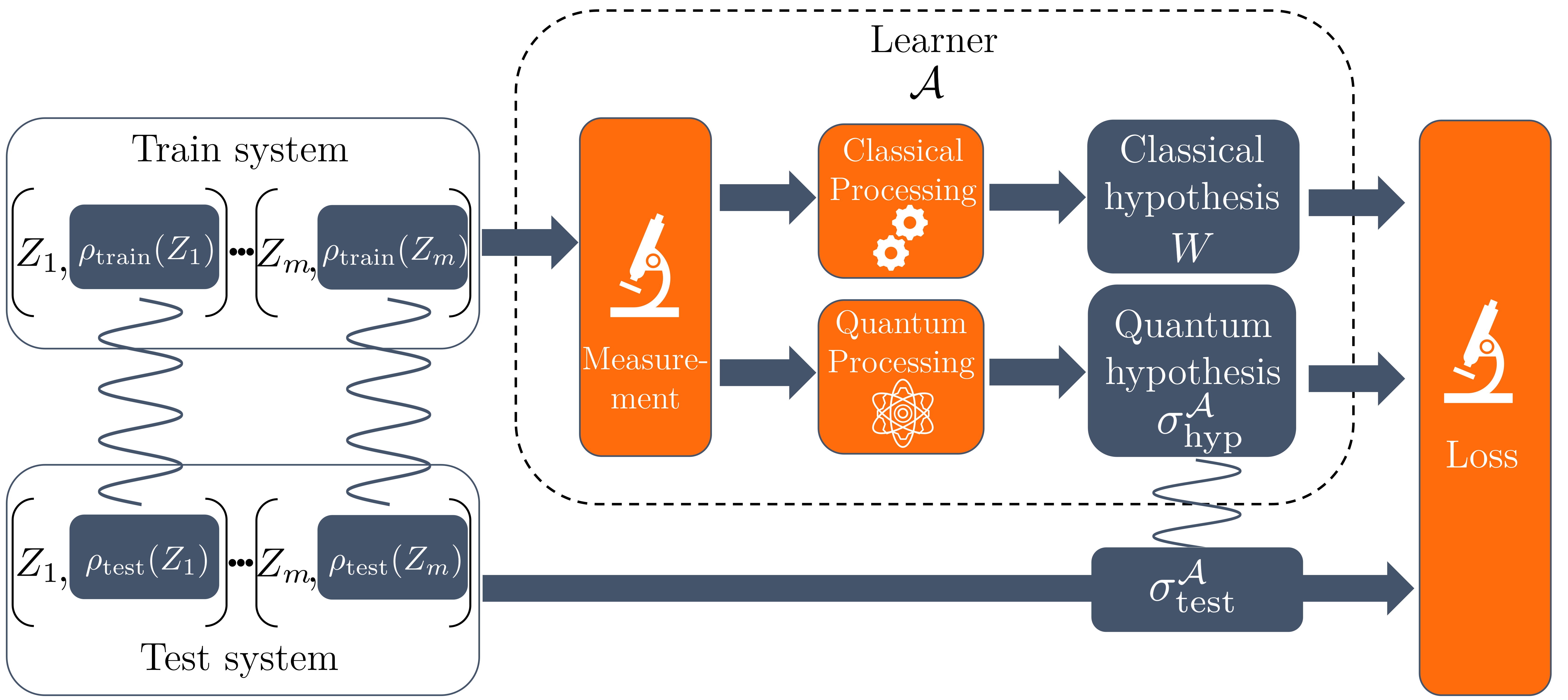}
    \caption{\textbf{Framework for learning from classical-quantum data:} The quantum learner $\alg$ acts on the classical data and on the training subsystem of the quantum data via a measurement followed by classical and quantum post-processing. The performance of the resulting classical and quantum hypotheses are then evaluated via a loss measurement that also takes the testing subsystem of the quantum data into account. The training and testing subsystems may initially be correlated or even entangled.}
    \label{fig:framework}
\end{figure}

\subsubsection{Generalization error bounds}\label{section:generelization-bounds-informal}

\begin{table}[htb]
  \begin{center}
      \renewcommand{\arraystretch}{1.4}
      \begin{tabular}{|l|c|}
        \toprule
        Object & Notation \\
        \midrule
        Input data CQ state & $\rho$\\
        Learner output & $\sigma^\alg$\\
        Decoupled learner output & $\tau^{\alg} = \rho_{\mathrm{test}}\otimes\sigma^\alg_{\mathrm{hyp}}$\\
        Loss observable & $L$ \\
        Probability density of classical data & $P$ \\
        Quantum mutual information & $I(\cdot~\!;\cdot)_{\small{\bullet}}$\\
        Holevo information & $\chi(\{\cdot~\!,\cdot\})$\\
        Quantum log-MGF bound & $\psi_{\pm}$ \\
        Classical log-MGF bound & $\phi_{\pm}$ \\
        \bottomrule
      \end{tabular}
  \end{center}
  \caption{Notation for the various mathematical objects appearing in this section.}
  \label{table:notation}
\end{table}

\paragraph{Assumptions.} 
The framework and formalism described above can capture a variety of learning scenarios. In order to prove bounds on the generalization error, we will assume mild properties to be satisfied by the learner, the data, and the loss observables.
To avoid clutter, in the following we suppress dependencies on $S,W$ in the notation where it is clear from the context. That is, we write $\sigma$ instead of $\sigma (S,W)$ and $L$ instead of $L(S,W)$.
Additionally, we will frequently denote $\tau^{\alg} = \tau^{\alg}(s,w) =\rho_{\mathrm{test}} (s)\otimes \sigma^{\alg}_{\mathrm{hyp}}(s,w) =\rho_{\mathrm{test}}\otimes \sigma^{\alg}_{\mathrm{hyp}}$. 
Here and throughout, a state with a subsystem subscript denotes the reduced state on that subsystem.
(See \Cref{table:notation} for an overview of notation.)

As in many classical works on this subject, bounds on the moment generating function (MGF) allow for characterizing the concentration properties of the value of the loss observable around its expectation value. However, due to the noncommutative setting at hand, we consider the following two generalizations:
\begin{itemize}
    \item[({\tiny{QMGF}})] \textbf{Quantum MGF/tail bound}: For every $(s,w)\in\mathsf{Z}^m\times\mathsf{W}$, let the logarithmic quantum moment generating function of the loss observable $L$ with respect to the product state $\tau^{\alg}$ 
    be bounded by convex functions $\psi_+, \psi_-:\mathbb{R}\to\mathbb{R}$ which satisfy $\psi_\pm (0)=\psi'_\pm (0)=0$, i.e.\    
    \begin{equation}
    \label{eqinf:qmgf}
        \log \tr\left[\tau^{\alg} e^{\lambda \left( L - \tr[L \tau^{\alg}]\mathbbm{1}\right)}\right] 
        \leq \begin{cases} \psi_+ (\lambda)\quad \textrm{if } \lambda \geq 0 \\ \psi_- (\lambda)\quad \textrm{if } \lambda <0 \end{cases} \tag{\texttt{QMGF}}. 
    \end{equation}

    \item[({\tiny{CMGF}})] \textbf{Classical MGF/tail bound:}
    For every $w\in\mathsf{W}$, let the logarithmic moment generating function of the expectation value $\tr[L\tau^{\alg}]$ of the loss observable $L$ in the product state $\tau^{\alg}$, viewed as a random variable, be bounded by convex functions $\phi_+, \phi_-:\mathbb{R}\to\mathbb{R}$ which satisfy $\phi_\pm (0)=\phi'_\pm (0)=0$, i.e.,  
    \begin{equation}
    \label{eqinf:cmgf}
           \log \E_{S\sim P^m}\left[e^{\lambda \left(\tr[L \tau^{\alg}] - \E_{S\sim P^m}\left[\tr[L\tau^{\alg}]\right]\right)}\right]
            \leq \begin{cases} \phi_+ (\lambda)\quad \textrm{if } \lambda \geq 0 \\ \phi_- (\lambda)\quad \textrm{if } \lambda <0 \end{cases} \tag{\texttt{CMGF}}.
    \end{equation} 
\end{itemize}

If the convex functions $\psi_\pm$ and $\phi_\pm$ are of the form $\lambda\mapsto \tfrac{\alpha^2 \lambda^2}{2}$ and $\lambda\mapsto \tfrac{\beta^2 \lambda^2}{2}$, respectively, then we speak of an $\alpha$-sub-gaussian QMGF and a $\beta$-sub-gaussian CMGF. We describe some scenarios of interest in which these assumptions are satisfied in \Cref{section:applications-informal}.

\paragraph{Generalization bounds.}
Can the generalization error of the quantum learner $\alg$ on the data $\rho$ be controlled in terms of quantities that we can interpret, giving us a handle on how one can produce a hypothesis that attains a balance between fitting the training data and performing well on unseen data? We answer this question in the affirmative, and show that assuming classical and quantum MGF bounds allows us to control the generalization error via quantities measuring the classical and quantum information shared between data and hypothesis. 

Our first main result is the following generalization bound for quantum learners.

\begin{mtheorem}[Classical and quantum information-theoretic generalization bound. Informally stated; see \Cref{theorem:qmi-gen-bound-qmgf-and-cmgf}]
\label{thminf:qmi-gen-bound-qmgf-and-cmgf}
    \textit{If the classical-quantum data state $\rho$ and the loss observable satisfy \eqref{eqinf:qmgf} and \eqref{eqinf:cmgf}, then the expected generalization error of $\alg$ satisfies}
    \small
    \begin{equation}
    \label{eqinf:qmi-gen-bound-qmgf-and-cmgf}
        \pm \gen_\rho (\alg)
        \leq \psi_\mp^{\ast -1}\left(\E_{(S,W)\sim P^{\alg}}\left[ I(\mathrm{test};\mathrm{hyp})_{\sigma^\alg}\right] + \E_{S\sim P^m}\left[ \chi\left(\{P^{\mathcal{A}}_{\mathsf{W}|S}(w), \rho^{\alg}_{\mathrm{test}} (S,w)\}_{w}\right)\right]\right) + \phi_\mp^{\ast -1}(I(S; W)) \, ,
    \end{equation}
    \normalsize
    \textit{where $\psi_\mp^{\ast -1}$ and $\phi_\mp^{\ast -1}$ denote the inverses of the Legendre transforms of $\psi_\mp$ and $\phi_{\mp}$. }
\end{mtheorem}

In \Cref{eqinf:qmi-gen-bound-qmgf-and-cmgf}, the following quantities from classical and quantum information theory appear: 
$I(S;W)$ is the classical mutual information (MI) between the training data and hypothesis random variables $S$ and $W$.
$I(\mathrm{test};\mathrm{hyp})_{\sigma^\alg} = I(\mathrm{test};\mathrm{hyp})_{\sigma^\alg(S,W)}$ denotes the quantum mutual information (QMI) between test and hypothesis systems in the output state $\sigma^\alg (S,W)$ produced by the learner.
Finally, $\chi (\{P(x),\rho(x)\}_{x\in\mathsf{X}})$ denotes the Holevo information of an ensemble of quantum states, which is connected to how much information about $x\sim P$ can be extracted from $\rho(x)$. It is given by $H(\mathbb{E}_{x\sim P}[\rho(x)]) - \mathbb{E}_{x\sim P}[H(\rho(x))]$, the difference between the (von Neumann) entropy of the average state and the expected (von Neumann) entropy of a state drawn from the ensemble. 
We formally define these quantities in \Cref{section:preliminaries}.

\Cref{thminf:qmi-gen-bound-qmgf-and-cmgf} provides a theoretical guideline for designing a learner $\alg$. 
Namely, we expect better generalization performance for learners whose measurements and post-processing do not induce strong correlations between the available data set and the output hypothesis.
Naturally, we inherit a caveat from classical learning theory: Learning typically requires both good performance on the training data and good generalization.
Thus, our bounds provide an information-theoretic perspective on the bias-variance trade-off in quantum learning.
On the one hand, for good training performance, a learner has to extract information about the underlying concept from the available classical-quantum data. 
On the other hand, for good generalization, the amount of extracted/accessible classical and quantum information should be limited.

In the sub-gaussian case, the inverse Legendre transforms can be computed explicitly and the generalization error bound takes an appealingly simple form. 

\begin{mcorollary}[Informally stated; see \Cref{corollary:qmi-gen-bound-qmgf-and-cmgf-subgaussian}]
\label{corinf:qmi-gen-bound-qmgf-and-cmgf-subgaussian}
    If the classical-quantum data and the loss observable satisfy an $\alpha$-sub-gaussian \eqref{eqinf:qmgf} and a $\beta$-sub-gaussian \eqref{eqinf:cmgf} condition, then 
    \small
    \begin{equation}
        \lvert \gen_\rho (\alg)\rvert
        \leq \sqrt{2\alpha^2 \left( \E_{(S,W)\sim P^{\alg}}\left[ I(\mathrm{test};\mathrm{hyp})_{\sigma^\alg}\right]+ \E_{S\sim P^m}\left[ \chi\left(\{P^{\mathcal{A}}_{\mathsf{W}|S}(w), \rho^{\alg}_{\mathrm{test}} (S,w)\}_{w}\right)\right]\right)} + \sqrt{2 \beta^2 I(S;W)}\; .
    \end{equation}
    \normalsize
\end{mcorollary}

Next, we specialize this guarantee to the important case of independent data.
We have already assumed that $S=(Z_i)_{i=1}^m$ consists of i.i.d.~examples.
Now, we additionally assume that $\rho(S)=\bigotimes_{i=1}^m \rho_i(Z_i)$ is a tensor product of quantum data states, and that the measurements and channels performed by the learner $\alg$ also factorize. 
(In fact, if $\alg$ produces only a classical but no quantum hypothesis, we can drop this factorization requirement on $\alg$.)
Then, our states after the action of the learner also factorize as $\sigma^\alg=\bigotimes_{i=1}^m \sigma^\alg_i$ and $\tau^{\alg} = \bigotimes_{i=1}^m \tau^{\alg}_i$, with $\sigma^\alg_i=\sigma^{\alg}_i(z_i,w)$ and $\tau^{\alg}_i =\tau^{\alg}_i(z_i,w)$.
Finally, we assume that the loss observable is an average of local losses, $L=\tfrac{1}{m}\sum_{i=1}^m L_i$, with $L_i=L_i(z_i, w)$ acting non-trivially only on the $i^{\mathrm{th}}$ tensor factor.
The natural analogues of \eqref{eqinf:qmgf} and \eqref{eqinf:cmgf} become: 
\vspace{0.5cm}
\begin{itemize}
    \item[({\tiny{locQMGF}})] \textbf{Local QMGF:} For every $i\in[m]$, for every $(z_i,w)$, each local $L_i$ satisfies \eqref{eqinf:qmgf} w.r.t.~$\tau^{\alg}_i$ with bound $\psi_{\pm,i}$.

    \item[({\tiny{locCMGF}})] \textbf{Local CMGF:} For every $i\in [m]$, for every $w$, each local $\tr[L_i\tau^{\alg}_i]$ satisfies \eqref{eqinf:cmgf} w.r.t.~$P$ with bound $\phi_{\pm,i}$.
\end{itemize}
\vspace{0.5cm}
In addition to serving as natural quantum counterparts of common assumptions used to derive classical information-theoretic generalization bounds, we identify several scenarios in which (locQMGF) and (locCMGF) are satisfied in \Cref{sec:quantum-info-gen-bounds}.
As before, we speak of $\alpha_i$-sub-gaussian (locQMGF) and $\beta_i$-sub-gaussian (locCMGF) if the convex functions $\psi_{\pm,i}$ and $\phi_{\pm,i}$ are of the form $\lambda\mapsto\tfrac{\alpha_i^2 \lambda^2}{2}$ and  $\lambda\mapsto\tfrac{\beta_i^2 \lambda^2}{2}$, respectively. If the sub-gaussianity parameters are the same for all $i$, that is, if $\alpha_i = \alpha$ and $\beta_i = \beta$ for all $i$, then we simply speak of $\alpha$-sub-gaussian (locQMGF) and $\beta$-sub-gaussian (locCMGF).
In this scenario, \Cref{corinf:qmi-gen-bound-qmgf-and-cmgf-subgaussian} becomes:

\begin{mcorollary}[Informally stated; see \Cref{corollary:qmi-gen-bound-qmgf-and-cmgf-subgaussian-independent}]
\label{corinf:qmi-gen-bound-qmgf-and-cmgf-subgaussian-independent}
    If the classical-quantum data and the loss observable satisfy an $\alpha$-sub-gaussian (locQMGF) and a $\beta$-sub-gaussian (locCMGF) condition, then 
    \small
    \begin{equation}
        \lvert \gen_\rho (\alg)\rvert
        \leq  \sqrt{\frac{2\alpha^2}{m}\left( \E_{(S,W)\sim P^\alg}\left[ I(\mathrm{test};\mathrm{hyp})_{\sigma^\alg}\right]+ \E_{S\sim P^m}\left[ \chi\left(\{P^{\mathcal{A}}_{\mathsf{W}|S}(w), \rho^{\alg}_{\mathrm{test}} (S,w)\}_{w}\right)\right]\right)} + \sqrt{\frac{2 \beta^2}{m} I(S;W)}\, .
    \end{equation}
    \normalsize
\end{mcorollary}

\Cref{corinf:qmi-gen-bound-qmgf-and-cmgf-subgaussian-independent} tells us: We can control the expected generalization error by choosing the training data size $m$ to be on the order of the maximum between the classical and quantum information shared between the data and the learner's hypothesis. Conversely,  if only a limited amount of data is available, then to guarantee good generalization, we have to limit the classical and quantum information that the learner accumulates about the data accordingly.
As we explain in \Cref{section:framework-main-result}, \Cref{corinf:qmi-gen-bound-qmgf-and-cmgf-subgaussian-independent} can be extended to the case of different local loss observables, which also have different sub-gaussianity parameters $\alpha_i$ and $\beta_i$ (see \Cref{corollary:qmi-gen-bound-qmgf-and-cmgf-subgaussian-independent}), and to stable learners employing channels that approximately preserve locality (see \Cref{corollary:lipschitz-gen-bound}).

\subsubsection{Applications}\label{section:applications-informal}

Our framework and generalization bounds capture a variety of settings.
Therefore, we envision that our approach can lead to new insights by providing a novel perspective on diverse quantum learning problems.
Here, we highlight only three applications, but to fundamental problems in quantum learning.
For further examples in quantum parameter estimation, variational quantum machine learning, approximate quantum membership problems, learning quantum state preparation procedures, quantum differential privacy, and inductive quantum learning see \Cref{sec:applications}.

\paragraph{PAC learning quantum states.} 
\cite{Aaronson2007} pioneered the use of learning-theoretic perspectives for quantum information problems. The seminal contribution of this work was to formulate ``pretty good state tomography'' in a PAC learning sense and to analyze its sample complexity. 
Here, instead of aiming for a (classically described) approximation to an unknown quantum state in trace distance, one considers the relaxed task of producing a (classically described) hypothesis state that accurately approximates the expectation value on a test measurement drawn from an underlying data distribution, with high success probability. 
While full state tomography requires resources scaling exponentially with the number $n$ of qubits \cite{odonnell2016efficient, haah2017sample}, this PAC relaxation has sample complexity scaling linearly in $n$ \cite{Aaronson2007}.

In \Cref{section:pac-learning-quantum-states}, we use \Cref{corinf:qmi-gen-bound-qmgf-and-cmgf-subgaussian-independent} to reproduce this fundamental insight into learning quantum states within our framework of in-expectation learning. 
Concretely, we give a simple learning strategy achieving an in-expectation version of \cite[Theorem 1.1]{Aaronson2007} with the same dependence on the Hilbert space dimension $d$ and the approximation accuracy $\varepsilon$.
Our formulation allows us to naturally describe an end-to-end learning strategy that starts from (possibly entangled) copies of the unknown quantum state.
As part of our derivation, we extend an argument due to \cite{Xu2017} to prove that information-theoretic generalization guarantees reproduce classical in-expectation excess risk bounds for regression based on the fat-shattering dimension \cite{kearns1994efficient, bartlett1998prediction, anthony2000function}.

We highlight that our in-expectation guarantees show that for each observable seen during training, a number of copies independent of $d$ is sufficient to achieve overall reliable expectation value estimates.
In essence, there are distinct classical (``how many observables'') and quantum (``how many copies of $\rho$ per observable'') aspects to the sample complexity. Only the first is $d$-dependent.
Our perspective thus provides a natural intermediary between the ``measure $\log\log(d)$ many times'' setting of \cite[Objection 6]{Aaronson2007} and the ``measure once'' scenario of \cite[Theorem 1.3]{Aaronson2007}. This illustrates how studying in-expectation bounds can complement studying the concentration properties of the generalization error.

\paragraph{Quantum PAC learning from entangled data.} 
A central question in quantum learning theory \cite{arunachalam2017survey} is whether and when quantum access to data allows one to learn an unknown classical object (typically a function) more sample- and/or computationally efficiently than is possible purely classically.
A prominent way of modeling quantum (access to) data is via \emph{superposition examples} \cite{bshouty1998learning}, which then admit questions of PAC learning from quantum oracle access.

We propose a 
variant of quantum superposition examples: Viewing a single classical training example as a mixed state $\rho = \sum_{z} P(z)\ketbra{z}{z}$ diagonal in the computational basis, we take a \emph{purification} and consider the resulting entangled state $\ket{\phi}=\sum_{z} \sqrt{P(z)}\ket{z}_{\mathrm{test}}\otimes \ket{z}_{\mathrm{train}}$ as describing the joint system of a single quantum test and training example. Multiple copies of this bipartite state then form the overall 
data. 
The entanglement between test and training data 
is an inherently quantum analogue to a classical scenario with perfectly correlated test and training data (see \Cref{subsection:comparison-frameworks} for a more detailed discussion).

For this notion of quantum data access, we study learners that perform simple measurements followed by classical post-processing.
We show how to analyze the generalization performance of such learners purely quantumly by describing the measurement and post-processing jointly by a quantum channel acting on the training data. In particular, we demonstrate that \Cref{corinf:qmi-gen-bound-qmgf-and-cmgf-subgaussian-independent} in this case reproduces the main result of \cite{Xu2017}. 
Notably, it does so \emph{via the QMI contribution} in the upper bound, which highlights the relevance and necessity of this term.

\paragraph{Quantum state discrimination and classification.} Distinguishing between different candidate states when given copies of an unknown quantum state is a fundamental task in quantum information science \cite{bae2015quantum}. 
The optimal measurement for binary state discrimination, the case of two candidates, is well understood  \cite{helstrom1969quantum, holevo1973statistical}. 
For distinguishing between multiple states, necessary and sufficient optimality criteria are known \cite{holevo1974remarks, yuen1975optimum}, but in general do not give rise to an explicit construction for the optimal POVM.
Only in certain symmetric cases can the optimal measurement be made explicit \cite{ban1997optimum, eldar2001quantum, eldar2004optimal}, often via the pretty good (or square root) measurement \cite{hausladen1994pretty, hausladen1996classical}.
These results, however, presuppose that classical descriptions for the possible candidate states are known in advance. 

More recently, distinguishing between two a priori unknown quantum states was considered as a classification problem inspired by machine learning approaches to pattern recognition \cite{Guta2010, Sentis2019Unsupervised, rosati2022learning}. 
Here, the goal is to learn a distinguishing POVM from (labelled) copies of the unknown states.
Within our framework, we formulate a PAC version of quantum state classification (see \Cref{section:framework-main-result}). Then, our information-theoretic generalization guarantees yield bounds on the sample size sufficient to ensure that a learned POVM, which performs well on available training data, will also successfully classify previously unseen state pairs in-expectation over an underlying distribution over pairs.
These may serve as a guiding principle for avoiding overfitting in quantum state classification.
In particular, our results imply that limiting the complexity of the admissible hypothesis POVMs and thus the maximum information content of a hypothesis, for instance by imposing locality restrictions, will favorably affect the required amount of quantum data.

\subsection{Discussion and outlook}

In this work, we have established a mathematical framework for reasoning about tasks of learning from data that is part classical and part quantum. In addition to proving generalization error bounds for quantum learners in such scenarios, we have also demonstrated a variety of applications that our framework encompasses. 
Importantly, our bounds are information-theoretic in nature. Thus, they come with an intuitive interpretation and provide a perspective on quantum learning that can benefit from insights in quantum information theory.

The average-case and in-expectation generalization bounds give an insightful perspective that is complementary to worst-case analyses, which have thus far been more widespread in the literature on quantum learning.
The former illuminates certain features that are not apparent in the latter, raising the question of re-examining established results in a new or different light. 
We hope our work motivates future work on quantum learning to also consider in-expectation generalization alongside worst-case behavior.

With part of our contribution being the formulation of a novel framework, our work raises many interesting follow-up questions. In the following, we highlight some of them.

\paragraph{Average-case vs.\ worst-case.}  As is typical in PAC learning, our results address the average performance on instances drawn from an (unknown) underlying distribution.
For instance, our risk bounds for ``pretty good tomography'' \cite{Aaronson2007} hold w.r.t.~a distribution over $2$-outcome POVMs. 
In contrast, recent progress in shadow tomography \cite{aaronson2019shadow, badescu2021improved, huang2021information} and classical shadows \cite{huang2020predicting, elben2022randomized} has focused on making correct predictions in the worst-case simultaneously over many observables. 
Moreover, recent work \cite{huang2021information} has drawn attention to the notable contrast between the average-case and the worst-case when it comes to the potential for a quantum advantage in learning.
Extending our information-theoretic perspective on quantum learning to these worst-case scenarios could give us novel ways of probing this frontier.

\begin{openproblem}
    Establish a quantum information-theoretic characterization of the performance of learners for shadow tomography. 
\end{openproblem}

\paragraph{Quantum-quantum learners.} A recent spate of results \cite{aharonov2022quantum, chen2022quantum, chen2022exponential,huang2022quantum, caro2022learning,huang2022learningmanybodyhamiltonians, dutkiewicz2023advantage} has emphasized the role of quantum-enhanced experiments for learning quantum channels.
In particular, the ability to coherently and sequentially query the unknown channel on input states of our choice is an example of quantum enhancement. Can our framework be further developed to incorporate learning from such \emph{query access} to a quantum-to-quantum channel?  

\begin{openproblem}
    Establish a quantum information-theoretic characterization of the performance in learning quantum-to-quantum channels in a query input model.
\end{openproblem}

\paragraph{Optimality and technical improvements.}

One might raise the question of whether information-theoretic bounds on the expected generalization error are tight. This is already a non-trivial open question in the classical setting. In the quantum world, the problem of state discrimination is very well understood information-theoretically. We speculate that a notion of average-case state discrimination may be an approach towards understanding the optimality of our bounds. 

Finally, \cite{Xu2017, russo2019much} have led to a series of follow-up works, including techniques to tighten information-theoretic generalization bounds \cite{asadi2018chaining, bu2019tightening}, improvements relying on (evaluated) sample-wise and/or conditional mutual information \cite{steinke2020reasoning, haghifam2020sharpened, hellstroem2021fast-rate, harutyunyan2021informationtheoretic, hellstroem2022new-family, chu2023unified, hellstroem2023generalization}, and connections to optimal transport \cite{esposito2022generalisation} and convex analysis \cite{lugosi2022generalization, lugosi2023online-to-pac}. 
These results may inspire improvements to our quantum generalization bounds and potential connections to quantum optimal transport \cite{depalma2021quantumoptimal,depalma2022quantumconcentration,depalma2023wassersteinspinsystems}.

In spite of the rich structure and wealth of open problems in this area of research, simply \emph{translating} these ideas to quantum learning is fraught with pitfalls: for example, there is no unique quantum analogue to the classical notion of conditioning. 
Breakthrough progress in our quantum information-theoretic understanding of learning will require proving \emph{genuinely quantum} statements which may not have classical analogues.

\section{Technical overview}
In this section we give a brief outline of the high-level conceptual ideas involved in the development of our framework, and a taste of the techniques that we use in proving our generalization bounds.
It is also natural to wonder how our framework for describing quantum learners and our quantum information-theoretic bounds on the generalization error compare with existing work. 
We provide such a comparison to information-theoretic generalization bounds in classical learning theory, starting from \cite{Xu2017} and arriving at our framework and results via an intermediate extension, which is reminiscent of information-theoretic approaches to out-of-distribution generalization \cite[Section 9.2]{hellstroem2023generalization}. 

\subsection{Classical $\to$ quantum: motivating our framework and bounds}\label{subsection:comparison-frameworks}

First, we recall the main result of \cite{Xu2017}:  Assuming that for $Z_i\sim P$ the random variable $\ell(w, Z_i)$ is $\beta$-sub-gaussian for all $w\in\mathsf{W}$ -- a special case of our assumption (locCMGF) -- \cite[Theorem 1]{Xu2017} proved that the classical expected generalization error defined in \Cref{eq:gen-err-def-xr17} is bounded as
\begin{equation}
    \lvert \operatorname{gen}_P(\alg)\rvert
    \leq \sqrt{\frac{2\beta^2}{m} I(S;W)}\, .
\end{equation}
A simple but crucial observation underlying this bound: It amounts to a statement about \emph{decoupling} two random variables. 
Namely, we can rewrite the expected true risk as $\mathbb{E}_{W\sim P^\alg_{\mathsf{W}}}[R_P(W)] = \E_{\bar{S}\sim P^m}\E_{\bar{W}\sim P^\alg_{\mathsf{W}}}[\hat{R}_{\bar{S}}(\Bar{W})]$. This has the same form as the expected empirical risk $\mathbb{E}_{(S,W)\sim P^\alg}[\hat{R}_S(W)]$, but the training data and hypothesis random variables have been replaced by independent copies thereof.
Informally speaking, \cite[Theorem 1]{Xu2017} thus tells us that decoupling training data and hypothesis comes at a cost depending on the mutual information $I(S;W)$, see \Cref{fig:table-xr}. 

\begin{figure}
    \centering
    \includegraphics[width=0.75\textwidth]{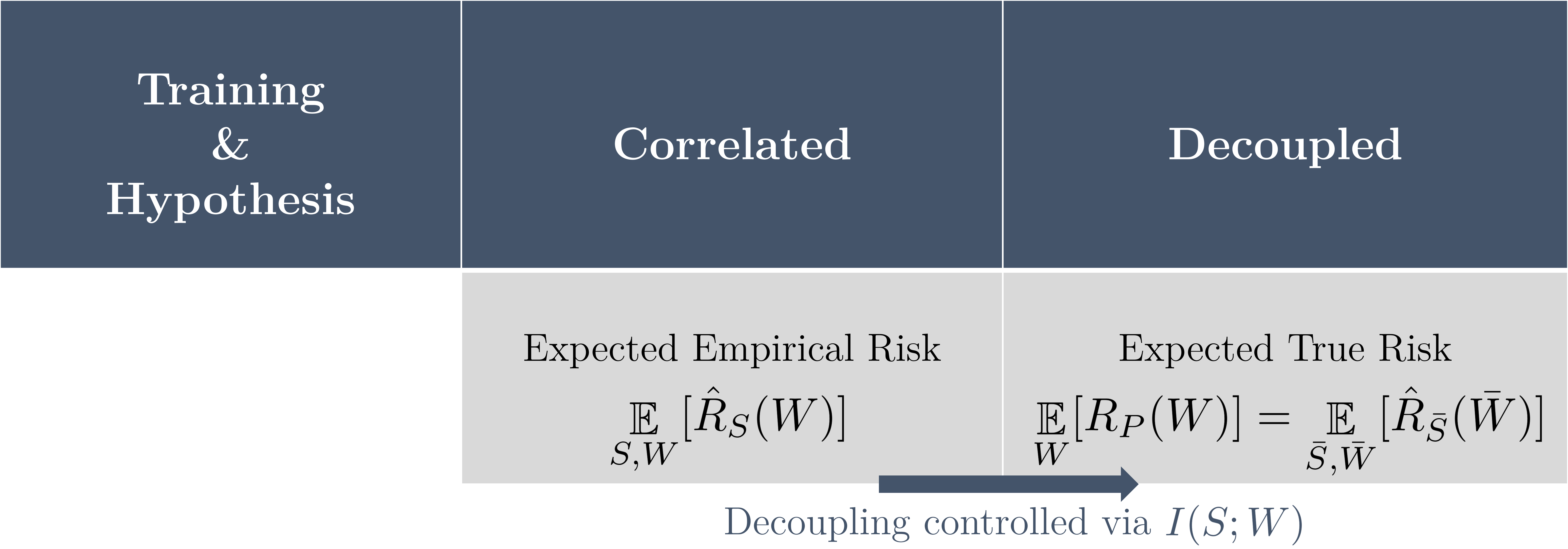}
    \caption{\textbf{{\cite{Xu2017}}'s classical framework:} The expected empirical and true risk of a classical learner differ only in whether the training data and hypothesis are correlated or not. Decoupling the two leads to a generalization bound in terms of the MI $I(S;W)$.}
    \label{fig:table-xr}
\end{figure}

Next, as an intermediate step towards our quantum framework, we introduce a variant of this result by adding test data to the classical learning-theoretic framework of \Cref{subsection:unified-framework}.
Concretely, suppose we have test data $S_{\mathrm{te}}=(Z_{\mathrm{te},i})_{i=1}^m$ and training data $S_{\mathrm{tr}}=(Z_{\mathrm{tr},i})_{i=1}^m$, where the pairs $(Z_{\mathrm{te},i}, Z_{\mathrm{tr},i})$ are drawn i.i.d.~from some probability distribution $P$ over $\mathsf{Z}\times\mathsf{Z}$. Note that while different pairs are independent, the two random variables $Z_{\mathrm{te},i}, Z_{\mathrm{tr},i}$ within any single pair may not be. 
During training, a learner $\alg$ has access to $S_{\mathrm{tr}}$ but not to $S_{\mathrm{te}}$, so its output behaviour may still be described by a conditional distribution $P(W|S_{\mathrm{tr}})$. 
However, the relevant performance measures are now taken w.r.t.~\emph{test} instead of training data. That is, we now consider the expected empirical \emph{testing} risk $\E_{S_{\mathrm{te}},S_{\mathrm{tr}},W}[\hat{R}_{S_{\mathrm{te}}}(W)]$ and the expected true \emph{testing} risk $\E_{W}[R_{P_{\mathrm{te}}}(W)]$, where $P_{\mathrm{te}}$ denotes the marginal of $P$ on the first subsystem.

Two extreme examples illustrate the utility of this setup: First, if $Z_{\mathrm{te},i}$ and $Z_{\mathrm{tr},i}$ are perfectly correlated, we recover exactly the setting considered in \cite{Xu2017}. 
In contrast, if $Z_{\mathrm{te},i}$ and $Z_{\mathrm{tr},i}$ are independent and have the same distribution, then the expected generalization error trivially vanishes. 

Also in this setting, the expected true risk can be obtained from the expected empirical risk via decoupling as before, starting with the rewriting $\E_{W}[R_{P_{\mathrm{te}}}(W)]= \E_{\bar{S}_{\mathrm{tr}},\bar{S}_{\mathrm{te}},\bar{W}}[\hat{R}_{\bar{S}_{\mathrm{te}}} (\bar{W})]$.
However, $W$ depends on $S_{\mathrm{te}}$ only through $S_{\mathrm{tr}}$, so decoupling can now be achieved in different ways: We can decouple $W$ from $S_{\mathrm{tr}}$ as before, or we can decouple $S_{\mathrm{te}}$ from $S_{\mathrm{tr}}$, or we can (unnecessarily) decouple both pairs simultaneously.
More rigorously, using \cite[Lemma 1]{Xu2017}, we can show that if $\ell(w,Z_{\mathrm{te},i})$, with $Z_{\mathrm{te},i}\sim P_{\mathrm{te}}$, is $\beta$-sub-gaussian, then the expected generalization error satisfies
\begin{equation}\label{eq:extended-xr-gen-bound}
    \left\lvert \E_{W}[R_{P_{\mathrm{te}}}(W)] - \E_{S_{\mathrm{te}},S_{\mathrm{tr}},W}[\hat{R}_{S_{\mathrm{te}}}(W)] \right\rvert
    \leq \sqrt{\frac{2\beta^2}{m} I(S_{\mathrm{te}};W)}
    \leq \sqrt{\frac{2\beta^2}{m} \min\{I(S_{\mathrm{tr}};W), I(S_{\mathrm{tr}};S_{\mathrm{te}})\}}\, ,
\end{equation}
where the last inequality follows from the data processing inequality and the chain rule.
\Cref{fig:table-classical} informally presents the different decoupling steps underlying this bound.

\begin{figure}
    \centering
    \includegraphics[width=0.85\textwidth]{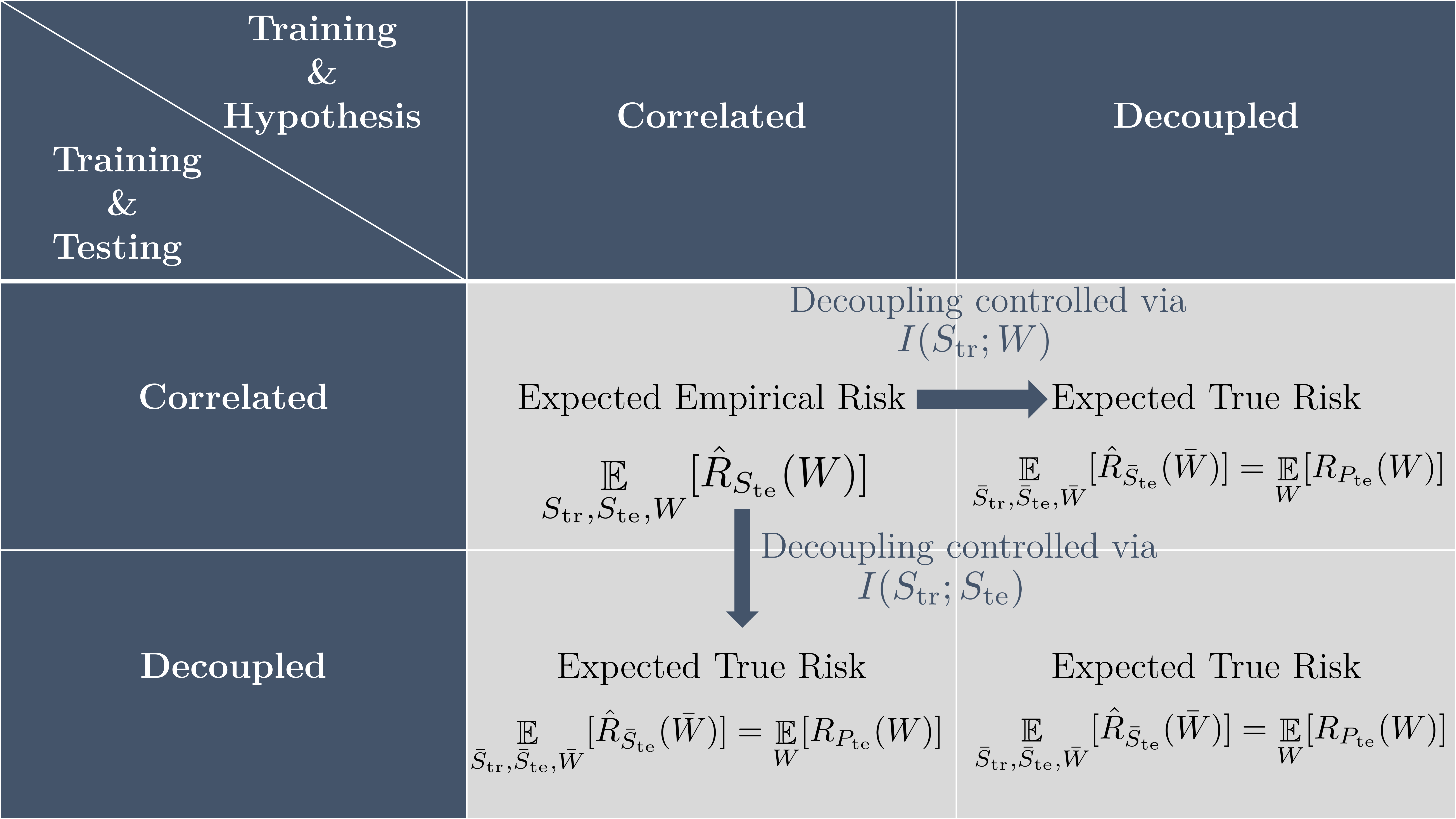}
    \caption{\textbf{Extended {\cite{Xu2017}} framework for classical learners with test data:} When taking test data into account, the expected empirical and true risk differ in whether training data and test data are correlated or decoupled and whether training data and hypothesis are correlated or decoupled. Thus, the expected generalization error can be bounded in terms of the MI quantities  $I(S_{\mathrm{tr}};S_{\mathrm{te}})$ and $I(S_{\mathrm{tr}}; W)$. Note that the resulting expected risks in three out of the four cells coincide.}
    \label{fig:table-classical}
\end{figure}

We can now transparently describe the final step towards our quantum framework. To do so, we return to the setting of \cite{Xu2017} on the classical side, assuming only training data but no test data. This is for simplicity of presentation, our bounds can be extended to the case with classical training and test data.
On the quantum side, however, we assume both a test and a training data system, which may share classical correlations or quantum entanglement.
Thus, going from the expected empirical risk $\hat{R}_{\rho}(\alg)$ to the expected true risk $R_\rho (\alg)$ now requires two decoupling steps, the first quantum -- going from a general bipartite state $\sigma^{\alg}(S,W)$ to a tensor product state $\tau^{\alg}(S,W)$ by decoupling the test and train systems before the action of the learner -- and the second classical -- going from correlated random variables $S,W$ to independent copies $\bar{S}, \bar{W}$.
Our generalization bounds make this intuition rigorous and show that the first decoupling step contributes an expected QMI plus Holevo information and the second a classical MI.
Notably, whereas a single decoupling step was already enough in the case of classical test data, our classical-quantum decoupling indeed consists of two non-trivial decoupling steps. This is reflected in our bounds having two separate terms, and in the informal depiction in \Cref{fig:table-quantum}, in which, in contrast to \Cref{fig:table-classical}, no two cells coincide.

\begin{figure}
    \centering
    \includegraphics[width=0.85\textwidth]{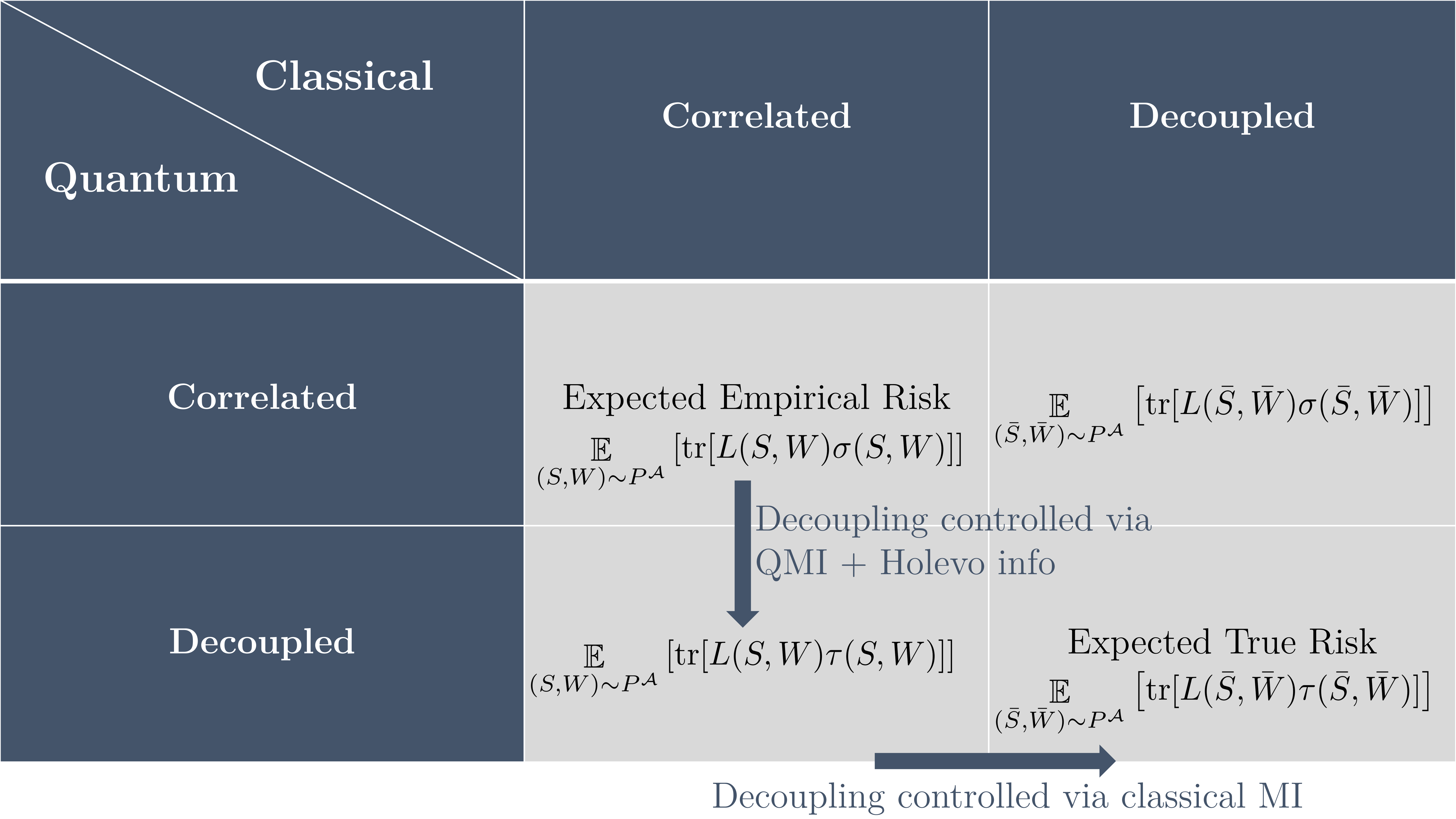}
    \caption{\textbf{Framework for quantum learners:} In our formalization of learning from CQ data, going from expected empirical to true risk requires decoupling the quantum training and test data as well as the classical hypothesis and classical training data. This leads to generalization bounds involving an average QMI plus Holevo information term and a classical MI term.}
    \label{fig:table-quantum}
\end{figure}

\subsection{Proof overview}

The proofs of classical information-theoretic generalization bounds, starting from assumptions analogous to \Cref{eqinf:cmgf}, typically proceed as follows\footnote{Starting from different sub-gaussianity assumptions, variations on this reasoning can be successful, see, e.g., \cite[Proposition 1]{bu2019tightening}.}:
First, the mutual information between data and hypothesis can be expressed as the expected relative entropy between the the distribution of the data conditioned on the hypothesis and the unconditioned distribution of the data, $I(S;W) = \mathbb{E}_{W\sim P^\alg_{\mathsf{W}}}[D(P^\alg_{\mathsf{Z}^m|W}\|P^\alg_{\mathsf{Z}^m})]$.
Next, the relative entropies are rewritten via the Donsker-Varadhan representation of the relative entropy (see, for example, \cite[Corollary 4.15]{boucheron2013concentration}), which in this case in particular implies
\begin{equation}
    D(P^\alg_{\mathsf{Z}^m|W}\|P^\alg_{\mathsf{Z}^m})
    \geq \E_{S\sim P^\alg_{\mathsf{Z}^m|W}}\left[\lambda f(W,S)\right] - \log\E_{S\sim P^\alg_{\mathsf{Z}^m}}\left[e^{\lambda f(W,S)}\right]\quad\forall \lambda\in\mathbb{R}\, ,
\end{equation}
for $f(W,S)=\tfrac{1}{m}\sum_{i=1}^m \ell(W,Z_i)$.
The second term is controlled based on assumptions on the logarithmic MGF, which in particular introduces a term $-\E_{S\sim P^\alg_{\mathsf{Z}^m}}\left[\lambda f(W,S)\right] = - \lambda \E_{S\sim P^\alg_{\mathsf{Z}^m}}[\hat{R}_S (W)]$. After an optimization over $\lambda$, one can rearrange and average over the hypothesis to obtain an information-theoretic generalization bound.

An obstacle to extending this argument to our setting with classical-quantum data is the lack of a quantum analogue of conditioning the data on the hypothesis. 
To circumvent this obstacle, we first decompose the generalization error into a classical and a quantum part. 
As highlighted in the previous subsection, this decomposition is a feature inherent to our classical-quantum setup: Even after extending the classical learning framework to include test data, it still admits a generalization bound with a single classical mutual information term, no decomposition into separate terms is needed.
In our decomposition, the classical part is a difference of two terms that differ only in whether the underlying classical data and hypothesis random variables are correlated or decoupled. Thus, it can be controlled with the classical proof strategy outlined above.
The quantum part takes the form of a classical expectation value of the difference between the quantum expectation values $\tr[L\sigma]$, of the loss observable on the state $\sigma$, and $\tr[L\tau^{\alg}]$, the expectation value on its decoupled counterpart $\tau^{\alg} = \rho_{\mathrm{test}}\otimes\sigma^{\alg}_{\mathrm{hyp}}$.
To control the error induced by this quantum decoupling, we lift the classical proof strategy to our non-commutative quantum setting, replacing Donsker-Varadhan by a combination of Petz's variational characterization of the relative entropy \cite{petz1988variational} and the Golden-Thompson inequality.
Assuming \eqref{eqinf:qmgf}, this yields the quantum relative entropy lower bound
\begin{equation}
    D(\sigma^\alg \| \tau^\alg)
    \geq \lambda\left(\tr[L\sigma^\alg] - \tr[L\tau^{\alg}]\right) - \begin{cases} \psi_+ (\lambda)\quad \textrm{if } \lambda \geq 0 \\ \psi_- (\lambda)\quad \textrm{if } \lambda <0 \end{cases}\, .
\end{equation}
Now, we can optimize over $\lambda$ and rearrange to obtain a bound on $\tr[L\sigma] - \tr[L\tau^{\alg}]$ in terms of $\mathbb{E}_{(S,W)\sim P^\alg}[D(\sigma^\alg \| \tau^\alg)]$. After showing that this expected relative entropy equals the expression $\E_{(S,W)\sim P^{\alg}}\left[ I(\mathrm{test};\mathrm{hyp})_{\sigma^\alg}\right] + \E_{S\sim P^m}\left[ \chi\left(\{P^{\mathcal{A}}_{\mathsf{W}|S}(w), \rho^{\alg}_{\mathrm{test}} (S,w)\}_{w}\right)\right]$, we combine this quantum decoupling bound with the bound on the classical part to obtain \Cref{thminf:qmi-gen-bound-qmgf-and-cmgf}.

The usefulness of classical generalization bounds depends on whether and how quickly they decay as the training data size $m$ increases. 
Typically, such a decrease is proved under an i.i.d.~assumption on the data.
To strengthen \Cref{thminf:qmi-gen-bound-qmgf-and-cmgf} for i.i.d.~quantum data, adhering to a tensor product structure, we invoke tools from quantum optimal transport \cite{depalma2021quantumoptimal,depalma2022quantumconcentration,depalma2023wassersteinspinsystems}.
On the one hand, \cite[Theorem 8.1]{depalma2023wassersteinspinsystems} (restated as \Cref{lemma:lipschitz-observables-subgaussian}) shows that Lipschitz observables have sub-gaussian QMGFs w.r.t.~any tensor product state:
\begin{equation}
    \tr\left[\exp\left(\log\left(\bigotimes_{i=1}^m \rho_i\right) + \lambda H\right)\right]
    \leq \exp\left(\frac{\lambda^2 m \lVert H\rVert_{\mathrm{Lip}}^2}{2}\right)\, .
\end{equation}
While this is weaker than bounds of the form \eqref{eqinf:qmgf} due to Golden-Thompson, we demonstrate that such a QMGF bound is still sufficient for our above proof strategy.
This then allows us to improve \Cref{thminf:qmi-gen-bound-qmgf-and-cmgf} achieve a bound that decays with $\nicefrac{1}{\sqrt{m}}$ if both quantum data and learner factorize, assuming local loss observables (\Cref{corinf:qmi-gen-bound-qmgf-and-cmgf-subgaussian-independent}).
On the other hand, the machinery of quantum Lipschitz constants allows us to go beyond quantum learners that factorize. 
In particular, it guides us to define a stability criterion for quantum learners in terms of Wasserstein-$1$ distances, a quantum version of classical replace-one stability \cite{bousquet2000algorithmic-stability, bousquet2002stability, shalev-shwartz2010stability}.
Namely, if the underlying classical data sets differ in only few data points, then the associated quantum processing channels employed by a stable learner differ only by a small amount, measured in terms of a Schatten-$1$--to--Wasserstein-$1$ norm.
Combining our newly established sub-gaussianity of Lipschitz observables w.r.t.~tensor products with the classical bounded differences concentration inequality \cite{mcdiarmid1989onthemethod}, we can then extend our generalization guarantees to stable quantum learners with a controlled increase in Wasserstein-$1$ distances (\Cref{corollary:lipschitz-gen-bound}).

\section{Preliminaries and Notation}\label{section:preliminaries}

We establish some minimal preliminaries and notation regarding quantum information and computing, and refer the reader to textbooks such as \cite{wilde2013book, Nielsen2010} for details.

We use $\mathcal{H}$ to denote a Hilbert space, and different Hilbert spaces are distinguished by subscripts. 
We denote the set of bounded operators on $\mathcal{H}$ by $\cB(\cH)$ and the trace class operators on $\mathcal{H}$ by  $\cT_1(\cH)$. 
The space of density operators (i.e., positive semidefinite trace class operators with trace $1$) on $\mathcal{H}$ is denoted by $\mathcal{S}(\cH)$. It describes the space of quantum states on $\mathcal{H}$, we will use the terms `density operator' and `quantum state' interchangeably.
Throughout the paper, we work with finite-dimensional Hilbert spaces $\mathcal{H}\cong \mathbb{C}^d$, but as we sometimes consider states with classical subsystems on a continuous alphabet, we nevertheless employ the notion of trace class operators.
When viewing multiple quantum systems with associated Hilbert spaces $\mathcal{H}_{1},\ldots,\mathcal{H}_{m}$ as a single composite quantum system, the associated Hilbert space is the tensor product $\bigotimes_{i=1}^m \mathcal{H}_{i}$.
We obtain the reduced density matrix $\rho_j$ on subsystem $j$ of a multipartite state $\rho_{1,\ldots,m}\in\mathcal{S}(\bigotimes_{i=1}^m \mathcal{H}_{i})$ via a partial trace over the remaining subsystems, $\rho_j = \tr_{1,\ldots,j-1,j+1,\ldots,m}[\rho_{1,\ldots,m}]$. 
A trace with a subscript always indicates a partial trace over the Hilbert space with the same subscript.

We next define states that have both classical and quantum subsystems:

\begin{definition}[Classical-Quantum (CQ) States]\label{definition:cq-states}
    Let $\mathsf{X}$ be a (classical) measurable space, let $\mathcal{H}$ be a Hilbert space.
    Let $P$ be a probability measure on $\mathsf{X}$ and let $\mathsf{X}\ni x\mapsto\rho(x)\in\mathcal{S}(\mathcal{H})$ be a (Borel-)measurable mapping from elements of the alphabet to quantum states.
    The associated \emph{classical-quantum (CQ)} state is given by
    \begin{equation}\label{eq:cq-state-definition}
        \E_{x\sim P}\biggl[\ketbra{x}{x}\otimes \rho(x)\biggr]\, .
    \end{equation}
\end{definition}

Here, the expectation value can be understood as a Bochner integral of a function mapping to a Banach space.
If $\mathsf{X}$ is a finite alphabet, then the expression in \Cref{eq:cq-state-definition} simplifies to 
\begin{equation}
    \E_{x\sim P}\biggl[\ketbra{x}{x}\otimes \rho(x)\biggr]
    = \sum_{x\in\mathsf{X}} P(x) \ketbra{x}{x}\otimes \rho(x)\, .
\end{equation}

Quantum information theory is a rich field and has successfully ``quantized'' a variety of notions from classical information theory. We will make use of the quantum counterpart of the classical relative entropy (also known as Kullback-Leibler divergence).

\begin{definition}[Quantum relative entropy]
    The \emph{quantum relative entropy} between a density operator $\rho\in\mathcal{S}(\cH)$ and a positive semi-definite $\sigma\in\cB(\cH)$ is given by
    \begin{equation}
        D(\rho\|\sigma)
        = \begin{cases} \tr [ \rho (\log\rho - \log\sigma)]\quad &\textrm{if } \operatorname{supp}(\rho)\subseteq\operatorname{supp}(\sigma) \\ +\infty &\textrm{else}\end{cases}.
    \end{equation}
    Here, the support of $\rho$ is, by Hermiticity, the orthogonal complement of its kernel, that is, $\operatorname{supp}(\rho) = (\operatorname{ker}(\rho))^\perp$.
\end{definition}

From the quantum relative entropy, we can now obtain the quantum mutual information. It measures how much information one subsystem in a bipartite quantum state carries about the other subsystem.

\begin{definition}[Quantum mutual information]
    Let $\rho=\rho_{A B}\in\mathcal{S}(\cH_A\otimes\cH_B)$ be a bipartite quantum state. 
    The \emph{quantum mutual information (QMI)} between subsystems $A$ and $B$ in the quantum state $\rho=\rho_{A B}$ is given by
    \begin{equation}
        I(A;B)_\rho
        = D(\rho_{AB}\|\rho_A\otimes\rho_B)
        = H(\rho_A) + H(\rho_B) - H(\rho_{AB}) \, ,
    \end{equation}
    where $H(\sigma)=-\tr[\sigma\log(\sigma)]$ denotes the von Neumann entropy.
\end{definition}

When applied to a CQ state, the QMI gives rise to the so-called Holevo information:

\begin{definition}[Holevo information]\label{definition:holevo-information}
    Let $\{P(x),\rho(x)\}_{x\in\mathsf{X}}$ be an ensemble of quantum states. The \emph{Holevo information} is given by the QMI between the classical and quantum registers in the associated CQ state:
    \begin{equation}
        \chi\left(\{P(x),\rho(x)\}_{x\in\mathsf{X}}\right)
        = I(C:Q)_{\E_{x\sim P} [\ketbra{x}{x}\otimes \rho(x) ]}\, .
    \end{equation}
\end{definition}

The Holevo information can equivalently be expressed as
\begin{align}\label{eq:holevo-information-alternate-expression}
    \chi\left(\{P(x),\rho(x)\}_{x\in\mathsf{X}}\right)
    = H\left(\E_{x\sim P} [\rho(x) ]\right) - \E_{x\sim P} [H(\rho(x) )]
    = \E_{x\sim P}\left[D\left(\rho(x)\Bigg\| \E_{\tilde{x}\sim P}[\rho(\Tilde{x})]\right)\right]\, .
\end{align}

In addition to quantum states, we need mathematical descriptions for measurements as well as for general processing of quantum systems. 
To describe measurements, we use positive operator-valued measures (POVMs).

\begin{definition}[POVMs and post-measurement states]
    The set of \emph{effect operators} $\mathcal{E}(\mathcal{H})$ on a Hilbert space $\mathcal{H}$ is given by $\mathcal{E}(\mathcal{H}) = \{E\in\mathcal{B}(\mathcal{H})~|~ E=E^\dagger~\wedge ~0\leq E\leq \mathbbm{1}_{\mathcal{H}}\}$.
    A collection $\{E_k\}_{k=1}^K\subset \mathcal{E}(\mathcal{H})$ with $\sum_{k=1}^K E_k = \mathbbm{1}_{\mathcal{H}}$ is called \emph{$K$-outcome POVM}.
    When measuring a POVM $\{E_k\}_{k=1}^K$ on a state $\rho\in\mathcal{S}(\mathcal{H})$, the probability of observing outcome $k$ is $\tr[E_k \rho]$.
    Moreover, conditioned on observing outcome $k$, the post-measurement state is given by
    \begin{equation}
        \rho_k
        \coloneqq \frac{\sqrt{E_k}\rho\sqrt{E_k}}{\tr[E_k \rho]}\, .
    \end{equation}
\end{definition}

The dynamics of quantum systems can be mathematically described by quantum channels.

\begin{definition}[Quantum channels -- Schrödinger picture]\label{definition:quantum-channels}
    A linear map $\Lambda: \mathcal{T}_1 (\mathcal{H}_{\mathrm{in}})\to\mathcal{T}_1 (\mathcal{H}_{\mathrm{out}})$ between trace class operators on Hilbert spaces $\mathcal{H}_{\mathrm{in}}$ and $\mathcal{H}_{\mathrm{out}}$ is called a \emph{quantum channel (in the Schrödinger picture)} if it is completely positive (CP) and trace-preserving (TP). Here, we call $\Lambda$ completely positive if, for any $\mathcal{H}_{\mathrm{aux}}$, $(\mathrm{id}_{\mathcal{T}_1(\mathcal{H}_{\mathrm{aux}})}\otimes \Lambda)(\rho)$ is positive-semidefinite whenever $\rho\in\mathcal{T}_1(\mathcal{H}_{\mathrm{aux}}\otimes\mathcal{H}_{\mathrm{in}})$ is positive semidefinite, and we call $\Lambda$ trace-preserving if $\tr[\Lambda(\rho)]=\tr[\rho]$ holds for all $\rho\in \mathcal{T}_1(\mathcal{H}_{\mathrm{in}})$.
\end{definition}

According to \Cref{definition:quantum-channels}, we describe a general quantum process with a CPTP map.
This is the Schrödinger picture perspective, in which we view states as evolving. 
Complementary to this, we can define the dual $\Lambda^\ast:\mathcal{B}(\mathcal{H}_{\mathrm{out}})\to \mathcal{B}(\mathcal{H}_{\mathrm{in}})$ of $\Lambda$ via the requirement $\tr[E\Lambda(\rho)] = \tr[\Lambda^\ast (E) \rho]$ $\forall \rho\in\mathcal{S}(\mathcal{H}_{\mathrm{in}}), E\in \mathcal{E}(\mathcal{H}_{\mathrm{out}})$. 
The \emph{Heisenberg picture} map $\Lambda^\ast$ is completely positive if and only if $\Lambda$ is.
Also, $\Lambda$ being TP is equivalent to $\Lambda^\ast$ being unital (U), i.e., $\Lambda^\ast(\mathbbm{1}_{\mathcal{H}_{\mathrm{out}}}) = \mathbbm{1}_{\mathcal{H}_{\mathrm{in}}}$.
Thus, quantum channels in the Heisenberg picture are linear CPU maps.

Finally, we recall two recently introduced notions from quantum optimal transport. These constitute alternatives to the trace distance between multi-qudit states and the operator norm for multi-qudit observables, respectively, and take locality into account.

\begin{definition}[Quantum Wasserstein-$1$ distance \cite{depalma2021quantumwasserstein}]
    Let $\rho,\sigma\in\mathcal{S}((\mathbb{C}^d)^{\otimes m})$ be two $m$-qudit states.
    The \emph{quantum Wasserstein-$1$ distance} $\lVert\rho - \sigma\rVert_{W_1}$ between $\rho$ and $\sigma$ is defined as
    \begin{equation}
        \lVert\rho - \sigma\rVert_{W_1}
        = \min\left\{\sum_{i=1}^m c_i~\Bigg|~
        \begin{array}{l}
             c_i\geq 0:\exists \rho^{(i)}, \sigma^{(i)}\in \mathcal{S}((\mathbb{C}^d)^{\otimes m}), 1\leq i\leq m\textrm{ s.t. }  \\
             \tr_i[\rho^{(i)}]=\tr_i[\sigma^{(i)}]\forall i ~\wedge~ \rho - \sigma = \sum_{i=1}^m c_i \left(\rho^{(i)} - \sigma^{(i)} \right)
        \end{array}
        \right\}   \, .
    \end{equation}
\end{definition}

The quantum Wasserstein-$1$ distance between quantum states induces a notion of quantum Lipschitz constant for observables via duality.

\begin{definition}[Quantum Lipschitz constant \cite{depalma2021quantumwasserstein}]
    Let $H=H^\dagger\in\mathcal{B}((\mathbb{C}^d)^{\otimes m})$ be an $m$-qudit observable.
    The \emph{quantum Lipschitz constant} $\lVert H\rVert_{\mathrm{Lip}}$ of $H$ is defined as
    \begin{align}
        \lVert H\rVert_{\mathrm{Lip}}
         &= \max\left\{\tr[HX]~|~ X=X^\dagger \in\mathcal{B}((\mathbb{C}^d)^{\otimes m}):\tr[X]=0 \wedge \lVert X\rVert_{W_1}\leq 1\right\}\\
         &= \max_{1\leq i\leq m}\max\left\{\tr[H(\rho - \sigma)]~|~\rho,\sigma\in\mathcal{S}((\mathbb{C}^d)^{\otimes m}):\tr_{i}[\rho] =\tr_{i}[\sigma]  \right\}\, .
    \end{align}
\end{definition}

\section{Framework and main result}\label{section:framework-main-result}

\subsection{Framework for learning from classical-quantum data}

We aim to provide a formalism for learning from quantum data given as a classical-quantum (CQ) state.
We suppose that the data comes in the form of a CQ state 
\begin{equation}\label{eq:data-state}
    \rho 
    = \E_{S\sim P^m}\biggl[\ketbra{S}{S}\otimes \rho (S)\biggr],
\end{equation}
with $P$ a probability measure over a classical measurable instance space $\mathsf{Z}$, and with $\rho (s)$ a density operator on a (typically composite) data Hilbert space $\cH_{\mathrm{data}}$, $\rho (s)\in \mathcal{S}(\cH_{\mathrm{data}})$, for each $s\in \mathsf{Z}^m$. 

A quantum learner $\alg$ now consists of:
\begin{enumerate}[(i)]
    \item a (possibly trivial) decomposition of the data Hilbert space into a tensor product of a test data and a training data Hilbert space, $\cH_{\mathrm{data}}=\cH_{\mathrm{test}}\otimes \cH_{\mathrm{train}}$, 
    \item a measurable hypothesis space  $\mathsf{W}$,
    \item POVMs $\{E^\alg_s (w)\}_{w\in \mathsf{W}}$ on $\cH_{\mathrm{train}}$ for each $s\in \mathsf{Z}^m$, describing the measurements used by the learner to extract classical information from the training data state and leading to probability distributions\footnote{This formulation implicitly assumes that $\mathsf{W}$ is discrete. If $\mathsf{W}$ is continuous, we can instead work with associated probability densities $q^{\alg}_s (w)=\tr[E_s (w) \tr_{\mathrm{test}}[\rho(s)]]$. Our framework and results encompass both the discrete and the continuous case. We choose discrete-case notation merely for simplicity.} $Q^{\alg}_s$ on $\mathsf{W}$ defined via $Q^{\alg}_s (w)=\tr[E_s (w) \tr_{\mathrm{test}}[\rho(s)]]$,  
    \item a quantum hypothesis Hilbert space $\cH_{\mathrm{hyp}}$,
    \item a family of quantum channels $\{\Lambda^\alg_{s,w}:\cT_1(\cH_{\mathrm{train}})\to \cT_1(\cH_{\mathrm{hyp}})\}_{(s,w)\in\mathsf{Z}^m\times \mathsf{W}}$.
\end{enumerate}
That is, the learner $\alg$ proceeds as follows: First, conditioned on the classical data $s$, $\alg$ performs the measurement described by the POVM $\{E^\alg_s (w)\}_{w\in \mathsf{W}}$ on the training data subsystem of $\rho(s)$ and classically records the measurement outcome. Second, conditioned on both the classical data $s$ and the observed measurement outcome $w$, $\mathcal{A}$ applies the quantum channel $\Lambda^\alg_{s,w}$ to the post-measurement state of the training data subsystem.
This way, the action of the learner $\alg$ on the CQ data state $\rho$ leads to the CQ output state
\begin{align}
    \sigma^{\alg} 
    &= \E_{S\sim P^m}\left[\ketbra{S}{S}\otimes \E_{W\sim Q^{\alg}_S}\left[(\id_{\mathrm{test}}\otimes\Lambda^\alg_{S,W})\left(\rho^{\alg} (S,W)\right)\otimes\ketbra{W}{W}\right]\right]\\
    &= \E_{S\sim P^m}\E_{W\sim Q^{\alg}_S}\left[\ketbra{S}{S}\otimes (\id_{\mathrm{test}}\otimes\Lambda^\alg_{S,W})\left(\rho^{\alg} (S,W)\right)\otimes\ketbra{W}{W}\right]\\
    &= \E_{S\sim P^m}\E_{W\sim Q^{\alg}_S}\left[\ketbra{S}{S}\otimes \sigma^{\alg}(S,W)\otimes\ketbra{W}{W}\right], \label{eq:learner-output-state}
\end{align}
where we have defined the post-measurement state
\begin{equation}
    \rho^{\alg} (s,w) = \frac{\left(\mathbbm{1}_{\mathrm{test}}\otimes\sqrt{E^\alg_s (w)}\right) \rho (s) \left(\mathbbm{1}_{\mathrm{test}}\otimes\sqrt{E^\alg_s (w)}\right)}{\tr[E^\alg_s (w) \rho_{\mathrm{train}} (s)]}\;,
\end{equation}
Note that $\sigma^{\alg}(s,w)\in\mathcal{S}(\cH_{\mathrm{test}}\otimes \cH_{\mathrm{hyp}})$ for every $(s,w)\in\mathsf{Z}^m\times\mathsf{W}$.
If we denote by $P^{\alg}$ the induced probability distribution over $\mathsf{Z}^m\times\mathsf{W}$ with
\begin{equation}
    P^{\alg} (s,w)
    = P^m(s)\cdot Q^{\alg}_s (w) ,
\end{equation}
denote its marginal on $\mathsf{W}$ by $P^{\alg}_{\mathsf{W}}$, and its conditional distribution for the data given the hypothesis $W$ by $P^{\alg}_{\mathrm{data}}|W$, we can interchange the order of the expectations in \Cref{eq:learner-output-state} and rewrite $\sigma^{\alg}$ as 
\begin{equation}\label{eq:learner-output-state-classical-hypothesis-expectation}
    \sigma^{\alg}
    = \E_{W\sim P^{\alg}_{\mathsf{W}}}\E_{S\sim P^{\alg}_{\mathrm{data}}|W}\left[\ketbra{S}{S}\otimes \sigma^{\alg}(S,W) \otimes\ketbra{W}{W}\right] .
\end{equation}

\begin{remark}
    In the language of \Cref{subsection:comparison-frameworks}, the setup described above assumes perfectly correlated classical training and test data. This choice was made to simplify the presentation.
    However, one may extend the framework by considering the classical part of the data to consist of (in general correlated) training and test data. Naturally, the POVMs and channels performed by the learner should only depend on the training data but not on the test data.
    This straightforward extension of our framework then also encompasses the classical extension of the \cite{Xu2017} framework with test data that we describe in \Cref{subsection:comparison-frameworks}.
    Note that including separate classical test data also enables us to describe tasks in which the training data distribution is different from the test data distribution, for example in scenarios of covariate shift, where out-of-distribution generalization becomes relevant.
    This may allow for connecting our framework to recent work on out-of-distribution generalization in learning quantum processes \cite{caro2022out-of-distribution, huang2022learning}.
\end{remark}

\begin{remark}
    Instead of describing $\alg$ in terms of POVMs and channels, we could merge these objects into a description in terms of quantum instruments (compare \cite[Chapter 5]{Heinosaari2011}). 
    We have chosen a formulation based on POVMs and channels in order to make the presentation more concrete and widely accessible.
\end{remark}

\begin{example}[Quantum state classification]\label{example:quantum-state-classification-framework}
    As an illustrative example, we consider a task of \emph{quantum state classification}, in which the quantum learner should PAC learn a two-outcome POVM that distinguishes between pairs of $d$-dimensional states weighted according to prior probabilities.
    This can be viewed as a version of the problem studied in \cite{Guta2010} but with an underlying distribution over weighted pairs of states.
    To formalize this problem, we consider a probability distribution $P_{\mathrm{weight}}\otimes P_{\mathrm{pair}}$ over the space $[0,1]\times \left(\mathcal{S}(\mathbb{C}^d)\times \mathcal{S}(\mathbb{C}^d)\right)$ of weights and pairs of states. 
    If the learner has access to $m$ labeled quantum examples generated from this distribution, the overall classical-quantum data is described by the state
    \begin{align}
        \rho
        &= \E_{\{\pi_0^{(i)}, (\sigma_0^{(i)},\sigma_1^{(i)})\}_{i=1}^m\sim (P_{\mathrm{weight}}\otimes P_{\mathrm{pair}})^m}\left[ \bigotimes_{i=1}^m  \left(\pi_0^{(i)} \ketbra{0}{0}\otimes (\sigma_0^{(i)})^{\otimes 2} + (1-\pi_0^{(i)}) \ketbra{1}{1}\otimes (\sigma_1^{(i)})^{\otimes 2}\right) \right]\\
        &= \E_{\{\pi_0^{(i)}\}_{i=1}^m\sim P_{\mathrm{weight}}^m}\left[\sum_{s=(z_1,\ldots,z_m)\in\{0,1\}^m} \left(\prod_{i=1}^m \pi_{z_i}^{(i)}\right) \ketbra{s}{s}\otimes \E_{\{(\sigma_0^{(i)},\sigma_1^{(i)})\}_{i=1}^m\sim P_{\mathrm{pair}}^m}\left[\left(\bigotimes_{i=1}^m \sigma_{z_i}^{(i)}\right)^{\otimes 2}\right]\right], 
    \end{align}
    where we used the notation $\pi_1^{(i)} = 1-\pi_0^{(i)}$.
    If we define $\mathsf{Z}=\{0,1\}$, if we let $P$ be the probability distribution on $\mathsf{Z}$ defined via
    \begin{equation}
        P(z_i)
        = \E_{\{\pi_0^{(i)}\}_{i=1}^m\sim P_{\mathrm{weight}}^m}\left[\pi_{z_i}^{(i)}\right]
        \quad \forall z_i\in\{0,1\}
        ,
    \end{equation}
    and if we further define the density operators $\rho (s)$ acting on the Hilbert space $\cH_{\mathrm{data}}=\cH_{\mathrm{test}}\otimes \cH_{\mathrm{train}}= (\mathbb{C}^d)^{\otimes m} \otimes (\mathbb{C}^d)^{\otimes m}$ as
    \begin{equation}
        \rho (s)
        = \E_{\{(\sigma_0^{(i)},\sigma_1^{(i)})\}_{i=1}^m\sim P_{\mathrm{pair}}^m}\left[\left(\bigotimes_{i=1}^m \sigma_{z_i}^{(i)}\right)^{\otimes 2}\right]\quad\forall s=(z_1,\ldots,z_m)\in\{0,1\}^m,
    \end{equation}
    then, we see that
    \begin{equation}
        \rho 
        = \E_{S\sim P^m}\biggl[\ketbra{S}{S}\otimes \rho (S)\biggr]
    \end{equation}
    in accordance with \Cref{eq:data-state}.
    
    To describe a quantum learner $\alg$ in this setting, take the quantum hypothesis space $\mathcal{H}_{\mathrm{hyp}}=\mathbb{C}$ to be trivial and consider a measurable hypothesis space $\mathsf{W}$. Here, we imagine each $w\in\mathsf{W}$ to be associated to a two-outcome qudit POVM $\{F(w), \mathds{1}_d-F(w)\}$, which describes a measurement that the learner could use for the distinguishing task.
    Now, to every $s\in\{0,1\}^m$ we associate a POVM $\{E^\alg_{s}(w)\}_{w\in\mathsf{W}}$.
    Note: As $\cT_1(\cH_{\mathrm{hyp}})=\cT_1(\mathbb{C})=\mathbb{C}$ is trivial, so is the family of quantum channels $\{\Lambda^\alg_{s,w}:\cT_1(\cH_{\mathrm{train}})\to  \cT_1(\cH_{\mathrm{hyp}})\}$ in this setting. That is, $\Lambda^\alg_{s,w}(\cdot)=\tr[(\cdot)]$ for all $s,w$.
    Thus, according to \Cref{eq:learner-output-state-classical-hypothesis-expectation}, the action of the learner $\alg$ on $\rho$ leads to the output state
    \begin{align}
        \sigma^{\alg}
        &= \E_{W\sim P^{\alg}_{\mathsf{W}}}\E_{S\sim P^{\alg}_{\mathrm{data}}|W}\left[\ketbra{S}{S}\otimes \sigma^{\alg}(S,W) \otimes\ketbra{W}{W}\right],
    \end{align}
    with the probability distribution $P^{\alg}$ on $\{0,1\}^m\times\mathsf{W}$ given by
    \begin{equation}
        P^{\alg}(s,w)
        = P^m(s)\cdot \tr\left[E_{s}(w) \rho_{\mathrm{train}}(s)\right]
    \end{equation}
    and with the post-measurement subsystem states
    \begin{equation}
        \rho_{\mathrm{test}}(s)
        = \rho_{\mathrm{train}}(s)
        = \E_{\{(\sigma_0^{(i)},\sigma_1^{(i)})\}_{i=1}^m\sim P_{\mathrm{pair}}^m}\left[\bigotimes_{i=1}^m \sigma_{z_i}^{(i)}\right]\quad\forall s=(z_1,\ldots,z_m)\in\{0,1\}^m.
    \end{equation}
    This concludes the example, we now return to the discussion of our general framework.
\end{example}

Given a learner $\alg$ and a data CQ state as described above, we now define relevant notions of risk/error. 
In classical notion theory, the most commonly used such notions are those of empirical and true risk. 
As discussed in \Cref{subsection:comparison-frameworks}, the expected empirical risk arises as an average of losses with correlated training data and hypothesis random variables. In contrast, the expected true risk can be understood as an average of losses after decoupling training data and hypothesis.
To define analogous notions for quantum learning, we go from loss functions to loss observables. Moreover, we extend the intuition that decoupling makes the difference between empirical and true risk to a decoupling on both the classical and the quantum level.

For the next three definitions, $\rho$, $\alg$, $P^{\alg}$, and $\sigma^{\alg}$ are as introduced above, and we consider a family of self-adjoint loss observables $\{L(s,w)\}_{(s,w)\in\mathsf{Z}^m\times\mathsf{W}}$ with $L(s,w)\in\mathcal{B}(\mathcal{H}_{\mathrm{test}}\otimes \mathcal{H}_{\mathrm{hyp}})$.

\begin{definition}[Expected empirical risk]\label{definition:expected-empirical-risk}
    The \emph{expected empirical risk} of $\alg$ w.r.t.~$\rho$ as measured by $\{L(s,w)\}_{(s,w)\in\mathsf{Z}^m\times\mathsf{W}}$ is defined as
    \begin{equation}
        \hat{R}_\rho(\alg)
        := \E_{(S,W)\sim P^{\alg}}\left[\tr[L(S,W) \sigma^{\alg}(S,W)] \right] .
    \end{equation}
\end{definition}

\begin{definition}[Expected true risk]\label{definition:expected-true-risk}
    The \emph{expected true risk} of $\alg$ w.r.t.~$\rho$ as measured by $\{L(s,w)\}_{(s,w)\in\mathsf{Z}^m\times\mathsf{W}}$ is defined as
    \begin{equation}
        R_\rho(\alg)
        := \E_{(\bar{S},\bar{W})\sim P^m\otimes P^{\alg}_{\mathsf{W}}}\left[ \tr\left[L(\bar{S},\bar{W}) \left(\rho_{\mathrm{test}}(\bar{S}) \otimes \sigma^{\alg}_{\mathrm{hyp}}(\bar{S},\bar{W})\right)\right] \right] .
    \end{equation}
    As before, here we let $\bar{S},\bar{W}$ denote independent copies of $S$ and $W$.
\end{definition}

\begin{definition}[Expected generalization error]\label{definition:expected-generalization-error}
    The \emph{expected generalization error} of $\alg$ w.r.t.~$\rho$ as measured by $\{L(s,w)\}_{(s,w)\in\mathsf{Z}^m\times\mathsf{W}}$ is defined to be
    \begin{equation}
        \gen_\rho (\alg) := R_\rho(\alg)-\hat{R}_\rho(\alg) ,
    \end{equation}
    the difference between the expected true and empirical risks of $\alg$ w.r.t.~$\rho$ as measured by $\{L(s,w)\}_{(s,w)\in\mathsf{Z}^m\times\mathsf{W}}$.
\end{definition}

\begin{remark}\label{remark:channel-absorbed-in-loss}
    Note that there is some freedom in our definition of channels $\Lambda_{s,w}^\alg$ and loss observables $L(s,w)$ because of the duality of Schrödinger and Heisenberg pictures. 
    Concretely, if $\Lambda_{s,w}^\alg= \Lambda_{s,w}^{''\alg}\circ \Lambda_{s,w}^{'\alg}$ and if we define $L^{'}(s,w) = (\operatorname{id}_{\mathrm{test}}\otimes \Lambda_{s,w}^{''\alg})^\ast (L(s,w))$, then the expected empirical and true risks obtained by considering $\Lambda_{s,w}^{'\alg}$ and $L^{'}(s,w)$ coincide with those originally obtained from $\Lambda_{s,w}^\alg$ and $L(s,w)$.
\end{remark}

Next, we illustrate these definitions in two concrete examples. First, we demonstrate how they recover the classical case, before continuing the discussion of our state classification application.

\begin{example}\label{example:recovering-classical-risks}
    Starting from \Cref{definition:expected-empirical-risk,definition:expected-true-risk,definition:expected-generalization-error}, we can reproduce the corresponding classical notions of expected empirical risk, expected true risk, and expected generalization error in (at least) the following two ways:
    On the one hand, if we assume all involved quantum systems to be trivial (i.e., $\mathcal{H}_{\mathrm{data}} = \mathcal{H}_{\mathrm{hyp}} = \mathbb{C}$), then the loss observables are real scalars. Interpreting these as classical loss functions, we recover the notions familiar from the classical case.
    On the other hand, even when (some of) the involved quantum systems are non-trivial, if we consider loss observables $L(s,w) = \ell(s,w)\cdot\mathbbm{1}_{\mathrm{test},\mathrm{hyp}}$ given by multiples of the identity, with classical loss function values $\ell(s,w)=\tfrac{1}{m}\sum_{i=1}^m \ell(w,z_i)$, then the trace-normalization of $\sigma^{\alg}(s,w)$ and $\rho_{\mathrm{test}}(s) \otimes \sigma^{\alg}_{\mathrm{hyp}}(s,w)$ ensures that we again obtain the same quantities as in the classical case. As we will see later, 
    our results for this latter setting indeed reproduce the classical bounds of \cite{Xu2017}.
\end{example}

\begin{example}[Quantum state classification -- \Cref{example:quantum-state-classification-framework} continued]\label{example:quantum-state-classification-risks}
    To obtain reasonable notions of risk in the quantum state classification setting of \Cref{example:quantum-state-classification-framework}, we can take the loss observables for $s=(z_1,\ldots,z_m)\in\{0,1\}^m$ and $w\in\mathsf{W}$ to be
    \begin{equation}
        L(s,w)
        = \frac{1}{m}\sum_{i=1}^m \mathbbm{1}_d^{\otimes (i-1)}\otimes\left((1-z_i)(\mathbbm{1}_d-F(w)) + z_i F(w)\right)\otimes \mathbbm{1}_d^{\otimes (m-i)} .
    \end{equation}
    With this choice, the expected empirical risk from \Cref{definition:expected-empirical-risk} becomes
    \small
    \begin{align}
        \hat{R}_{\rho} (\alg)
        &=\E_{(S,W)\sim P^{\alg}}\left[\tr[L(S,W) \sigma^{\alg}(S,W)] \right]\\
        &= \E_{(S,W)\sim P^{\alg}}\left[\frac{1}{m}\sum_{i=1}^m \tr\left[\left((1-Z_i)(\mathbbm{1}_d-F(W)) + Z_i F(W)\right) \E_{(\sigma_0^{(i)},\sigma_1^{(i)})\sim P_{\mathrm{pair}}}\left[ \sigma_{Z_i}^{(i)}\right] \right]\right]\\
        &= \E_{(\{\pi_0^{(i)}, (\sigma_0^{(i)},\sigma_1^{(i)})\}_{i=1}^m, W) \sim P^{\alg}}\left[\frac{1}{m}\sum_{i=1}^m \left(\pi_0^{(i)} \tr[(\mathbbm{1}_d-F(W))\sigma_0^{(i)}] + \pi_1^{(i)} \tr[F(W)\sigma_1^{(i)}]\right)\right], \label{eq:state-discrimination-expected-empirical-risk}
    \end{align}
    \normalsize
    where the last step uses the definition of $P$ from \Cref{example:quantum-state-classification-framework} and, in a slight abuse of notation, uses $P^\alg$ to also denote the induced joint distribution over weighted pairs of states and hypotheses.
    This induced distribution can explicitly be written as
    \begin{equation}
        \begin{split}
            &P^\alg \left(\{\pi_0^{(i)}, (\sigma_0^{(i)},\sigma_1^{(i)})\}_{i=1}^m, w\right)\\
            &= \left(\prod_{i=1}^m P_{\mathrm{weight}} (\pi_0^{(i)})\right) \cdot \left(\prod_{i=1}^m P_{\mathrm{pair}} (\sigma_0^{(i)},\sigma_1^{(i)})\right)\cdot \sum_{s\in\{0,1\}^m} \left(\prod_{i=1}^m  \pi_{z_i}^{(i)}\right)\cdot\tr\left[E_s^\alg (w) \rho_{\mathrm{train}}(s)\right] .
        \end{split}
    \end{equation}
    Thus, \Cref{eq:state-discrimination-expected-empirical-risk} is exactly the expected probability that the quantum learner misclassifies an unknown state, where the average is over the joint distribution of training data and hypothesis. This is the natural notion of expected empirical risk in this scenario. 
    
    The expected true risk according to \Cref{definition:expected-true-risk} is
    \small
    \begin{align}
        R_\rho (\alg)
        &=\E_{(\bar{S},\bar{W})\sim P^m\otimes P^{\alg}_{\mathsf{W}}}\left[ \tr\left[L(\bar{S},\bar{W}) \left(\rho_{\mathrm{test}}(\bar{S}) \otimes \sigma^{\alg}_{\mathrm{hyp}}(\bar{S},\bar{W})\right)\right] \right] \\
        &= \E_{(\{\bar{\pi}_0^{(i)}, (\bar{\sigma}_0^{(i)},\bar{\sigma}_1^{(i)})\}_{i=1}^m, \bar{W}) \sim (P_{\mathrm{weight}}\otimes P_{\mathrm{pair}})^m\otimes P^\alg_{\mathsf{W}} }\left[\frac{1}{m}\sum_{i=1}^m \left(\bar{\pi}_0^{(i)} \tr[(\mathbbm{1}_d-F(\bar{W}))\bar{\sigma}_0^{(i)}] + \bar{\pi}_1^{(i)} \tr[F(\bar{W})\bar{\sigma}_1^{(i)}]\right)\right]\\
        &= \E_{(\bar{\pi}_0, (\bar{\sigma}_0,\bar{\sigma}_1), \bar{W}) \sim P_{\mathrm{weight}}\otimes P_{\mathrm{pair}}\otimes P^\alg_{\mathsf{W}}}\left[\bar{\pi}_0 \tr\left[(\mathbbm{1}_d-F(\bar{W}))\bar{\sigma}_0] + \bar{\pi}_1 \tr[F(\bar{W})\bar{\sigma}_1\right]\right] ,
    \end{align}
    \normalsize
    where the random variables with bars again denote independent copies of the respective unbarred random variables.
    Thus, the expected true risk is exactly the expected probability that the quantum learner misclassifies an unknown state from a new, independently drawn weighted pair, a natural choice of expected true risk in this setting.
    Hence, the expected generalization error from \Cref{definition:expected-generalization-error} indeed reproduces the natural expression, namely the difference between the expected misclassification probability on a randomly drawn new data point and the expected average misclassification probability over the training data.
    This concludes the discussion of risks for our state classification tasks.
\end{example}

In \Cref{sec:applications}, we demonstrate that the general notions of risks introduced in \Cref{definition:expected-empirical-risk,definition:expected-true-risk,definition:expected-generalization-error} reproduce further natural performance measures for suitably chosen loss observable $L$ in learning scenarios such as PAC learning quantum states, learning classical functions from entangled quantum data, and quantum parameter estimation, among others.

\subsection{Generalization bounds for learning from classical-quantum data}\label{sec:quantum-info-gen-bounds}

The remainder of this section is concerned with proving that classical and quantum moment generating function assumptions lead to expected generalization error bounds in terms of quantities measuring the classical and quantum information between the data and the output of the learner.
This lifts the following intuition from classical to quantum learning: Learners generalize well (in distribution) if they produce hypotheses that do not depend too strongly on the specific dataset that they were trained on.

\begin{table}
  \begin{center}
      \renewcommand{\arraystretch}{1.4}
      \begin{tabular}{|l|c|}
        \toprule
        Object & Notation \\
        \midrule
        Probability density of classical data & $P$ \\
        Input data CQ state & $\rho= \E_{S\sim P^m}[\ketbra{S}{S}\otimes \rho (S)]$\\
        POVMs associated with learner $\alg$ & $E^{\alg}_s(w)$ \\
        Joint distribution induced by learner $\alg$ & $P^\alg$\\
        CPTP maps associated with learner $\alg$ & $\Lambda^{\alg}_{s,w}$ \\
        Learner output & $\sigma^\alg = \E_{(S,W)\sim P^\alg}\left[\ketbra{S}{S}\otimes \sigma^{\alg}(S,W)\otimes\ketbra{W}{W}\right]$\\
        Loss observables & $L(s,w)$ \\
        Quantum mutual information & $I(\cdot;\cdot)_{\bullet}$\\
        Holevo information & $\chi(\{\cdot,\cdot\})$\\
        Quantum log-MGF bound & $\psi_{\pm}$ \\
        Classical log-MGF bound & $\phi_{\pm}$ \\
        \bottomrule
      \end{tabular}
  \end{center}
  \caption{Notation for the various mathematical objects appearing in this section.}
  \label{table:notation-details}
\end{table}

\Cref{table:notation-details} compiles relevant notation for the formulation of our results.
Before stating them, we recall the following definition from convex analysis and a lemma about the quantum relative entropy:

\begin{definition}[Fenchel-Legendre dual]
    Let $\psi:\mathbb{R}\to\mathbb{R}$ be lower-semi-continuous and convex.
    The Fenchel-Legendre dual $\psi^\ast:\mathbb{R}\to\mathbb{R}$ is defined as
    \begin{align}
        \psi^\ast (t)
        = \sup_{\lambda\in\mathbb{R}} \{\lambda t - \psi(\lambda)\}.
    \end{align}
\end{definition}

\begin{lemma}[Petz's variational characterization of the quantum relative entropy~{\cite{petz1988variational}}]\label{lemma:variational-quantum-relative-entropy}    
    Let $\sigma_1,\sigma_2\in\mathcal{S}(\cH)$ be two quantum states. Then, the relative entropy between $\sigma_1$ and $\sigma_2$ can be rewritten as follows:
    \begin{equation}
        D(\sigma_1\|\sigma_2)
        =\sup_{H=H^\dagger\in \cB(\cH)} \{ \tr[ \sigma_1 H] - \log \tr[\exp\left(\log(\sigma_2) + H\right)]\}.
    \end{equation}
\end{lemma}

We can now state and prove our main result:

\begin{theorem}[Expected generalization error bound via quantum mutual information]\label{theorem:qmi-gen-bound-qmgf-and-cmgf}
    Assume that, for every $(s,w)\in\mathsf{Z}^m\times\mathsf{W}$,
    \small
    \begin{equation}\label{eq:quantum-moment-generating-function-bound-assumption}
        \log \tr\left[(\rho_{\mathrm{test}}(s) \otimes \sigma^{\alg}_{\mathrm{hyp}}(s,w)) e^{\lambda \left( L(s,w) - \tr[L(s,w) (\rho_{\mathrm{test}}(s) \otimes \sigma^{\alg}_{\mathrm{hyp}}(s,w))]\mathbbm{1}_{\mathrm{test},\mathrm{hyp}}\right)}\right]
        \leq \begin{cases} \psi_+ (\lambda)\quad \textrm{if } \lambda \geq 0 \\ \psi_- (\lambda)\quad \textrm{if } \lambda <0 \end{cases} ,\tag{\texttt{QMGF}}
    \end{equation}
    \normalsize
    where $\psi_+, \psi_-:\mathbb{R}\to\mathbb{R}$ are convex, differentiable at $0$, and satisfy $\psi_\pm (0)=\psi'_\pm (0)=0$.
    Moreover, assume that, for every $w\in \mathsf{W}$,
    \small
    \begin{equation}\label{eq:classical-moment-generating-function-bound-assumption}
        \log \E_{S\sim P^m}\left[e^{\lambda (\tr[L(S,w) (\rho_{\mathrm{test}}(S) \otimes \sigma^{\alg}_{\mathrm{hyp}}(S,w))] - \E_{\tilde{S}\sim P^m}\left[\tr[L(\tilde{S},w)(\rho_{\mathrm{test}}(\tilde{S}) \otimes \sigma^{\alg}_{\mathrm{hyp}}(\tilde{S},w))]\right])}\right]
        \leq \begin{cases} \phi_+ (\lambda)\quad \textrm{if } \lambda \geq 0 \\ \phi_- (\lambda)\quad \textrm{if } \lambda <0 \end{cases}, \tag{\texttt{CMGF}}
    \end{equation}
    \normalsize
    where $\phi_+, \phi_-:\mathbb{R}\to\mathbb{R}$ are convex, differentiable at $0$, and satisfy $\phi_\pm (0)=\phi'_\pm (0)=0$. 
    Then,
    \begin{equation}\label{eq:qmi-gen-bound-qmgf-and-cmgf}
        \begin{split}
            \pm \gen_\rho (\alg)
            &\leq \psi_\mp^{\ast -1}\left(\E_{(S,W)\sim P^{\mathcal{A}}}[I(\mathrm{test};\mathrm{hyp})_{\sigma(S,W)}] + \E_{S\sim P^m}\left[ \chi\left(\{P^{\mathcal{A}}_{\mathsf{W}|S}(w), \rho^{\alg}_{\mathrm{test}} (S,w)\}_{w\in\mathsf{W}}\right)\right]\right) \\
            &\hphantom{\leq }~~+ \phi_\mp^{\ast -1}\left(I(S; W)\right) \, .
        \end{split}
    \end{equation}
\end{theorem}

Our proof is inspired by the reasoning used in the classical case, for instance in \cite{Xu2017, raginsky2019information, bu2020tightening}, but differs from it in three non-trivial ways.
First, one central ingredient in the classical argument, namely the Donsker-Varadhan representation of the classical relative entropy, has to be replaced by its quantum counterpart, \Cref{lemma:variational-quantum-relative-entropy}. 
Second, to deal with potential complications about matrix exponentials arising from non-commutativity, we rely on the Golden-Thompson inequality.
Finally, while there is only one decoupling step in the classical proof, our scenario requires both a classical and a quantum decoupling. Thus, our analysis uses an additional decomposition of the expected generalization error compared to the classical case.
    
\begin{proof}
    When combined with the Golden-Thompson inequality \cite[see, e.g.,][Section IX.3]{bhatia1997matrix}, which tells us that $\tr[e^{A+B}]\leq\tr[e^A e^B]$ for Hermitian matrices $A$ and $B$, \Cref{lemma:variational-quantum-relative-entropy} implies, for every $(s,w)\in\mathsf{Z}^m\times\mathsf{W}$ and for all $\lambda\in\mathbb{R}$,
    \begin{align}
        &D(\sigma^{\alg}(s,w)\|\rho_{\mathrm{test}}(s) \otimes \sigma^{\alg}_{\mathrm{hyp}}(s,w))\\
        &\geq \lambda \tr[L(s,w) \sigma^{\alg}(s,w)] - \log\tr\left[\exp\left(\log(\rho_{\mathrm{test}}(s) \otimes \sigma^{\alg}_{\mathrm{hyp}}(s,w))+\lambda L(s,w)\right)\right]\\
        &\geq \lambda \tr[L(s,w) \sigma^{\alg}(s,w)] - \log \tr\left[(\rho_{\mathrm{test}}(s) \otimes \sigma^{\alg}_{\mathrm{hyp}}(s,w)) \exp\left(\lambda L(s,w)\right)\right]\\
        \begin{split}
            &= \lambda\left(\tr[L(s,w) \sigma^{\alg}(s,w)] - \tr[L(s,w) (\rho_{\mathrm{test}}(s) \otimes \sigma^{\alg}_{\mathrm{hyp}}(s,w))]\right)\\
            &\hphantom{=}~ - \log \tr\left[(\rho_{\mathrm{test}}(s) \otimes \sigma^{\alg}_{\mathrm{hyp}}(s,w)) e^{\lambda \left( L(s,w) - \tr[L(s,w) (\rho_{\mathrm{test}}(s) \otimes \sigma^{\alg}_{\mathrm{hyp}}(s,w))]\mathbbm{1}_{\mathrm{test},\mathrm{hyp}}\right)}\right]
        \end{split}
        \\
        &\geq \lambda\left(\tr[L(s,w) \sigma^{\alg}(s,w)] - \tr[L(s,w) (\rho_{\mathrm{test}}(s) \otimes \sigma^{\alg}_{\mathrm{hyp}}(s,w))]\right) - \begin{cases} \psi_+ (\lambda)\quad \textrm{if } \lambda \geq 0 \\ \psi_- (\lambda)\quad \textrm{if } \lambda <0 \end{cases} .
    \end{align}
    Here, the first step uses \Cref{lemma:variational-quantum-relative-entropy}, the second is due to the Golden-Thompson inequality, the third step is a simple rewriting, and the final step consists in plugging in \Cref{eq:quantum-moment-generating-function-bound-assumption}.
    
    We can now rearrange this inequality and optimize over $\lambda$ to obtain:
    \small
    \begin{align}
        \tr[L(s,w) \sigma^{\alg}(s,w)] - \tr[L(s,w) (\rho_{\mathrm{test}}(s) \otimes \sigma^{\alg}_{\mathrm{hyp}}(s,w))]
        &\leq \inf_{\lambda >0} \frac{D(\sigma^{\alg}(s,w)\|\rho_{\mathrm{test}}(s) \otimes \sigma^{\alg}_{\mathrm{hyp}}(s,w)) + \psi_+ (\lambda)}{\lambda},\\
        -\left(\tr[L(s,w) \sigma^{\alg}(s,w)] - \tr[L(s,w) (\rho_{\mathrm{test}}(s) \otimes \sigma^{\alg}_{\mathrm{hyp}}(s,w))]\right)
        &\leq \inf_{\lambda <0} \frac{D(\sigma^{\alg}(s,w)\|\rho_{\mathrm{test}}(s) \otimes \sigma^{\alg}_{\mathrm{hyp}}(s,w)) + \psi_- (\lambda)}{\lambda}.
    \end{align}
    \normalsize
    Using {\cite[Lemma 2.4]{boucheron2013concentration}}, we can rewrite the infima in terms of the generalized inverses $\psi_\pm^{\ast -1} (s) = \inf\{t\geq 0~|~ \psi_\pm^{\ast}(t) > s\}$ of the Fenchel-Legendre duals of $\psi_\pm$ to obtain
    \small
    \begin{align}
        \tr[L(s,w) \sigma^{\alg}(s,w)] - \tr[L(s,w) (\rho_{\mathrm{test}}(s) \otimes \sigma^{\alg}_{\mathrm{hyp}}(s,w))]
        &\leq \psi_+^{\ast -1}(D(\sigma^{\alg}(s,w)\|\rho_{\mathrm{test}}(s) \otimes \sigma^{\alg}_{\mathrm{hyp}}(s,w))),\label{eq:intermediate-fenchel-legendre-bound-plus-v1}\\
        -\left(\tr[L(s,w) \sigma^{\alg}(s,w)] - \tr[L(s,w) (\rho_{\mathrm{test}}(s) \otimes \sigma^{\alg}_{\mathrm{hyp}}(s,w))]\right)
        &\leq \psi_-^{\ast -1}(D(\sigma^{\alg}(s,w)\|\rho_{\mathrm{test}}(s) \otimes \sigma^{\alg}_{\mathrm{hyp}}(s,w))).\label{eq:intermediate-fenchel-legendre-bound-minus-v1}
    \end{align}
    \normalsize    
    Next, we rewrite the expression of interest as
    \small
    \begin{align}
        &\pm \gen_\rho (\alg)\\
        &= \pm \E_{W\sim P^{\alg}_{\mathsf{W}}}\left[\E_{\bar{S}\sim P^m}\left[ \tr[L(\bar{S},W) \left(\rho_{\mathrm{test}}(\bar{S}) \otimes \sigma^{\alg}_{\mathrm{hyp}}(\bar{S},W)\right)] \right] - \E_{S\sim P^{\alg}_{\mathrm{data}}|W}\left[\tr[L(S,W) \sigma^{\alg}(S,W)] \right]\right]\\
        &= \E_{W\sim P^{\alg}_{\mathsf{W}}}\E_{S\sim P^{\alg}_{\mathrm{data}}|W}\left[ \pm\left(\tr[L(S,W) \left(\rho_{\mathrm{test}}(S) \otimes \sigma^{\alg}_{\mathrm{hyp}}(S,W)\right)] -  \tr[L(S,W) \sigma^{\alg}(S,W)] \right)\right]\\
        \begin{split}
            &\hphantom{=}~ + \E_{W\sim P^{\alg}_{\mathsf{W}}}\Bigg[\pm\Bigg(\E_{\bar{S}\sim  P^m}\left[\tr[L(\bar{S},W) \left(\rho_{\mathrm{test}}(\bar{S}) \otimes \sigma^{\alg}_{\mathrm{hyp}}(\bar{S},W)\right)]\right] \\
            &\hphantom{=+ \E_{W\sim P^{\alg}_{\mathsf{W}}}\Bigg[\pm\Bigg(}~~ - \E_{S\sim P^{\alg}_{\mathrm{data}}|W}\left[\tr[L(S,W) \left(\rho_{\mathrm{test}}(S) \otimes \sigma^{\alg}_{\mathrm{hyp}}(S,W)\right)]\right]\Bigg) \Bigg] 
        \end{split}
    \end{align}
    \normalsize
    For the first summand, we can use \Cref{eq:intermediate-fenchel-legendre-bound-plus-v1,eq:intermediate-fenchel-legendre-bound-minus-v1} to obtain:
    \begin{align}
        &\E_{W\sim P^{\alg}_{\mathsf{W}}}\E_{S\sim P^{\alg}_{\mathrm{data}}|W}\left[ \pm\left(\tr[L(S,W) \left(\rho_{\mathrm{test}}(S) \otimes \sigma^{\alg}_{\mathrm{hyp}}(S,W)\right)] -  \tr[L(S,W) \sigma^{\alg}(S,W)] \right)\right]\\
        &\leq \E_{W\sim P^{\alg}_{\mathsf{W}}}\E_{S\sim P^{\alg}_{\mathrm{data}}|W}\left[ \psi_\mp^{\ast -1}\left(D(\sigma^{\alg}(S,W)\|\rho_{\mathrm{test}}(S) \otimes \sigma^{\alg}_{\mathrm{hyp}}(S,W))\right)\right] .
    \end{align}
    For the second summand, thanks to \Cref{eq:classical-moment-generating-function-bound-assumption}, we can apply \cite[Theorem 2]{jiao2017dependence} or \cite[Theorem 1]{bu2020tightening} (see also \cite[p.~22]{raginsky2019information} for a pedagogical presentation) to the classical random variable $\tr[L(S,W) \left(\rho_{\mathrm{test}}(S) \otimes \sigma^{\alg}_{\mathrm{hyp}}(S,W)\right) ]$ and obtain:
    \begin{align}
        \begin{split}
            &\E_{W\sim P^{\alg}_{\mathsf{W}}}\Bigg[\pm\Bigg(\E_{\bar{S}\sim  P^m}\left[\tr[L(\bar{S},W) \left(\rho_{\mathrm{test}}(\bar{S}) \otimes \sigma^{\alg}_{\mathrm{hyp}}(\bar{S},W)\right)]\right] \\
            &\hphantom{=+ \E_{W\sim P^{\alg}_{\mathsf{W}}}\Bigg[\pm\Bigg(}~~ - \E_{S\sim P^{\alg}_{\mathrm{data}}|W}\left[\tr[L(S,W)\left(\rho_{\mathrm{test}}(S) \otimes \sigma^{\alg}_{\mathrm{hyp}}(S,W)\right)]\right]\Bigg) \Bigg] 
        \end{split}
        \\
        &\leq \phi_\mp^{\ast -1}(I(S;W)) \, .
    \end{align}
    Thus, we have shown the inequalities
    \begin{equation}
        \pm \gen_\rho (\alg)
        \leq \E_{(S,W)\sim P^{\alg}}\left[ \psi_\mp^{\ast -1}\left(D(\sigma^{\alg}(S,W)\|\rho_{\mathrm{test}}(S) \otimes \sigma^{\alg}_{\mathrm{hyp}}(S,W))\right)\right] + \phi_\mp^{\ast -1}\left(I(S; W)\right) \, .
    \end{equation}
    As the $\psi_\mp^{\ast -1}$ are concave (since $\psi_\mp^\ast$ are convex), we can pull the expectation value inside the $\psi_\mp^{\ast -1}$ without making the right-hand side smaller, by Jensen's inequality. 
    Then, it remains to observe that 
    \begin{align}
        &\E_{(S,W)\sim P^{\mathcal{A}}}[D(\sigma^\alg(S,W)\|\rho_{\mathrm{test}}(S)\otimes \sigma^\alg_{\mathrm{hyp}}(s,W))]\\
        &= \E_{(S,W)\sim P^{\mathcal{A}}}\left[-H(\sigma^\alg(S,W)) + H(\sigma^\alg_{\mathrm{hyp}}(S,W)) - \tr\left[\sigma^\alg_{\mathrm{test}}(S,W)\log\left(\rho_{\mathrm{test}}(S)\right)\right]\right]\\
        &= \E_{(S,W)\sim P^{\mathcal{A}}}\left[I(\mathrm{test};\mathrm{hyp})_{\sigma^\alg(S,W)} - H(\sigma^\alg_{\mathrm{test}}(S,W)) - \tr\left[\sigma^\alg_{\mathrm{test}}(S,W)\log\left(\rho_{\mathrm{test}}(S)\right)\right]\right]\\
        &= \E_{(S,W)\sim P^{\mathcal{A}}}\left[I(\mathrm{test};\mathrm{hyp})_{\sigma^\alg(S,W)} - H(\rho^\alg_{\mathrm{test}}(S,W)) - \tr\left[\rho^\alg_{\mathrm{test}}(S,W)\log\left(\rho_{\mathrm{test}}(S)\right)\right]\right]\\
        &= \E_{(S,W)\sim P^{\mathcal{A}}}\left[I(\mathrm{test};\mathrm{hyp})_{\sigma(S,W)}\right] - \E_{(S,W)\sim P^{\mathcal{A}}}\left[H(\rho^\alg_{\mathrm{test}}(S,W))\right] + \E_{S\sim P^m}\left[H(\rho_{\mathrm{test}}(S))\right]\\
        &= \E_{(S,W)\sim P^{\mathcal{A}}}\left[I(\mathrm{test};\mathrm{hyp})_{\sigma(S,W)}\right] +    \E_{S\sim P^m}\left[\chi\left(\left\{P^\alg_{\mathsf{W}|S}(w),\rho^\alg_{\mathrm{test}}(S,w)\right\}_{w\in\mathsf{W}}\right)\right]\, .
    \end{align}
    Here, the third equality used that $\sigma^\alg_{\mathrm{test}}(s,w)=\rho^\alg_{\mathrm{test}}(s,w)$, because $\sigma^\alg(s,w)$ and $\rho^\alg(s,w)$ differ only by a CPTP map applied on the $\mathrm{train}$ subsystem. The fourth and fifth equalities used that $\E_{W\sim P^{\alg}_{\mathsf{W}|S}}[\rho^\alg_{\mathrm{test}}(S,W)]=\rho_{\mathrm{test}}(S)$. This holds because the state $\E_{W\sim P^{\alg}_{\mathsf{W}|S}}[\rho^\alg(S,W)]$ is obtained from $\rho(S)$ by applying the CPTP map $\operatorname{id}_{\mathrm{test}}\otimes \left(\sum_{w} \sqrt{E_S^\alg(w)} (\cdot )\sqrt{E_S^\alg(w)}\right)$, which acts non-trivially only on the training data register and thus leaves the test data marginal invariant.
    The fifth step also used \Cref{eq:holevo-information-alternate-expression}.
    Thus, after using Jensen to pull the expectation value inside $\psi_\mp^{\ast -1}$ and then rewriting the expected relative entropy as above, we have completed the proof.
\end{proof}

\begin{remark}
    Our framework and \Cref{theorem:qmi-gen-bound-qmgf-and-cmgf} also encompass cases where classical and quantum side information can be generated during the learning process. 
    If the risks and sub-gaussianity assumptions depend only on the data and hypothesis but not on the side information random variables and quantum registers, then we recover \Cref{eq:qmi-gen-bound-qmgf-and-cmgf}.
    That is, despite having more objects to take into account, the final bound remains the same and in particular only depends on the data and the hypothesis, not on additional side information. 
\end{remark}

\begin{remark}\label{remark:proof-modifications}
    Having presented the proof of \Cref{theorem:qmi-gen-bound-qmgf-and-cmgf}, we comment on some modifications.
    On the one hand, if we change the assumed \Cref{eq:quantum-moment-generating-function-bound-assumption} by allowing for $(s,w)$-dependent functions $\psi_{\pm;s,w}$, we can follow the same proof strategy. The obtained expected generalization error bound will differ from \Cref{eq:qmi-gen-bound-qmgf-and-cmgf} only in the first term on the r.h.s., which gets replaced by $\E_{(S,W)\sim P^{\alg}}\left[ \psi_{\mp;S,W}^{\ast -1}\left(D(\sigma^{\alg}(S,W)\|\rho_{\mathrm{test}}(S) \otimes \sigma^{\alg}_{\mathrm{hyp}}(S,W))\right)\right]$.

    On the other hand, if we change \Cref{eq:quantum-moment-generating-function-bound-assumption} to the (by Golden-Thompson weaker) assumption that
    \begin{equation}\label{eq:weakened-qmgf-assumption}
        \tr\left[e^{ \log (\rho_{\mathrm{test}}(s) \otimes \sigma^{\alg}_{\mathrm{hyp}}(s,w))  + \lambda \left( L(s,w) - \tr[L(s,w) (\rho_{\mathrm{test}}(s) \otimes \sigma^{\alg}_{\mathrm{hyp}}(s,w))]\mathbbm{1}_{\mathrm{test},\mathrm{hyp}}\right)}\right]\\
        \leq \begin{cases} \psi_+ (\lambda)\quad \textrm{if } \lambda \geq 0 \\ \psi_- (\lambda)\quad \textrm{if } \lambda <0 \end{cases},
    \end{equation}
    we can still recover \Cref{eq:qmi-gen-bound-qmgf-and-cmgf}. This can be seen by noticing that the second step in the proof of \Cref{theorem:qmi-gen-bound-qmgf-and-cmgf} was exactly to apply Golden-Thompson.

    Finally, \Cref{theorem:qmi-gen-bound-qmgf-and-cmgf} and its proof simplify in different scenarios, for instance for learners that produce either only a classical or only a quantum hypothesis. Concretely, if $\mathsf{W}$ is trivial, then we obtain a variant of \Cref{eq:qmi-gen-bound-qmgf-and-cmgf} without the Holevo information term and without the second summand on the r.h.s. In this case, the assumption \Cref{eq:classical-moment-generating-function-bound-assumption} is not needed.
    Furthermore, if $\mathcal{H}_{\mathrm{hyp}}$ is trivial, then we obtain a variant of \Cref{eq:qmi-gen-bound-qmgf-and-cmgf} without the first summand on the right-hand side. In this case, the assumption \Cref{eq:quantum-moment-generating-function-bound-assumption} is not needed.
    Similarly, if $\mathsf{Z}$ is trivial, the second summand vanishes, whereas if $\mathcal{H}_{\mathrm{data}}$ is trivial, the first summand vanishes, so that we recover \cite[Lemma 1]{Xu2017}.
    Moreover, if $ \sigma^{\alg}(s,w)= \sigma^{\alg}_{\mathrm{test}}(s,w) \otimes \sigma^{\alg}_{\mathrm{hyp}}(s,w)$ is already a tensor product state -- for example if each $\rho^{\alg} (s,w)$ factorizes or if each $E_s^\alg (w)$ is a pure state projector (so that monogamy of entanglements forbids the pure post-measurement state on the training system from being correlated or entangled with the test system) --, then we get a variant of \Cref{eq:qmi-gen-bound-qmgf-and-cmgf} without the QMI term.
    Finally, if $\rho(s)=\rho_{\mathrm{test}}(s)\otimes\rho_{\mathrm{train}}(s)$ factorizes, then both the QMI and the Holevo information contribution vanish and the assumption \Cref{eq:quantum-moment-generating-function-bound-assumption} is not needed.
\end{remark}

\begin{remark}\label{remark:measured-relative-entropy}
    As a consequence of \cite[Lemma 1 and Theorem 2]{berta2017variational} -- who applied Golden-Thompson in \cite[Proposition 5]{berta2017variational} similarly to our use in the proof of \Cref{theorem:qmi-gen-bound-qmgf-and-cmgf} --, we have in fact established an expected generalization error bound in terms of measured quantum information quantities.
    Namely, relying on \cite{berta2017variational}, we can tighten the initial inequality in our proof to 
    \begin{align}
        &D^{\mathbb{M}}(\sigma^{\alg}(s,w)\|\rho_{\mathrm{test}}(s) \otimes \sigma^{\alg}_{\mathrm{hyp}}(s,w))\\
        &\geq \lambda \tr[L(s,w) \sigma^{\alg}(s,w)] - \log \tr\left[(\rho_{\mathrm{test}}(s) \otimes \sigma^{\alg}_{\mathrm{hyp}}(s,w)) \exp\left(\lambda L(s,w)\right)\right] \, ,
    \end{align}
    where $D^{\mathbb{M}}(\rho\|\sigma)$ denotes the measured relative entropy. The quantum relative entropy $D(\rho\|\sigma)$ upper bounds  $D^{\mathbb{M}}(\rho\|\sigma)$, but there can be a gap between these two quantities.
\end{remark}

\begin{example}[\Cref{example:recovering-classical-risks} continued]
    The loss observables $L(s,w)=\ell(s,w)\cdot\mathbbm{1}_{\mathrm{test},\mathrm{hyp}}$ considered in \Cref{example:recovering-classical-risks} trivially satisfy \Cref{eq:quantum-moment-generating-function-bound-assumption} even for $\psi_\pm$ given by the $0$-function. With this choice, $\psi_\pm^\ast (t) = +\infty$ for all $t$ and $\psi_\pm^{\ast -1} (s) = 0$ for all $s$, so the first term in our bound vanishes. 
    Thus, \Cref{theorem:qmi-gen-bound-qmgf-and-cmgf} reproduces \cite[Lemma 1]{Xu2017} in this special case.
\end{example}

\Cref{theorem:qmi-gen-bound-qmgf-and-cmgf} takes a particularly simple and appealing form if the assumptions on the moment-generating functions are sub-gaussianity assumptions. Before stating the corresponding result, we recall the notions of sub-gaussianity in the cases of observables and random variables: 

\begin{definition}[Sub-gaussianity for observables]\label{definition:subgaussian-observables}
    Let $\alpha>0$. A self-adjoint loss observable $L\in\cB(\cH)$ is called $\alpha$-sub-gaussian with respect to~a quantum state $\sigma\in\mathcal{S}(\cH)$ if 
    \begin{equation}
        \log \tr\left[ \sigma \cdot e^{\lambda\left(L-\tr[L \sigma]\mathds{1}\right)}  \right]
        \leq \frac{\alpha^2 \lambda^2}{2} \label{eq:quantum-moment-generating-function-subgaussianity-assumption}
    \end{equation}
    holds for all $\alpha\in\mathbb{R}$.
\end{definition}

\begin{example}\label{example:quantum-concentration-inequalities}
    Quantum concentration inequalities recently received considerable attention in the literature. Prominent examples of classes of states for which bounds on the MGF are known include the following:
    \begin{enumerate}
        \item Local observables w.r.t.~high-temperature Gibbs states~\cite{kuwahara2020gaussian} and, more generally, Lipschitz observables w.r.t.~high temperature commuting Gibbs states~\cite{depalma2022quantumconcentration,capel2020modified} or $1D$-commuting Gibbs states~\cite{Bardet2021}, are known to satisfy sub-gaussianity with $\alpha=\mathcal{O}(1)$.
        \item Local observables w.r.t.~outcomes of shallow circuits also satisfy sub-gaussianity with $\alpha=\mathcal{O}(1)$~\cite{anshu2022concentration}.
        \item Lipschitz observables w.r.t.~tensor product states, up to a weakening à la Golden-Thompson analogously to \Cref{eq:weakened-qmgf-assumption}, satisfy sub-gaussianity with $\alpha=\mathcal{O}(1)$ \cite[Theorem 8.1]{depalma2023wassersteinspinsystems}.
        \item More generally, \cite{anshu2016concentration} proved concentration bounds for local observables w.r.t.~states with finite correlation length by bounding the MGF. However, they are weaker than sub-gaussian concentration and depend on the dimension of the underlying lattice.
    \end{enumerate}
\end{example}

\begin{definition}[Sub-gaussianity for random variables~{\cite[Section 2.5]{vershynin2018high-dimensional}}]\label{definition:subgaussian-random-variables}
    Let $\alpha>0$. 
    A real-valued random variable $X$ is $\alpha$-sub-gaussian if 
    \begin{equation}
        \log\E\left[e^{\lambda (X - \E[X])}\right]
        \leq \frac{\alpha^2 \lambda^2}{2}
    \end{equation}
    holds for all $\alpha\in\mathbb{R}$.
\end{definition}

\begin{example}
    Trivially, a gaussian random variable with variance $\beta^2$ is $\beta$-sub-gaussian.
    By Hoeffding's Lemma \cite{hoeffding1963probability}, any random variable that almost surely takes values in a bounded interval $[a,b]$ is $(\tfrac{b-a}{2})$-sub-gaussian.
    Finally, any $L$-Lipschitz function of a Haar-random variable on the unit sphere in $\mathbb{R}^n$ is $(\tfrac{CL}{\sqrt{n}})$-sub-gaussian for a suitable $C>0$ (see, e.g., \cite[Chapter 5]{vershynin2018high-dimensional}).
\end{example}

With these definitions, we can now compactly state the sub-gaussian versions of \Cref{theorem:qmi-gen-bound-qmgf-and-cmgf}:

\begin{corollary}\label{corollary:qmi-gen-bound-qmgf-and-cmgf-subgaussian}
    Let $\alpha,\beta >0$.
    Assume that the loss observable $L(s,w)$ is $\alpha$-sub-gaussian w.r.t.~$\rho_{\mathrm{test}}(s) \otimes \sigma^{\alg}_{\mathrm{hyp}}(s,w)$ for every $(s,w)\in\mathsf{Z}^m\times\mathsf{W}$. 
    Moreover, assume that the random variable $\tr[L(S,w)(\rho_{\mathrm{test}}(S) \otimes \sigma^{\alg}_{\mathrm{hyp}}(S,w))]$, with $S\sim P^m$, is $\beta$-sub-gaussian for every $w\in\mathsf{W}$. 
    Then,
    \begin{equation}        
        \begin{split}
            \lvert \gen_\rho (\alg)\rvert
            &\leq \sqrt{2\alpha^2 \left(\E_{(S,W)\sim P^{\mathcal{A}}}[I(\mathrm{test};\mathrm{hyp})_{\sigma(S,W)}] + \E_{S\sim P^m}\left[ \chi\left(\{P^{\mathcal{A}}_{\mathsf{W}|S}(w), \rho^{\alg}_{\mathrm{test}}(S,w)\}_{w\in\mathsf{W}}\right)\right]\right)}\\
            &\hphantom{\leq}~~+ \sqrt{2 \beta^2 I(S;W)} \, .
        \end{split}
    \end{equation}
\end{corollary}
\begin{proof}
    This follows from~\Cref{theorem:qmi-gen-bound-qmgf-and-cmgf} with the log-MGF bounds $\psi_\pm:\mathbb{R}\to\mathbb{R}$, $\psi_\pm(x)=\tfrac{\alpha^2 x^2}{2}$ and $\phi_\pm:\mathbb{R}\to\mathbb{R}$, $\phi_\pm(x)=\tfrac{\beta^2 x^2}{2}$. 
    This leads to $\psi_\pm^{\ast -1}(\xi) =\sqrt{2\alpha^2\xi}$ and $\phi_\pm^{\ast -1}(\xi) =\sqrt{2\beta^2\xi}$.
\end{proof}

So far, our generalization error bounds do not explicitly depend on the training data size $m$.
To achieve such a dependence, we now impose an i.i.d.~structure on the quantum data, in addition to the already assumed (but not yet fully exploited) i.i.d.~structure on the classical training data $S\sim P^m$.
Namely, we assume that the data Hilbert space and states factorize as
\begin{align}
    \cH_{\mathrm{data}}
    &=\cH_{\mathrm{test}}\otimes \cH_{\mathrm{train}} = \bigotimes_{i=1}^m (\cH_{\mathrm{test},i}\otimes \cH_{\mathrm{train},i}) = \bigotimes_{i=1}^m \cH_{\mathrm{data},i}\, ,\\
    \rho(s) 
    &= \rho(z_1,\ldots,z_m) = \bigotimes_{i=1}^m \rho_i (z_i)\, , \quad\textrm{with } \rho_i(z_i)\in\mathcal{S}(\cH_{\mathrm{test},i}\otimes \cH_{\mathrm{train},i})\, .
\end{align}
For our next result, we consider learners and loss observables that adhere to this factorization. 
On the one hand, we assume that the POVMs and channels used by the learner $\alg$ factorize as $E_{s}^\alg(w) = E_{z_1,\ldots,z_m}^\alg(w) = \bigotimes_{i=1}^m E_{z_i}^\alg(w)$ and $\Lambda_{s,w}^\alg = \Lambda_{z_1,\ldots,z_m,w}^\alg = \bigotimes_{i=1}^m \Lambda_{z_i, w}^\alg$ with $E_{z_i}^\alg(w)\in\mathcal{E}(\cH_{\mathrm{train},i})$ and $\Lambda_{z_i, w}^\alg:\mathcal{T}_{1}(\cH_{\mathrm{train},i})\to \mathcal{T}_{1}(\cH_{\mathrm{hyp},i})$. 
Note that this in particular comes with factorizations $\cH_{\mathrm{hyp}} = \bigotimes_{i=1}^m \cH_{\mathrm{hyp},i}$ of the hypothesis Hilbert space and $\sigma^{\alg}(s,w) = \bigotimes_{i=1}^m \sigma^{\alg}_i(z_i, w)$ of the state after the action of $\alg$, with $\sigma^{\alg}_i(z_i, w)\in\mathcal{S}(\cH_{\mathrm{test},i}\otimes \cH_{\mathrm{hyp},i})$.
On the other hand, we assume the loss observables to be of the local form $L(s,w)=\tfrac{1}{m}\sum_{i=1}^m L_i(z_i,w)$, with $L_i(z_i,w)\in\mathcal{B}(\cH_{\mathrm{test},i}\otimes \cH_{\mathrm{hyp},i})$ acting only on the $i$th test and hypothesis subsystems. (For readability, we notationally suppress identities on the remaining subsystems when convenient.)
In this setting, \Cref{corollary:qmi-gen-bound-qmgf-and-cmgf-subgaussian} gives the following result:

\begin{corollary}\label{corollary:qmi-gen-bound-qmgf-and-cmgf-subgaussian-independent}
    Assume the above factorization for the quantum data and the learner $\alg$ as well as the above local structure of the loss observables.
    Moreover, assume that $L_i(z_i, w)$ is $\alpha_i$-sub-gaussian w.r.t.~$\rho_{\mathrm{test},i}(z_i) \otimes \sigma^{\alg}_{\mathrm{hyp},i}(z_i,w)$ for every $(z_i,w)\in\mathsf{Z}\times \mathsf{W}$ and $1\leq i\leq m$, and that the random variable $\tr\left[L_i(Z_i, w) ( \rho_{\mathrm{test},i}(Z_i) \otimes \sigma^{\alg}_{\mathrm{hyp},i}(Z_i,w))\right]$, with $Z_i\sim P$, is $\beta_i$-sub-gaussian for every $w\in\mathsf{W}$ and $1\leq i\leq m$.
    Then,
    \small
    \begin{equation}
        \begin{split}
            \lvert \gen_\rho (\alg)\rvert
            &\leq  \sqrt{\frac{2 \sum_{i=1}^m \alpha_i^2}{m^2}\left( \E_{(S,W)\sim P^{\alg}}\left[ \sum_{i=1}^m I(\mathrm{test};\mathrm{hyp})_{\sigma^{\alg}_i(Z_i,W)}\right]+ \E_{S\sim P^m}\left[ \chi\left(\{P^{\mathcal{A}}_{\mathsf{W}|S}(w), \rho^{\alg}_{\mathrm{test}} (S,w)\}_{w\in\mathsf{W}}\right)\right]\right)} \\
            &\hphantom{\leq}~~+ \sqrt{\frac{2 \sum_{i=1}^m \beta_i^2}{m^2} I(S;W)}\,  .
        \end{split}
    \end{equation}
    \normalsize
    In particular, if $\alpha_i = \alpha_0$ and $\beta_i=\beta_0$ for all $1\leq i\leq m$, then 
    \small
    \begin{equation}
        \begin{split}
            \lvert \gen_\rho (\alg)\rvert
            &\leq  \sqrt{\frac{2\alpha_0^2}{m} \left(\E_{(S,W)\sim P^{\alg}}\left[ \sum_{i=1}^m I(\mathrm{test};\mathrm{hyp})_{\sigma^{\alg}_i(Z_i,W)}\right] + \E_{S\sim P^m}\left[ \chi\left(\{P^{\mathcal{A}}_{\mathsf{W}|S}(w), \rho^{\alg}_{\mathrm{test}} (S,w)\}_{w\in\mathsf{W}}\right)\right]\right)} \\
            &\hphantom{\leq}~~+ \sqrt{\frac{2 \beta_0^2}{m} I(S;W)}\, .
        \end{split}
    \end{equation}
    \normalsize
\end{corollary}
\begin{proof}
    See \Cref{appendix:proofs}.
\end{proof}

We point out that the factorization assumption on the POVM 
elements $E_s^{\alg}(w)$ is not needed if $\alg$ produces only a classical hypothesis. In this case, the $\mathrm{hyp}$ quantum system is trivial. Thus, $\rho_{\mathrm{test}}(s)\otimes\sigma^\alg_{\mathrm{hyp}}(s,w)=\rho_{\mathrm{test}}(s)=\bigotimes_{i=1}^m \rho_{\mathrm{test},i}(z_i)$ factorizes by assumption, which is sufficient for the proof of \Cref{corollary:qmi-gen-bound-qmgf-and-cmgf-subgaussian-independent}. Even in this setting, the Holevo information term in
the bound is an in general non-trivial quantum contribution.

If, however, the learner produces a non-trivial quantum hypothesis, our current proof strategy does rely on the factorization assumption. Notice, however, that \Cref{example:quantum-concentration-inequalities} already contains QMGF bounds w.r.t.~non-product states. Thus, insights into CMGF bounds w.r.t.~non-product states may allow future work to improve upon our proof of Corollary 24, extending it to more general (non-product) learners.

Let us return to our continuing example of quantum state classification and see the implications of our generalization bounds in that setting. 

\begin{example}[Quantum state classification -- \Cref{example:quantum-state-classification-framework,example:quantum-state-classification-risks} continued]\label{example:quantum-state-classification-gen-bound}
    As the learner $\alg$ in our quantum state classification example produces only a classical hypothesis and as the initial quantum data states $\rho(s)$ factorize across the test-train bipartition, it suffices to verify a suitable classical sub-gaussianity assumption. 
    Observe that, for every $(s,w)\in\{0,1\}^m\times\mathsf{W}$, the state
    \begin{equation}
        \rho_{\mathrm{test}}(s)
        = \E_{\{(\sigma_0^{(i)},\sigma_1^{(i)})\}_{i=1}^m\sim P_{\mathrm{pair}}^m}\left[\bigotimes_{i=1}^m \sigma_{z_i}^{(i)}\right]
        = \bigotimes_{i=1}^m \E_{\{(\sigma_0^{(i)},\sigma_1^{(i)})\}_{i=1}^m\sim P_{\mathrm{pair}}^m}\left[\sigma_{z_i}^{(i)}\right]
        = \bigotimes_{i=1}^m\rho_{\mathrm{test},i}(z_i)
    \end{equation}
    is an $m$-fold tensor product.
    Moreover, the loss observables defined in \Cref{example:quantum-state-classification-risks} are local w.r.t.~this tensor factorization.
    So, to apply \Cref{corollary:qmi-gen-bound-qmgf-and-cmgf-subgaussian-independent}, we consider the sub-gaussianity parameter $\beta_0$ of the random variable $\tr[((1-Z_i)(\mathds{1}_d-F(w)) + Z_i F(w)) \rho_{\mathrm{test}}(Z_i)]$, with $Z_i\sim P$.
    Without any prior assumptions on the distribution $P$ and on the mapping $z\mapsto \rho(z)$, the random variable of interest takes values in $[0,1]$ because $0\leq F(w), 1-F(w)\leq 1$.
    Thus, Hoeffding's Lemma \cite{hoeffding1963probability} implies $\beta_0\leq\nicefrac{1}{2}$ for every $w\in\mathsf{W}$.
    Therefore, \Cref{corollary:qmi-gen-bound-qmgf-and-cmgf-subgaussian-independent} yields
    \begin{equation}\label{eq:gen-bound-quantum-state-classification}
        \lvert R_\rho (\alg) - \hat{R}_\rho (\alg)\rvert
        \leq \sqrt{\frac{1}{2m} I(S;W)} \, .
    \end{equation}
    If $\mathsf{W}$ is finite, then we immediately have the mutual information upper bound $I(S;W) \leq\log\lvert\mathsf{W}\rvert$. Hence, our above bound implies that we can guarantee a small expected generalization error as soon as the training data size $m$ is of the same order as the number of bits needed to describe the classical hypotheses.
    If $\mathsf{W}$ is infinite, we may first discretize and then apply the bound. 
    Concretely, if $\varepsilon>0$ and if $\mathsf{W}_\varepsilon\subseteq\mathsf{W}$ is an $\varepsilon$-covering net for $\mathsf{W}$ w.r.t.~the sup-norm, then $\lvert R_\rho (\alg) - \hat{R}_\rho (\alg)\rvert \leq \varepsilon + \sqrt{\frac{1}{2m} \log\lvert \mathsf{W}_\varepsilon\rvert}$. 
    If there are no prior assumptions on the admissible effect operators $\{F(w)\}_{w\in\mathsf{W}}$, then we cannot expect better bounds on the cardinality of an $\varepsilon$-covering net for $\mathsf{W}$ than $\log\lvert \mathsf{W}_\varepsilon\rvert \leq \tilde{\mathcal{O}}\left(\min\{\nicefrac{d}{\varepsilon^2}, d^2 \log(\nicefrac{1}{\varepsilon})\}\right)$ \cite[Section 4]{cheng2016learnability}\footnote{Here, the $\tilde{\mathcal{O}}$ hides non-leading logarithmic factors.}.
    In the case of $n$ qubits, we have $d=2^n$ and the resulting bound scales exponentially with $n$. This can be improved if $\{F(w)\}_{w\in\mathsf{W}}$ is limited. 
    For example, if $F(w)$ is a sum of $k$-local Pauli terms for every $w\in\mathsf{W}$, where $k=\mathcal{O}(1)$, then, since there are at most $\mathcal{O}(n^k)$ such terms, one can obtain an improved covering number bound of $\log\lvert \mathsf{W}_\varepsilon\rvert \leq \tilde{\mathcal{O}} ( n^k \log(\nicefrac{1}{\varepsilon}))$, which scales polynomially in $n$. This can be improved further if the locality assumption is strengthened to geometric locality.
    Note that these bounds on $I(S;W)$ are worst-case and we expect tighter algorithm-dependent bounds to be possible when taking the POVMs $\{E_s^\alg (w)\}_{w\in\mathsf{W}}$ chosen by the learner into account. This concludes the discussion of our state classification example.
\end{example}

As it concerns a special case with only a classical hypothesis, \Cref{eq:gen-bound-quantum-state-classification} can already be deduced from the classical generalization bounds of \cite{Xu2017}.
In the next section, we demonstrate the applicability of our general framework and our generalization error bounds for a variety of quantum learning problems, including scenarios that cannot be studied with the purely classical framework. 
Before this discussion, we conclude this section with an extension of \Cref{corollary:qmi-gen-bound-qmgf-and-cmgf-subgaussian-independent} to stable learners that use channels leading to a controlled increase of Lipschitz constants:

\begin{corollary}\label{corollary:lipschitz-gen-bound}
    Assume the above factorization for the quantum data and the POVMs used by the learner as well as the above local structure for the loss observables.
    Furthermore, assume that the Heisenberg picture duals $(\Lambda_{s,w}^\alg)^\ast$ of the channels $\Lambda_{s,w}^\alg$ used by $\alg$ satisfy $\|(\Lambda_{s,w}^\alg)^\ast\|_{\mathrm{Lip}\to\mathrm{Lip}}\leq C_{1}$ as well as $\max_{s\sim s', w}\lVert (\Lambda_{s,w}^\alg - \Lambda_{s',w}^\alg)^\ast \rVert_{\mathrm{Lip}\to \infty}\leq C_2$, where $s\sim s'$ denotes neighboring training data sets (i.e., training data sets that differ only in a single data point).
    Then,
    \small
    \begin{equation}
        \begin{split}
            &\lvert \gen_\rho (\alg)\rvert\\
            &\leq \frac{2\sqrt{2} \displaystyle\max_{i,z_i,w}\| L_i(z_i,w)\|}{\sqrt{m}}\Bigg(\sqrt{C_1\left(\E_{(S,W)\sim P^{\alg}}\left[ I(\mathrm{test};\mathrm{hyp})_{\sigma^{\alg}(S,W)}\right]+ \E_{S\sim P^m}\left[ \chi\left(\{P^{\mathcal{A}}_{\mathsf{W}|S}(w), \rho^{\alg}_{\mathrm{test}}(S,w)\}_{w\in\mathsf{W}}\right)\right]\right)}\\
            &\hphantom{\frac{2\sqrt{2} \displaystyle\max_{i,z_i,w}\| L_i(z_i,w)\|}{\sqrt{m}}\Bigg(}~~~+ \sqrt{(1 + C_1 (1+C_2)) I(S;W)} \Bigg)\, .
        \end{split}
    \end{equation}
    \normalsize
\end{corollary}

In the assumed bound $\|(\Lambda_{s,w}^\alg)^\ast\|_{\mathrm{Lip}\to\mathrm{Lip}}\leq C_{1}$, the Lipschitz constants considered are w.r.t.~the factorizations $\mathcal{H}_{\mathrm{test}}\otimes \mathcal{H}_{\mathrm{hyp}} = \bigotimes_{i=1}^m (\mathcal{H}_{\mathrm{test},i}\otimes \mathcal{H}_{\mathrm{hyp},i})$ and $\mathcal{H}_{\mathrm{test}}\otimes \mathcal{H}_{\mathrm{train}} = \bigotimes_{i=1}^m (\mathcal{H}_{\mathrm{test},i}\otimes \mathcal{H}_{\mathrm{train},i})$.
Similarly, the Lipschitz constants relevant for the stability assumption $\max_{s\sim s', w}\lVert (\Lambda_{s,w}^\alg - \Lambda_{s',w}^\alg)^\ast\rVert_{\mathrm{Lip}\to \infty}\leq C_2$ are w.r.t.~$\mathcal{H}_{\mathrm{test}}\otimes \mathcal{H}_{\mathrm{hyp}} = \bigotimes_{i=1}^m (\mathcal{H}_{\mathrm{test},i}\otimes \mathcal{H}_{\mathrm{hyp},i})$.
Again, the POVM factorization assumption is not needed if the learner only produces a classical hypothesis.

\begin{proof}
    Recall from \Cref{remark:channel-absorbed-in-loss} that we obtain the same notions of risk when absorbing the channels $\Lambda_{s,w}^\alg$ into the loss observables via the Heisenberg picture. Thus, instead of proving sub-gaussianity of $L(s,w)$ w.r.t.~$\rho_{\mathrm{test}}(s) \otimes \sigma^{\alg}_{\mathrm{hyp}}(s,w)$, we can also establish sub-gaussianity of $(\Lambda_{s,w}^\alg)^\ast (L(s,w))$ w.r.t.~$\rho_{\mathrm{test}}(s) \otimes \rho^\alg_{\mathrm{train}}(s,w)$. We do this in the first part of the proof.
    As $\|(\Lambda_{s,w}^\alg)^\ast\|_{\mathrm{Lip}\to\mathrm{Lip}}\leq C_1$, we have 
    \begin{equation}
        \|(\Lambda_{s,w}^\alg)^\ast (L(s,w))\|_{\mathrm{Lip}}
        \leq C_1\| L(s,w)\|_{\mathrm{Lip}}
        \leq \frac{2C_1 \max_{i,z_i,w}\| L_i(z_i,w)\|}{m}\, ,
    \end{equation}
    where the last step used \cite[Proposition 8]{depalma2021quantumwasserstein}.
    Therefore, according to \cite[Theorem 8.1]{depalma2023wassersteinspinsystems}, which we restate as \Cref{lemma:lipschitz-observables-subgaussian}, the observable $(\Lambda_{s,w}^\alg)^\ast (L(s,w))$ satisfies a version of $(m^{-1/2}\cdot 2 C_1 \max_{i,z_i,w}\| L_i(z_i,w)\|)$-sub-gaussianity w.r.t.~the $m$-fold tensor product $\bigotimes_{i=1}^m \rho_{\mathrm{test},i} (z_i)\otimes\rho^\alg_{\mathrm{train},i} (z_i,w)$ weakened analogously to \Cref{eq:weakened-qmgf-assumption}. As argued in \Cref{remark:proof-modifications}, this weaker version is a sufficient quantum sub-gaussianity for our purposes.

    Next, we establish a suitable classical sub-gaussianity. To this end, take two training data sets $s=(z_1,\ldots,z_m),s'=(z'_1,\ldots,z'_m)\in\mathsf{Z}^m$ that differ in exactly one data point, i.e., $\exists 1\leq i\leq m$ such that $z_i\neq z'_i$ and $z_j=z'_j$ for all $j\neq i$. For this relation, we use the shorthand $s\sim s'$.
    Then, because of our assumed factorization of the quantum data states and of the POVMs used by the learner, the post-measurement states $\rho_{\mathrm{test}}(s) \otimes \rho^\alg_{\mathrm{train}}(s,w)$ and $\rho_{\mathrm{test}}(s) \otimes \rho_{\mathrm{train}}(s',w)$ agree after tracing out the $i$th subsystem, i.e., $\tr_{\mathrm{test},i;\mathrm{hyp},i}[\rho_{\mathrm{test}}(s) \otimes \rho^\alg_{\mathrm{train}}(s,w)] = \tr_{\mathrm{test},i;\mathrm{hyp},i}[\rho_{\mathrm{test}}(s') \otimes \rho^\alg_{\mathrm{train}}(s',w)]$ for all $w\in\mathsf{W}$.
    Hence, by definition of the quantum Lipschitz constant (compare \cite[Definition 8]{depalma2021quantumwasserstein}), we obtain the bound
    \begingroup
    \allowdisplaybreaks
    \small
    \begin{align}
        &\left\lvert \tr[L(s,w) \left(\rho_{\mathrm{test}}(s) \otimes \sigma^{\alg}_{\mathrm{hyp}}(s,w)\right)] - \tr[L(s',w) \left(\rho_{\mathrm{test}}(s') \otimes \sigma^{\alg}_{\mathrm{hyp}}(s',w)\right)] \right\rvert\\
        &= \left\lvert  \tr[(\Lambda_{s,w}^\alg)^\ast (L(s,w)) \left(\rho_{\mathrm{test}}(s) \otimes \rho^{\alg}_{\mathrm{train}}(s,w)\right)] - \tr[(\Lambda_{s',w}^\alg)^\ast (L(s',w)) \left(\rho_{\mathrm{test}}(s') \otimes \rho^{\alg}_{\mathrm{train}}(s',w)\right)]\right\rvert\\
        &\leq \left\lvert  \tr[(\Lambda_{s,w}^\alg)^\ast (L(s,w)) \left(\rho_{\mathrm{test}}(s) \otimes \rho^{\alg}_{\mathrm{train}}(s,w)\right)] - \tr[(\Lambda_{s,w}^\alg)^\ast (L(s',w)) \left(\rho_{\mathrm{test}}(s) \otimes \rho^{\alg}_{\mathrm{train}}(s,w)\right)]\right\rvert\\
        &\hphantom{\leq }~~ + \left\lvert \tr[(\Lambda_{s,w}^\alg)^\ast (L(s',w)) \left(\rho_{\mathrm{test}}(s) \otimes \rho^{\alg}_{\mathrm{train}}(s,w)\right)] - \tr[(\Lambda_{s',w}^\alg)^\ast (L(s',w)) \left(\rho_{\mathrm{test}}(s) \otimes \rho^{\alg}_{\mathrm{train}}(s,w)\right)]\right\rvert\\
        &\hphantom{\leq }~~ + \left\lvert \tr[(\Lambda_{s',w}^\alg)^\ast (L(s',w)) \left(\rho_{\mathrm{test}}(s) \otimes \rho^{\alg}_{\mathrm{train}}(s,w)\right)] - \tr[(\Lambda_{s',w}^\alg)^\ast (L(s',w)) \left(\rho_{\mathrm{test}}(s') \otimes \rho^{\alg}_{\mathrm{train}}(s',w)\right)]\right\rvert\\
        &\leq \left\lvert \tr[(L(s,w) - L(s',w))\left(\rho_{\mathrm{test}}(s) \otimes \sigma^{\alg}_{\mathrm{hyp}}(s,w)\right) ] \right\rvert\\
        &\hphantom{\leq }~~ +\left\lvert \tr[(\Lambda_{s,w}^\alg - \Lambda_{s',w}^\alg)^\ast (L(s',w)) \left(\rho_{\mathrm{test}}(s) \otimes \rho^{\alg}_{\mathrm{train}}(s,w)\right)] \right\rvert\\
        &\hphantom{\leq }~~ +  \frac{2C_1 \max_{i,z_i,w}\| L_i(z_i,w)\|}{m}\\
        &\leq \lVert L(s,w) - L(s',w)\rVert\cdot\lVert\rho_{\mathrm{test}}(s) \otimes \sigma^{\alg}_{\mathrm{hyp}}(s,w)\rVert_1\\
        &\hphantom{\leq }~~ +\lVert (\Lambda_{s,w}^\alg - \Lambda_{s',w}^\alg)^\ast (L(s',w))\rVert\cdot \lVert\rho_{\mathrm{test}}(s) \otimes \rho^{\alg}_{\mathrm{train}}(s,w)\rVert_1\\
        &\hphantom{\leq }~~ +  \frac{2C_1 \max_{i,z_i,w}\| L_i(z_i,w)\|}{m}\\
        &\leq \frac{2 \max_{i,z_i,w}\| L_i(z_i,w)\|}{m} + \frac{2C_1 \max_{i,z_i,w}\| L_i(z_i,w)\|}{m}\cdot \max_{s\sim s', w}\lVert (\Lambda_{s,w}^\alg - \Lambda_{s',w}^\alg)^\ast\rVert_{\mathrm{Lip}\to \infty}\\
        &\hphantom{\leq }~~ +  \frac{2C_1 \max_{i,z_i,w}\| L_i(z_i,w)\|}{m}\\
        &= \frac{2 \max_{i,z_i,w}\| L_i(z_i,w)\|}{m}\left(1 + C_1 \left(1 + \max_{s\sim s', w}\lVert (\Lambda_{s,w}^\alg - \Lambda_{s',w}^\alg)^\ast\rVert_{\mathrm{Lip}\to \infty}\right)\right)\\
        &\leq \frac{2 \max_{i,z_i,w}\| L_i(z_i,w)\|}{m}\left(1 + C_1 \left(1 + C_2\right)\right)\, .
    \end{align}
    \normalsize
    \endgroup
    Therefore, the random variable $\tr[L(S,w) \left(\rho_{\mathrm{test}}(S) \otimes \sigma^{\alg}_{\mathrm{hyp}}(S,w)\right)]$, with $S\sim P^m$, is sub-gaussian with sub-gaussianity parameter $\left(m^{-1/2} \cdot 2\max_{i,z_i,w}\| L_i(z_i,w)\|\left(1 + C_1 \left(1 + C_2\right)\right)\right)$, by McDiarmid's bounded differences inequality \cite{mcdiarmid1989onthemethod}.
    
    We can now apply \Cref{corollary:qmi-gen-bound-qmgf-and-cmgf-subgaussian} with the classical and quantum sub-gaussianity parameters established above and obtain the claimed generalization bound.
\end{proof}

A short discussion of the assumptions made on the channels $\Lambda^\alg_{s,w}$ is in order.
On the one hand, we assume that their Heisenberg duals $(\Lambda^\alg_{s,w})^\ast$ lead to a bounded increase in quantum Lipschitz constants, namely that $\|(\Lambda_{s,w}^\alg)^\ast\|_{\mathrm{Lip}\to\mathrm{Lip}}\leq C_{1}$. 
Equivalently, the maps $\Lambda^\alg_{s,w}$ should lead to a limited increase of quantum Wasserstein-$1$ norms, that is, $\|\Lambda_{s,w}^\alg\|_{W_1\to W_1}\leq C_{1}$.
This is satisfied for approximately locality-preserving channels such as constant-depth circuits or short-time evolutions under a local Lindblad generator~\cite{depalma2023wassersteinspinsystems,PRXQuantum.4.010309}(with associated Lieb-Robinson bound).
Also, as the proof of \cite[Theorem 8.1]{depalma2023wassersteinspinsystems} shows, this property is satisfied with $C_1 = 1$ for $m$-fold tensor products of single-qudit channels.
Moreover, for channels described by quantum circuits with local depolarizing noise, we obtain a $C_1$ that decays exponentially with the circuit depth for large enough noise strength compared to the size of the light-cone of each layer (compare the proof of \cite[Proposition IV.8.]{hirche2023quantum}).

On the other hand, we assume that $\lVert (\Lambda_{s,w}^\alg - \Lambda_{s',w}^\alg)^\ast \rVert_{\mathrm{Lip}\to \infty}\leq C_2$ for any neighboring data sets $s\sim s'$. Note that, as a consequence of \cite[Proposition 9]{depalma2021quantumwasserstein} we can rewrite this as $\lVert (\Lambda_{s,w}^\alg - \Lambda_{s',w}^\alg)^\ast \rVert_{\mathrm{Lip}\to \infty} = \lVert \Lambda_{s,w}^\alg - \Lambda_{s',w}^\alg \rVert_{1\to W_1}\leq C_2$.
This is a stability assumption: When the two classical data sets $s,s'$ differ in only a single data point, the quantum channels $\Lambda_{s,w}^\alg$ and $\Lambda_{s',w}^\alg$ employed by the learner $\alg$ must not differ too much.
It is reminiscent of classical replace-one stability \cite{bousquet2000algorithmic-stability, bousquet2002stability, shalev-shwartz2010stability}.
Using \cite[Corollary 2]{depalma2021quantumwasserstein}, we see that this quantum stability assumption is for example satisfied if, for every $s\sim s'$ and for every $w$, we can write $\Lambda_{s,w}^\alg - \Lambda_{s',w}^\alg = (\mathcal{N}_{s,w} - \mathcal{N}_{s',w}) \mathcal{M}_w $ with $\mathcal{M}_w$ an arbitrary CPTP map and with CPTP maps $\mathcal{N}_{s,w}, \mathcal{N}_{s',w}$ that act non-trivially only on a constant number of training data subsystems.
As this is in particular satisfied for learners that factorize, we can indeed view \Cref{corollary:lipschitz-gen-bound} as an extension of \Cref{corollary:qmi-gen-bound-qmgf-and-cmgf-subgaussian-independent}.

One strength of the results presented in this section is that they encompass a variety of learning tasks. However, when applied to a specific scenario, they do not necessarily lead to optimal bounds. For instance, our bounds in \Cref{corollary:qmi-gen-bound-qmgf-and-cmgf-subgaussian-independent,corollary:lipschitz-gen-bound} have a ``slow rate'' of $\nicefrac{1}{\sqrt{m}}$, which is to be contrasted with the ``fast rate'' of $\nicefrac{1}{m}$ recently achieved by, among others, \cite{hellstrom2021fast, grunwald2021pac-bayes, wang2023tighter} for classical information-theoretic generalization bounds and by \cite{mai2017pseudo} in the context of PAC-Bayesian quantum state tomography w.r.t.~squared Frobenius norm.
We leave proving improved quantum information-theoretic generalization bounds with fast rates to future work.

\section{Applications}\label{sec:applications}
\numberwithin{equation}{subsection}

\subsection{PAC learning quantum states}\label{section:pac-learning-quantum-states}

For our first application, we consider a setting of PAC learning quantum states, going back to \cite{Aaronson2007}. Here, the goal is to predict expectation values w.r.t.~an unknown state on average over an unknown distribution over effect operators.
Take the data Hilbert space
\begin{equation}
    \cH_{\mathrm{data}}
    = \cH_{\mathrm{test}}\otimes \cH_{\mathrm{train}}
    =((\mathbb{C}^d)^{\otimes m_{\mathrm{test}}})^{\otimes m}\otimes ((\mathbb{C}^d)^{\otimes m_{\mathrm{train}}})^{\otimes m}
\end{equation}
for some $d\in\mathbb{N}$ and $m, m_{\mathrm{test}}, m_{\mathrm{train}}\in\mathbb{N}$.  
Let the quantum data state $\rho$ be the CQ state given by
\begin{equation}
    \rho
    = \E_{S=(Z_1,\ldots,Z_{2m}) \sim P^{2m}}\left[ \left( \bigotimes_{i=1}^{2m} \ketbra{Z_i}{Z_i}\right) \otimes (\rho_0^{\otimes m_{\mathrm{test}}})^{\otimes m} \otimes (\rho_0^{\otimes m_{\mathrm{train}}})^{\otimes m}\right] ,
\end{equation}
where $\rho_0\in\mathcal{S}(\mathbb{C}^d)$ is the unknown qudit state to be PAC-learned, we imagine that each $z\in\mathsf{Z}$ comes with an associated qudit effect operator $E(z)\in\mathcal{E}(\mathbb{C}^d)$, and $P$ is an unknown probability distribution over $\mathsf{Z}$.
That is, the CQ data consists of independent copies of an unknown state that we are trying to learn, as well as of (classical descriptions of) random two-outcome POVM measurements drawn i.i.d.~from $P$.

We describe a simple quantum learner $\alg$ for this scenario as follows: Take $\cH_{\mathrm{hyp}}$ to be trivial, and take $\mathsf{W}$ to be some measurable hypothesis space. 
Here, we imagine each classical hypothesis $w\in\mathsf{W}$ to be associated to some hypothesis state $\rho_0 (w)\in\mathcal{S}(\mathbb{C}^d)$ that the learner could output.
Upon seeing the classical data $s=(z_1,\ldots,z_{2m})\in\mathsf{Z}^{2m}$, the learner performs a two-step procedure:
Let $\tilde{\varepsilon}>0$ be an auxiliary accuracy parameter, which we determine later.
First, the learner takes $\mathsf{W}_1 \subset \mathsf{W}$ to be a $\tilde{\varepsilon}$-covering of the hypothesis space $\mathsf{W}$ w.r.t.~the empirical seminorm $\Norm{\cdot}_{2, \{z_j\}_{j=1}^m}$ defined as 
\begin{equation}
    \Norm{w}_{2, \{z_j\}_{j=1}^m} 
    = \sqrt{\frac{1}{m}\sum_{j=1}^m \lvert \tr[E(z_j) \rho_0 (w)]\rvert^2}\, .
\end{equation}
Second, for each $m+1\leq i\leq 2m$, the learner measures the $2$-outcome POVM $\{E(z_i), \mathbbm{1}-E(z_i)\}$ separately on $m_{\mathrm{train}}$ copies of $\rho_0$, obtaining outcomes $b_\ell^{(i)}, 1\leq \ell\leq m_{\mathrm{train}}$, and then uses the empirical average $\tilde{b}^{(i)}:=\tfrac{1}{m_{\mathrm{train}}}\sum_{\ell =1}^{m_{\mathrm{train}}} b_\ell^{(i)}$ as an estimate of $\tr[ E(z_i) \rho_0]$. The quantum learner then outputs an empirical risk minimizing hypothesis
\begin{equation}
    \hat{w}\in\argmin_{w\in\mathsf{W}_1} \frac{1}{m}\sum_{i=m+1}^{2m} \left\lvert \tr[ E(z_i) \rho_0(w)] - \tilde{b}^{(i)}\right\rvert
    =: \argmin_{w\in\mathsf{W}_1} \hat{R}_{s_{(m+1):2m}, b_\ell^{(i)} }^{\mathrm{train}} (w).
\end{equation}
If there are multiple empirical risk minimizers, the tie is broken arbitrarily (but, for simplicity of notation, deterministically). 
Note: Both building the empirical covering net and performing empirical risk minimization over that net are computationally inefficient in general. Here, we focus on information-theoretic aspects and ignore computational complexity.

As in \Cref{example:quantum-state-classification-framework}, the family of quantum channels associated to this quantum learner is trivial, since there is no quantum hypothesis.
Thus, following \Cref{eq:learner-output-state-classical-hypothesis-expectation}, when letting the learner $\alg$ act on the quantum data state $\rho$, we obtain the output state
\begin{equation}
    \sigma
    = \E_{W\sim P^{\alg}_{\mathsf{W}}}\E_{S\sim P^{\alg}_{\mathrm{data}}|W} \left[ \ketbra{S}{S} \otimes \rho_{\mathrm{test}}\otimes \ketbra{\hat{w}}{\hat{w}} \right] ,
\end{equation}
with quantum test state $\rho_{\mathrm{test}}= (\rho_0^{\otimes m_{\mathrm{test}}})^{\otimes m}$ and with the probability distribution $P^{\alg}$ on $\mathsf{Z}^m\times \mathsf{W}$ given by
\begin{equation}
    P^{\alg} (s, \hat{w})
    = P^m(s)\cdot P^{\alg} (\hat{w}\rvert s)
    = P^m(s)\cdot \mathbb{P}_{B_\ell^{(i)}\rvert s}\left[ \hat{w}\in\argmin_{w\in\mathsf{W}_1} \hat{R}_{s_{(m+1):2m}, B_{\ell}^{(i)}}^{\mathrm{train}} (w)\right] ,
\end{equation}
where the $B_\ell^{(i)}$ are $\{0,1\}$-valued random variables which become independent when conditioned on $s$, with probability distributions 
\[
    \mathbb{P}_{B_\ell^{(i)}\rvert s}[B_\ell^{(i)} = 1] = \tr[ E(z_i) \rho_0] = 1 - \mathbb{P}_{B_\ell^{(i)}\rvert s}[B_\ell^{(i)} = 0]
\] 
for all $1\leq\ell\leq m_{\mathrm{train}}$ and for all $m+1\leq i\leq 2m$.
While more general quantum learners are possible, for instance by allowing for general $s$-dependent POVM elements, the simple quantum learner presented here is similar in spirit to \cite{Aaronson2007} and \cite[Section 4.2]{Xu2017}. As we show below, we can make guarantees on its performance based on \Cref{corollary:qmi-gen-bound-qmgf-and-cmgf-subgaussian-independent}.

Given that our quantum learner is based on empirical risk minimization, we define the loss observables in analogy to the notion of empirical risk used above. Namely, for each $1\leq i\leq m$ and $c_\ell^{(i)}$, we set
\begin{equation}
    L_{c_\ell^{(i)}}^{(i)}(z_i, w)
    = L_{c_\ell^{(i)}}^{(i)}(z_i)
    = \bigotimes_{\ell =1}^{m_{\mathrm{test}}}\left( c_{\ell}^{(i)} E(z_i) + (1-c_{\ell}^{(i)}) (\mathbbm{1}_d-E(z_i))\right)
\end{equation}
and
\begin{equation}
    L(s, w)
    = L(s)
    = \sum_{c_{\ell}^{(i)}\in \{0,1\}} \hat{R}_{s_{(m+1):2m}, c_\ell^{(i)} }^{\mathrm{test}} (w)\cdot (\mathbbm{1}_d^{\otimes m_{\mathrm{test}}})^{\otimes (i-1)}\otimes L_{c_\ell^{(i)}}^{(i)}(z_i)\otimes (\mathbbm{1}_d^{\otimes m_{\mathrm{test}}})^{\otimes (m-i)} ,
\end{equation}
with $\hat{R}_{s_{(m+1):2m}, c_\ell^{(i)} }^{\mathrm{test}} (w) = \frac{1}{m}\sum_{i=m+1}^{2m} \left\lvert \tr[ E(z_i) \rho_0(w)] - \frac{1}{m_{\mathrm{test}}}\sum_{\ell =1}^{m_{\mathrm{test}}} c_\ell^{(i)}\right\rvert$, 
Plugging these choices into \Cref{definition:expected-empirical-risk}, we obtain the expected empirical risk 
\begin{equation}
    \hat{R}_\rho (\alg)
    = \E_{S\sim P^{2m}} \E_{C_\ell^{(i)}\rvert S} \E_{\hat{W}}\left[ \hat{R}_{S_{(m+1):2m}, C_\ell^{(i)} }^{\mathrm{test}} (\hat{W}) \right] ,
\end{equation}
where the $C_\ell^{(i)}$ are $\{0,1\}$-valued random variables which become independent when conditioned on $s$, with probability distributions 
\begin{equation}
    \mathbb{P}_{C_\ell^{(i)}\rvert s}[C_\ell^{(i)} = 1] = \tr[ E(z_i) \rho_0] = 1 - \mathbb{P}_{C_\ell^{(i)}\rvert s}[C_\ell^{(i)} = 0]
\end{equation}
for all $1\leq\ell\leq m_{\mathrm{test}}$ and for all $m+1\leq i\leq 2m$. Note that the $\hat{W}$ in this expression depends on the random variables $B_\ell^{(i)}$, which in turn depend on the random variables $z_{i}$.
Similarly, by \Cref{definition:expected-true-risk}, the expected true risk is
\begin{align}
    R_\rho (\alg)
    &= \E_{\bar{S}\sim P^{2m}} \E_{\bar{C}_\ell^{(i)}\rvert \bar{S}} \E_{\bar{\hat{W}}}\left[ \hat{R}_{\bar{S}_{(m+1):2m}, \bar{C}_\ell^{(i)} }^{\mathrm{test}} (\bar{\hat{W}}) \right]\\
    &= \E_{\bar{Z}_{m+1}\sim P} \E_{\bar{C}_\ell^{(m+1)}\rvert \bar{S}_{m+1}} \E_{\bar{\hat{W}}}\left[ \left\lvert \tr[ E(\bar{Z}_{m+1}) \rho_0(\bar{\hat{W}})] - \frac{1}{m_{\mathrm{test}}}\sum_{\ell =1}^{m_{\mathrm{test}}} \bar{C}_\ell^{(m+1)}\right\rvert \right],
\end{align}
where $(\bar{Z}_{m+1}, \bar{C}_\ell^{(m+1)})$ has the same distribution as $(Z_{m+1}, C_\ell^{(m+1)})$, $\bar{\hat{W}}$ has the same distribution as $\hat{W}$ (induced via the random variables $\bar{B}_\ell^{(i)}$), but $(\bar{Z}_{m+1}, \bar{C}_\ell^{(m+1)})$ and $\bar{\hat{W}}$ are independent.

Next, we apply \Cref{corollary:qmi-gen-bound-qmgf-and-cmgf-subgaussian-independent}. 
As there is no quantum hypothesis and as the initial quantum data factorizes across the test-train bipartition, it suffices to verify the classical sub-gaussianity assumption. We can rewrite
\begin{equation}
    \tr\left[L_{c_\ell^{(i)}}^{(i)}(Z_i, w) \rho_0^{\otimes m_{\mathrm{test}}}\right]
    = \mathbb{E}_{C_\ell^{(i)}\rvert Z_i}\left[ \left\lvert \tr[ E(Z_i) \rho_0(w)] - \frac{1}{m_{\mathrm{test}}}  \sum_{\ell =1}^{m_{\mathrm{test}}} C_\ell^{(i)}\right\rvert \right] ,
\end{equation}
where, for any $m+1\leq i\leq 2m$, conditioned on $Z_i$ the random variables $C_1^{(i)},\ldots, C_{m_{\mathrm{test}}}^{(i)}$ are i.i.d., take values in $\{0,1\}$, and have mean $\tr[ E(Z_i) \rho_0]$.
So, Hoeffding's inequality \cite{hoeffding1963probability} implies that the random variable $\tr[ E(Z_i) \rho_0] - \tfrac{1}{m_{\mathrm{test}}}  \sum_{\ell =1}^{m_{\mathrm{test}}} C_\ell^{(i)}$ is $\tfrac{C}{\sqrt{m_{\mathrm{test}}}}$-sub-gaussian conditioned on $Z_i$. Here and below, we use $C$ to denote a constant that may change with each occurrence.
Next, using a triangle inequality and the equivalent formulation of sub-gaussianity in terms of $L_p$-norm bounds (compare \cite[Proposition 2.5.2]{vershynin2018high-dimensional}, we obtain the bound
\begin{align}
    &\mathbb{E}_{C_\ell^{(i)}\rvert Z_i}\left[ \left\lvert \tr[ E(Z_i) \rho_0(w)] - \frac{1}{m_{\mathrm{test}}}  \sum_{\ell =1}^{m_{\mathrm{test}}} C_\ell^{(i)}\right\rvert \right]\\
    &\leq \left\lvert \tr[ E(Z_i) \rho_0(w)] - \tr[ E(Z_i) \rho_0]\right\rvert + \mathbb{E}_{C_\ell^{(i)}\rvert Z_i}\left[ \left\lvert \tr[ E(Z_i) \rho_0] - \frac{1}{m_{\mathrm{test}}}  \sum_{\ell =1}^{m_{\mathrm{test}}} C_\ell^{(i)}\right\rvert \right]\\
    &\leq 2 + \frac{C}{\sqrt{m_{\mathrm{test}}}}
\end{align}
almost surely. 
So, the random variable $\tr[L_{c_\ell^{(i)}}^{(i)}(Z_i, w) \rho_0^{\otimes m_{\mathrm{test}}}]$, with $Z_i\sim P$, is $\left(C(1 + \tfrac{1}{\sqrt{m_{\mathrm{test}}}})\right)$-sub-gaussian by Hoeffding's Lemma \cite{hoeffding1963probability}.
Notice also that this sub-gaussianity remains true if we further condition on $Z_1,\ldots,Z_m$, since $\tr[L(Z, w) \rho_{\mathrm{test}}]$ is independent of these random variables.
Thus, first conditioning on $Z_1,\ldots,Z_m$ and then applying \Cref{corollary:qmi-gen-bound-qmgf-and-cmgf-subgaussian-independent}, we obtain the following expected generalization error bound:
\begin{align}
    \lvert \operatorname{gen}_\rho (\alg)\rvert
    &= \left\lvert \mathbb{E}_{Z_1,\ldots,Z_m}\left[\operatorname{gen}_\rho (\alg)\rvert Z_1,\ldots,Z_m\right]\right\rvert\\
    &\leq \mathbb{E}_{Z_1,\ldots,Z_m}\left[\lvert \operatorname{gen}_\rho (\alg)\rvert~\rvert Z_1,\ldots,Z_m\right]\\
    &\leq \mathbb{E}_{Z_1,\ldots,Z_m}\left[\sqrt{\left(\frac{C}{m} \left(1 + \frac{1}{\sqrt{m_{\mathrm{test}}}}\right)^2\right) I(S; \hat{W}\rvert Z_1,\ldots,Z_m)}\right] .
\end{align}

Next, we bound the conditional mutual information $I(S; \hat{W}\rvert Z_1,\ldots,Z_m)$. By construction, conditioned on $Z_1,\ldots,Z_m$, the output hypothesis random variable $\hat{W}$ takes values in $\mathsf{W}_1$. Thus, $I(S; \hat{W}\rvert Z_1,\ldots,Z_m)\leq\log_2 (\lvert\mathsf{W}_1\rvert)$. 
We can control $\lvert\mathsf{W}_1\rvert$ using bounds from classical learning theory. Notice that $\mathsf{W}_1$ is an empirical $\tilde{\varepsilon}$-covering net for (a subset of) the function class $\mathcal{F}_{\mathcal{S}(\mathbb{C}^d)}$ of $d$-dimensional quantum states viewed as functionals on effect operators, that is,
\begin{equation}
    \mathcal{F}_{\mathcal{S}(\mathbb{C}^d)}
    = \left\{ \mathcal{E}(\mathbb{C}^d)\ni E\mapsto \tr[E\rho]\in [0,1] \right\}_{\rho\in \mathcal{S}(\mathbb{C}^d)}\subseteq [0,1]^{\mathcal{E}(\mathbb{C}^d)} .
\end{equation}
By {\cite[Theorem $1$]{mendelson2003entropy} (see also \cite[Sections $12$ and $18$]{anthony1999neural},~\cite[Sections $4.2.2$ and $4.2.4$]{vidyasagar2003theory}, or \cite[Section 3.3]{caro2022phdthesis})}, we can find such a covering net of cardinality $\lvert \mathsf{W}_1\rvert \leq \left(\nicefrac{2}{\tilde{\varepsilon}}\right)^{C\cdot \operatorname{fat}(\mathcal{F}_{\mathcal{S}(\mathbb{C}^d)}, c \tilde{\varepsilon})}$, where $c,C>0$ are some constants and $\operatorname{fat}(\mathcal{F},\alpha)$ denotes the $\alpha$-fat-shattering dimension of a real-valued function class $\mathcal{F}$, introduced in \cite{kearns1994efficient}. 
For our purposes, it suffices to know that the fat-shattering dimension of $\mathcal{F}_{\mathcal{S}(\mathbb{C}^d)}$ scales logarithmically in $d$: As shown in \cite[Corollary 2.7]{Aaronson2007}, $\operatorname{fat}(\mathcal{F}_{\mathcal{S}(\mathbb{C}^d)}, \gamma)\leq \nicefrac{C\log(d)}{\gamma^2}$ holds for all $\gamma>0$, with $C>0$ some constant.
Therefore, we can take our covering net $\mathsf{W}_1$ to have cardinality $\lvert \mathsf{W}_1\rvert \leq \left(\nicefrac{2}{\tilde{\varepsilon}}\right)^{C\log(d)/\tilde{\varepsilon}^2}$, for some constant $C>0$.
This gives the conditional mutual information bound
\begin{equation}
    I(S; \hat{W}\rvert Z_1,\ldots,Z_m)
    \leq \log_2 (\lvert\mathsf{W}_1\rvert)
    \leq \frac{C\log(d)}{\tilde{\varepsilon}^2}\cdot \log\left(\frac{2}{\tilde{\varepsilon}}\right) .
\end{equation}
Plugging this back into our expected generalization error bound, we have shown:
\begin{equation}\label{eq:gen-bound-aaronson-like}
    \operatorname{gen}_\rho (\alg)
    \leq \sqrt{\left(\frac{C}{m} \left(1 + \frac{1}{\sqrt{m_{\mathrm{test}}}}\right)^2\right) \frac{\log(d)}{\tilde{\varepsilon}^2}\cdot \log\left(\frac{2}{\tilde{\varepsilon}}\right) } .
\end{equation}
This shows that we can achieve good expected generalization performance with a training data size $m$ scaling only logarithmically in the dimension $d$.

We now demonstrate the usefulness of this expected generalization error bound as a tool in bounding the expected excess prediction error of $\mathcal{A}$, which we denote by $\operatorname{excess}_\rho (\alg)$ and which is defined as the difference between the expected prediction error of $\alg$, given by 
\begin{equation}
    \E_{\bar{\hat{W}}} \E_{\bar{Z}_{m+1}} \left[ \left\lvert  \tr[E(\bar{Z}_{m+1})\rho_0] - \tr[E(\bar{Z}_{m+1})\rho_0(\bar{\hat{W}})] \right\rvert \right],
\end{equation}
and the optimal achievable expected prediction error, given by
\begin{equation}
    \inf_{w\in\mathsf{W}} \E_{\bar{Z}_{m+1}} \left[ \left\lvert  \tr[E(\bar{Z}_{m+1})\rho_0] - \tr[E(\bar{Z}_{m+1})\rho_0(w)] \right\rvert \right].
\end{equation}
Namely, based on \Cref{eq:gen-bound-aaronson-like}, we show in \Cref{appendix:proofs}: 

\begin{corollary}\label{corollary:pac-learning-states-excess-risk-bound}
    The quantum learning algorithm described above satisfies the excess prediction error bound
    \begin{equation}
        \operatorname{excess}_\rho (\alg)
        \leq \tilde{\varepsilon} + \tilde{\mathcal{O}}\left(\sqrt{\frac{\log (d)}{m \tilde{\varepsilon}^2}} + \frac{1}{\sqrt{m_{\mathrm{train}}}} + \frac{1}{\sqrt{m_{\mathrm{test}}}}\right) .
    \end{equation}
\end{corollary}

In particular, picking $\tilde{\varepsilon} = \nicefrac{\varepsilon}{2}$, our procedure achieves an expected excess prediction error of at most $\varepsilon$ for $m\leq\tilde{\mathcal{O}} (\nicefrac{\log (d)}{\varepsilon^4})$ and $m_{\mathrm{train}}, m_{\mathrm{test}} \leq\tilde{\mathcal{O}} (\nicefrac{1}{\varepsilon^2})$.
This way, our information-theoretic approach reproduces the essential feature of \cite[Theorem 1.1]{Aaronson2007}, namely the favorable dimension-dependence, as well as the $(\nicefrac{1}{\varepsilon^4})$-scaling.
Moreover, whereas \cite{Aaronson2007} starts from classical training data obtained by measuring copies of the unknown state, our analysis begins with the quantum data and thereby simultaneously leads to bounds on $m$, $m_{\mathrm{train}}$, and $m_{\mathrm{test}}$.
Here, $m_{\mathrm{train}}$ and $m_{\mathrm{test}}$ are $d$-independent. 
Note: If we consider $m$, $m_{\mathrm{train}}$, and $m_{\mathrm{test}}$ as fixed, determining our resources, then we can achieve an excess prediction error of order $\max\{\sqrt[4]{\nicefrac{\log(d)}{m}}, \sqrt{\nicefrac{1}{m_{\mathrm{train}}}}, \sqrt{\nicefrac{1}{m_{\mathrm{test}}}}\}$.

\begin{remark}
    From our reasoning leading to \Cref{corollary:pac-learning-states-excess-risk-bound}, one can extract a proof that extends the reasoning from \cite[Section 4.2]{Xu2017} beyond binary classification to regression with a continuous target space.
    This then shows how to recover in-expectation versions of known generalization bounds in terms of the fat-shattering dimension \cite{bartlett1998prediction, anthony2000function} via an information-theoretic approach to generalization and may be of independent interest.
\end{remark}

\paragraph{Extension to entangled quantum data.}
The above discussion of PAC learning quantum states assumed access to independent copies of the unknown state $\rho_0$.
We now discuss how our framework and results can be applied if the copies of $\rho_0$ are correlated/entangled across the test-train bipartition. 
This should be viewed as a proof-of-principle demonstration, similar extensions beyond the case of independent quantum data are possible also for the applications discussed in the following subsections. 
Moreover, our framework can be modified to incorporate entanglement inside the test and train subsystems, respectively, upon suitably redefining the expected true risk.

Consider CQ data of the form 
\begin{equation}
    \rho
    = \E_{S=(Z_1,\ldots,Z_{2m}) \sim P^{2m}}\left[ \left( \bigotimes_{i=1}^{2m} \ketbra{Z_i}{Z_i}\right) \otimes \tilde{\rho}\right] ,
\end{equation} 
where $\tilde{\rho}\in\mathcal{S}(\mathcal{H}_{\mathrm{data}})$ satisfies $\tr_{\mathrm{test}}[\tilde{\rho}] =(\rho_0^{\otimes m_{\mathrm{train}}})^{\otimes m}$ and $\tr_{\mathrm{train}}[\tilde{\rho}] = (\rho_0^{\otimes m_{\mathrm{train}}})^{\otimes m}$.
Let us analyze the same learning strategy as discussed above with the same choice of loss observable.
The expected empirical risk now becomes
\begin{equation}
    \hat{R}_\rho (\alg)
    = \E_{S\sim P^{2m}} \E_{D_\ell^{(i)}\rvert S} \E_{\hat{W}}\left[ \hat{R}_{S_{(m+1):2m}, D_\ell^{(i)} }^{\mathrm{test}} (\hat{W}) \right] ,
\end{equation}
where the $D_\ell^{(i)}$ are $\{0,1\}$-valued random variables that conditioned on $s$ have the joint distribution
\small
\begin{equation}
    \mathbb{P}_{\{D_\ell^{(i)}\}|s}[(D_\ell^{(i)})_{\ell,i}=(d_\ell^{(i)})_{\ell,i}]
    = \tr\left[\left(\bigotimes_{i=1}^m\bigotimes_{\ell=1}^{m_{\mathrm{test}}} (d_\ell^{(i)} E(z_i) + (1-d_\ell^{(i)})(\mathbbm{1}_d-E(z_i))\right) \rho^\alg_{\mathrm{test}}(s,\{b_\ell^{(i)}\}_{\ell,i})\right] , 
\end{equation}
\normalsize
where the $b_\ell^{(i)}$ are the measurement outcomes obtained by measuring for each $m+1\leq i\leq 2m$, the $2$-outcome POVM $\{E(z_i), \mathbbm{1}-E(z_i)\}$ on the $i^{th}$ set of $m_{\mathrm{train}}$ subsystems of $\tilde{\rho}$.
Crucially, whereas in our previous analysis the expected empirical risk depended on random variables $C_\ell^{(i)}$ that, conditioned on $s$, were independent of the outcome random variables $B_\ell^{(i)}$ seen during training (and thus of the induced hypothesis $\hat{W}$), now it depends on random variables $D_\ell^{(i)}$ that may depend on the $B_\ell^{(i)}$. 
This occurs because, due to the initially present correlations and entanglement, the collapsing measurement performed by the learner on the training data subsystem may also influence the test data subsystem.
Thus, using the ``contaminated'' test data for validation may lead to a worse risk estimate than in the i.i.d.~case.

In our definition of expected true risk, we decoupled the test and training data subsystems before letting the learner act.
This ensures that, even if correlations or entanglement are present across the test-train bipartition initially, our notion of expected true risk still reproduces the same quantity as in the case of independent quantum copies,
\begin{equation}
    R_\rho (\alg)
    = \E_{\bar{Z}_{m+1}\sim P} \E_{\bar{C}_\ell^{(m+1)}\rvert \bar{S}_{m+1}} \E_{\bar{\hat{W}}}\left[ \left\lvert \tr[ E(\bar{Z}_{m+1}) \rho_0(\bar{\hat{W}})] - \frac{1}{m_{\mathrm{test}}}\sum_{\ell =1}^{m_{\mathrm{test}}} \bar{C}_\ell^{(m+1)}\right\rvert \right] .
\end{equation}

The classical sub-gaussianity analysis is exactly the same as before. 
Now, we in addition have to determine the quantum sub-gaussianity behavior.
To this end, note that $\rho_{\mathrm{test}}(s)=(\rho_0^{\otimes m_{\mathrm{train}}})^{\otimes m}$ factorizes by assumption. 
Moreover, our loss observable is $(\nicefrac{2}{m\cdot m_{\mathrm{test}}})$-Lipschitz w.r.t.~the factorization into $m\cdot m_{\mathrm{test}}$ subsystems. (This can be seen by a bounded differences argument: If two density matrices coincide after tracing out a single subsystem, then at most one of the $\bar{C}_\ell^{(i)}$ in $\hat{R}_{s_{(m+1):2m}, \bar{C}_\ell^{(i)} }^{\mathrm{test}} (w)$ changes, leading to an overall change bounded by $\nicefrac{2}{m\cdot m_{\mathrm{test}}}$.) 
Thus, after conditioning on $Z_1,\ldots,Z_m$, we can apply \Cref{corollary:qmi-gen-bound-qmgf-and-cmgf-subgaussian-independent} and, as there is no quantum hypothesis, obtain the following generalization bound:
\begin{equation}
    \begin{split}
        \lvert\operatorname{gen}_\rho(\alg)\rvert
        &\leq \E_{Z_1,\ldots,Z_m}\left[\sqrt{\frac{8}{m\cdot m_{\mathrm{test}}}\E_{Z_{m+1},\ldots,Z_{2m}\sim P^m}\left[ \chi\left(\left\{P^{\mathcal{A}}_{B_\ell^{(i)}|S}(\{b_\ell^{(i)}\}_{\ell,i}), \rho^{\alg} _{\mathrm{test}}(S,\{b_\ell^{(i)}\}_{\ell,i})\right\}_{b_\ell^{(i)}}\right)\right]}\,\right]\\
        &\hphantom{\leq }~~+ \E_{Z_1,\ldots,Z_m}\left[\sqrt{\left(\frac{C}{m} \left(1 + \frac{1}{\sqrt{m_{\mathrm{test}}}}\right)^2\right) I(S; \hat{W}\rvert Z_1,\ldots,Z_m)}\right] \, .
    \end{split}
\end{equation}
The second summand can be controlled as in the case of i.i.d.~quantum copies.
The first summand, which can be viewed as a proxy for the maximal information about the training outcomes $b_\ell^{(i)}$ accessible from the post-measurement state $\rho^{\alg} _{\mathrm{test}}(S,\{b_\ell^{(i)}\}_{\ell,i})$ on the test subsystem, requires a separate analysis. 
Obtaining bounds on this term via quantities measuring the initial correlations/entanglement between the $\mathrm{test}$ and $\mathrm{train}$ subsystems or via properties of the POVMs used by the learner is an interesting challenge that we leave open for future work.

\subsection{Quantum PAC learning from entangled data}

Next, we demonstrate that our framework allows us to prove information-theoretic generalization bounds for quantum PAC learning from entangled data, which can be viewed as a variation on the usual standard PAC learning framework \cite{bshouty1998learning, arunachalam2017survey}.
The classical framework of \cite{Xu2017}, as reviewed in \Cref{section:introduction}, considers training data $S$ consisting of i.i.d.~examples $Z_i$ drawn from $P$. Written in terms of states diagonal in the computational basis, this data corresponds to the mixed state $\left(\sum_{z\in\mathsf{Z}} P(z)\ketbra{z}{z}\right)^{\otimes m}$. 
Instead of this classical data, we consider entangled quantum data representing a purification of this probabilistic mixture. Namely, we consider a quantum data state $\rho = (\ketbra{\phi}{\phi})^{\otimes m}$ with $\ket{\phi} = \sum_{z\in \mathsf{Z}} \sqrt{P(z)} \ket{z}_{\mathrm{test}}\otimes \ket{z}_{\mathrm{train}}$ and thus
\begin{align}
    \ket{\phi}^{\otimes m}
    &= \sum_{z_1,\ldots z_m\in \mathsf{Z}} \sqrt{P(z_1)\cdots P(z_m)} \ket{z_1,\ldots z_m}_{\mathrm{test}}\otimes \ket{z_1,\ldots z_m}_{\mathrm{train}}\\
    &= \sum_{s\in\mathsf{Z}^m} \sqrt{P^m(s)}\ket{s}_{\mathrm{test}}\otimes \ket{s}_{\mathrm{train}} \label{eq:modified-quantum-superposition-example}\, ,
\end{align}
where we identify the purifying system as the test data system. Here, the data is purely quantum, there is no classical part.
As our focus is on learning a classical function, we take $\mathsf{Z} = \mathsf{X}\times \mathsf{Y}$, with $\mathsf{X}=\{0,1\}^n$ and $\mathsf{Y}=\{0,1\}$ and accordingly $\mathcal{H}_{\mathrm{test}} = \mathcal{H}_{\mathrm{train}} = ((\mathbb{C}^{2})^{\otimes n}\otimes \mathbb{C}^2)^{\otimes m}$. We write $Z_i=(X_i,Y_i)$ and take $\mathsf{W}\subset \mathsf{Y}^{\mathsf{X}}$. 
We note that quantum data states as in \Cref{eq:modified-quantum-superposition-example} can be obtained from the more established quantum superposition examples of \cite{bshouty1998learning} by attaching an auxiliary register and applying CNOT gates, and the reverse conversion can be achieved by applying CNOTs and discarding the auxiliary system.

Before proceeding further, let us comment on how this formulation compares to the classical framework obtained by extending \cite{Xu2017} to include test data, discussed in \Cref{subsection:comparison-frameworks}.
Recall that this classical description involved perfectly correlated test and training data random variables;
the entanglement between test and training subsystems in the pure state $\ket{\phi}^{\otimes m}$ can be viewed as a fully quantum analogue of this perfect correlation, with respect to the computational basis.

We are now ready to quantumly analyze a learner that acts according to a conditional probability distribution $P^{\alg}(W|S)$.
To this end, we consider a quantum learner $\alg$ that measures the quantum data in the computational basis and processes the observed outcomes via $P^{\alg}(W|S)$. 
To model this without introducing classical random variables, we take the hypothesis space $\mathcal{H}_{\mathrm{hyp}} = \mathbb{C}^{\lvert W\rvert}$. The quantum learner $\mathcal{A}$, without performing any POVM with observed outcomes, implements the channel 
\begin{equation}
    \Lambda^\alg (\rho)
    = \sum_{s\in\mathsf{Z}^m}\sum_{w\in\mathsf{W}} \bra{s}\rho\ket{s} P^\alg (w|s) \ketbra{w}{w}\, .
\end{equation}
Thus, the state after the action of the learner is given by
\begin{equation}
    \sigma^{\alg}
    = \sum_{s\in\mathsf{Z}^m}\sum_{w\in\mathsf{W}} P^\alg (s,w) \ketbra{s}{s}_{\mathrm{test}}\otimes\ketbra{w}{w}_{\mathrm{hyp}}\, .
\end{equation}
To evaluate the performance of $\alg$, we take the loss observable $L = \tfrac{1}{m}\sum_{i=1}^m L_i$ with
\begin{equation}
    L_i
    = \sum_{z_i\in\mathsf{Z}}\sum_{w\in\mathsf{W}} \ell(w,z_i) \ketbra{z_i}{z_i}_{\mathrm{test},i}\otimes \ketbra{w}{w}_{\mathrm{hyp}}\, ,
\end{equation}
where $\ell:\mathsf{W}\times\mathsf{Z}\to\mathbb{R}_{\geq 0}$ is some classical loss function.
As the relevant operators commute, it is easy to see that this choice reproduces the clasical notions of expected empirical risk
\begin{equation}
    \tr[L \sigma^{\alg}]
    = \E_{(S,W)\sim P^\alg}\left[\hat{R}_S (W)\right]
\end{equation}
and expected true risk
\begin{equation}
    \tr[L (\rho_{\mathrm{test}}\otimes \sigma^{\alg}_{\mathrm{hyp}})]
    = \E_{(\bar{S},\bar{W})\sim P^m\otimes P^\alg_{\mathsf{W}}}\left[\hat{R}_{\bar{S}} (\bar{W})\right]
    = \E_{W\sim P^\alg_{\mathsf{W}}}\left[ R_P(W)\right]\, .
\end{equation}
These are exactly the notions of risk familiar from the classical case.

Moreover, the QMGF bound for $L$ w.r.t.~$\rho_{\mathrm{test}}\otimes \sigma^{\alg}_{\mathrm{hyp}}$ coincides with the classical MGF bound for $\tfrac{1}{m}\sum_{i=1}^m \ell(\bar{W},\bar{Z}_i)$.
Also, as $\sigma^{\alg}$ is diagonal, we see that $I(\mathrm{test};\mathrm{hyp})_{\sigma^{\alg}} = I(S;W)$.
Thus, \Cref{corollary:qmi-gen-bound-qmgf-and-cmgf-subgaussian} reproduces the main result of \cite{Xu2017} via the QMI term\footnote{While the statement of \cite[Theorem 1]{Xu2017} is correct, the argument there was based on the claim that, if $\Bar{X}, \Bar{Y}$ are independent random variables and if $f(x,\Bar{Y})$ is $\beta$-sub-gaussian for every $x$, then also $f(\Bar{X},\Bar{Y}$) is $\beta$-sub-gaussian. This claim is in general not correct because of complications regarding centering, as pointed out, e.g., in Appendix C of the arXiv version of \cite{negrea2019information-theoretic}. This issue can be circumvented by first conditioning on the hypothesis random variable (see, e.g., \cite[p.~22]{raginsky2019information}). Thus, our claim here is that we have reproduced the following version of \cite[Theorem 1]{Xu2017} without the improvement via conditioning: If $\tfrac{1}{m}\sum_{i=1}^m \ell(\bar{W},\bar{Z}_i)$ is $(\tfrac{\beta}{\sqrt{m}})$-sub-gaussian, then Eq.~(10) of \cite{Xu2017} holds.}. Here, both the classical MI and the Holevo information terms vanish because there is no classical hypothesis.

\begin{remark}
    In this section, we have described learning from quantum data in the form of a pure entangled state.
    Recently, \cite{caro2023classical} proposed \emph{mixture-of-superposition quantum examples} as an alternative to the more established superposition examples \cite{bshouty1998learning, arunachalam2017survey} for agnostic quantum learning.
    Similarly, one may change the model considered here and instead work with quantum data of the form $\rho
        = \left(\E_{f\sim F_{P}}\left[ \left(\ketbra{\phi_f}{\phi_f}\right)\right]\right)^{\otimes m} $, where
    \begin{equation}
        \ket{\phi_f}
        = \sum_{x\in\mathsf{X}} \sqrt{P_{\mathsf{X}}(x)} \ket{x,f(x)}_{\mathrm{test}}\otimes \ket{x,f(x)}_{\mathrm{train}}\, ,
    \end{equation}
    and where $F_{P}$ is the probability distribution on the function space $\{0,1\}^{\{0,1\}^n}$ induced by $P$ via
    \begin{equation}
        F_{P}(f)
        = \prod_{x'\in\{0,1\}^n} \mathbb{P}_{(x,y)\sim P}\left[f(x')=y~|~x=x'\right]\, .
    \end{equation}
    An analysis similar to the one presented above can also be carried out for this notion of quantum data and again reproduces the classical bound of \cite{Xu2017}.
\end{remark}

\subsection{Quantum parameter estimation}

Next, we demonstrate how to incorporate quantum parameter estimation tasks, typically considered in quantum metrology \cite{giovannetti2006quantum} and quantum sensing \cite{degen2017quantum}, into our framework. 
Let $\mathsf{Z}= \Theta\subseteq \mathbb{R}^n$ be a parameter space, equipped with the induced Borel $\sigma$-algebra. 
Consider the data Hilbert space
\begin{equation}
    \mathcal{H}_{\mathrm{data}}
    = \mathcal{H}_{\mathrm{test}}\otimes \mathcal{H}_{\mathrm{train}}
    = ((\mathbb{C}^d)^{\otimes m_{\mathrm{test}}})^{\otimes m}\otimes ((\mathbb{C}^d)^{\otimes m_{\mathrm{train}}})^{\otimes m} .
\end{equation} 
For an unknown probability measure $P$ over $\Theta$, let the quantum data state $\rho$ be the CQ state
\begin{equation}
    \rho
    = \mathbb{E}_{S=(Z_1,\ldots,Z_m)\sim P^m}\left[ \left(\bigotimes_{i=1}^m \ketbra{Z_i}{Z_i}\right)\otimes \left(\bigotimes_{i=1}^m \rho(Z_i)^{\otimes m_{\mathrm{test}}}\right) \otimes \left(\bigotimes_{i=1}^m \rho(Z_i)^{\otimes m_{\mathrm{train}}}\right) \right] ,
\end{equation}
where the $\rho (Z_i)$ are parameter-dependent qudit states, with the mapping $z\mapsto\rho (z)$ known in advance. 
Note: Even if this mapping is known in principle, one may not be able to prepare copies of the respective state. Thus, when aiming to learn how to extract information about the unknown parameter from the quantum system, it nevertheless makes sense to work with a finite number of copies of each $\rho(Z_i)$.

The goal of a quantum learner here is to learn a POVM that, when performed on copies of $\rho (Z)$, produces an accurate estimate of the unknown parameter $Z$. 
Therefore, to model the learner, we let $\mathcal{H}_{\mathrm{hyp}}$ be trivial, and we take $\mathsf{W}$ to be some measurable hypothesis space such that each $w\in\mathsf{W}$ is associated with a POVM $\{F_w (\hat{z})\}_{\hat{z}\in\mathsf{Z}}\subseteq\mathcal{E}((\mathbb{C}^d)^{\otimes m_{\mathrm{test}}})$. 
The action of the learner is described by POVMs $\{E_{s}^\alg (w)\}_{w\in \mathsf{W}}\subseteq \mathcal{E}(((\mathbb{C}^d)^{\otimes m_{\mathrm{train}}})^{\otimes m})$, for $s=(z_i)_i\in\mathsf{Z}^m$.
If we now define the loss observables as
\begin{equation}
    L(s, w)
    = L ((z_i)_i, w)
    = \sum_{\hat{z}\in\mathsf{Z}} \frac{1}{m}\sum_{i=1}^m \Norm{z_i - \hat{z}}_p (\mathbbm{1}_{d}^{\otimes m_{\mathrm{test}}})^{\otimes (i-1)}\otimes  F_w (\hat{z}) \otimes (\mathbbm{1}_{d}^{\otimes m_{\mathrm{test}}})^{\otimes (m-i)},
\end{equation}
for some $p\geq 1$, then we can evaluate the expected empirical risk (\Cref{definition:expected-empirical-risk}) as 
\begin{equation}
    \hat{R}_\rho (\alg)
    = \mathbb{E}_{(S,\hat{Z})\sim P^\alg} \left[ \frac{1}{m}\sum_{i=1}^m \Norm{Z_i - \hat{Z}}_p \right] ,
\end{equation}
where the classical data $S=(Z_i)_i$ and the estimated parameter $\hat{Z}$ have the joint probability distribution
\begin{equation}
    P^\alg ((z_i)_i, \hat{z})
    = \left( \prod\limits_{i=1}^m P(z_i)\right)\cdot \sum_{w\in \mathsf{W}} \tr\left[E^\alg_{(z_i)_i} (w) \left(\bigotimes_{i=1}^m \rho(z_i)^{\otimes m_{\mathrm{train}}}\right)\right] \cdot \tr\left[ F_w (\hat{z}) \rho(z_i)^{\otimes m_{\mathrm{test}}} \right]\, .
\end{equation}
Similarly, the expected true risk (\Cref{definition:expected-true-risk}) is 
\begin{equation}
    R_\rho (\alg)
    = \mathbb{E}_{\bar{Z}, \hat{Z}} \left[\Norm{\bar{Z} - \hat{Z}}_p \right] ,
\end{equation}
where 
\begin{equation}
    P^\alg (\bar{z}, \hat{z})
    = P(\bar{z})\cdot \mathbb{E}_{\bar{W}}\left[\tr\left[ F_{\bar{W}} (\hat{z}) \rho(\bar{z})^{\otimes m_{\mathrm{test}}} \right]\right],
\end{equation}
with the random variables $\bar{Z}$ and $\bar{W}$ being independent copies of $Z$ and $W$.
That is, the expected empirical risk measures the expected average norm error that estimates produced from the learned POVM make on the states that it has been learned from.
Meanwhile, the expected true risk measures the expected average norm error that estimates produced from the learned POVM make on a new parameter setting drawn at random from the underlying distribution.

We next evaluate the guarantees of \Cref{section:framework-main-result} for this setting. 
We are in the scenario of \Cref{corollary:qmi-gen-bound-qmgf-and-cmgf-subgaussian-independent} without a quantum hypothesis and without initial test-train entanglement, so it suffices to study the sub-gaussianity parameter of the random variable $\sum_{\hat{z}\in\mathsf{Z}}\Norm{Z_i - \hat{z}}_p \tr[ F_w (\hat{z}) \rho (Z_i)^{\otimes m_{\mathrm{test}}}] = \mathbb{E}_{\hat{Z}|Z_i,w} \left[\Norm{Z_i - \hat{Z}}_p \right]$ for $Z_i\sim P$ and for fixed $w$.
If we assume the parameter space $\mathsf{Z}$ to have a $p$-norm diameter $B_p<\infty$, then this random variable is bounded by $B_p$ and thus $(\tfrac{B_p}{2})$-sub-gaussian by Hoeffding.
Then, \Cref{corollary:qmi-gen-bound-qmgf-and-cmgf-subgaussian-independent} implies
\begin{equation}
    \lvert R_\rho (\alg) - \hat{R}_\rho (\alg)\rvert
    \leq \sqrt{\frac{B_p}{2m} I((Z_i)_i;W)} .
\end{equation}
Informally, this tells us: If the learned POVM performs well on the available classical-quantum data and does not depend too strongly on any specific sample parameter setting seen during training, then the POVM will also accurately extract the parameter of a previously unseen $\rho(Z)$.
Similarly to \Cref{example:quantum-state-classification-gen-bound}, we may further bound the relevant mutual information in terms of the complexity of the admissible POVMs.
To the best of our knowledge, this is the first generalization bound for quantum parameter estimation.

\subsection{Variational quantum machine learning}
 
In this subsection, we consider a task of classifying classical data via an embedding into quantum states, similarly to \cite{banchi2021generalization}.
To formalize this task, consider the data Hilbert space
\begin{equation}
    \mathcal{H}_{\mathrm{data}}
    = \mathcal{H}_{\mathrm{data}}
    = ((\mathbb{C}^d)^{\otimes m_{\mathrm{test}}})^{\otimes m}\otimes ((\mathbb{C}^d)^{\otimes m_{\mathrm{train}}})^{\otimes m} .
\end{equation} 
Let $P$ be an unknown probability measure over a measurable input space $\mathsf{X}$, let $f:\mathsf{X}\to \{1,\ldots,k\}$ be an unknown labelling function, and consider the quantum data state
\begin{equation}
    \rho
    = \mathbb{E}_{X_1,\ldots,X_{m}\sim P^m} \left[ \left(\bigotimes_{i=1}^m \ketbra{X_i, f(X_i)}{X_i, f(X_i)}\right)\otimes \left(\bigotimes_{i=1}^m \rho(X_i)^{\otimes m_{\mathrm{test}}}\right)\otimes \left(\bigotimes_{i=1}^m \rho(X_i)^{\otimes m_{\mathrm{train}}}\right)\right] ,
\end{equation}
where the $\rho(x_i)$ are quantum states into which the classical inputs $x_i$ are embedded according to a mapping $x\mapsto\rho(x)$, which may be known or unknown. 
While the mapping $x\mapsto\rho(x)$ is typically in principle known in variational QML, since it is given by the parametrized circuit, it can nevertheless make sense to work with a restricted number of copies of output states, for example if running the quantum circuit itself is expensive.
Importantly, while with a known mapping the output state and expectation values thereof could be computed classically, this will become infeasible for large system sizes. Then, using actual quantum circuits to prepare and measure states may be necessary.

The goal of a quantum learner in this scenario is to learn a POVM that, when performed on copies of $\rho(x)$, produces the correct label $f(x)$ with high probability. Accordingly, we model the learner by taking $\mathcal{H}_{\mathrm{hyp}}$ to be trivial, and by taking $\mathsf{W}$ to be some hypothesis space such that each $w\in\mathsf{W}$ is associated with a $k$-outcome POVM $\{F_w (\ell)\}_{\ell=1}^k \subseteq \mathcal{E}((\mathbb{C}^d)^{\otimes m_{\mathrm{test}}})$. We describe the action of the learner by POVMs $\{E^\alg_{((x_i,f(x_i))_{i}} (w)\}_{w\in\mathsf{W}}\subseteq \mathcal{E}(((\mathbb{C}^d)^{\otimes m_{\mathrm{train}}})^{\otimes m})$, for $((x_i,f(x_i))_{i}\in (\mathsf{X}\times\{1,\ldots,k\})^m$.
We now consider the loss observables
\small
\begin{equation}
    L(s,w)
    = L((x_i,f(x_i))_{i},w)
    = \frac{1}{m}\sum_{i=1}^m \sum_{\ell\in\{1,\ldots,k\}\setminus\{f(x_i)\}} (\mathbbm{1}_{d}^{\otimes m_{\mathrm{test}}})^{\otimes (i-1)}\otimes  F_w (\ell) \otimes (\mathbbm{1}_{d}^{\otimes m_{\mathrm{test}}})^{\otimes (m-i)} .
\end{equation}
\normalsize
According to \Cref{definition:expected-empirical-risk}, this leads to the expected empirical risk
\begin{align}
    \hat{R}_{\rho}(\mathcal{A})
    &= \mathbb{E}_{(X_1,\ldots, X_m; W)\sim P^\alg} \left[ \frac{1}{m}\sum_{i=1}^m \sum_{\ell\in\{1,\ldots,k\}\setminus\{f(X_i)\}} \tr[ F_W (\ell) \rho(X_i)^{\otimes m_{\mathrm{test}}} ] \right]\\
    &= \mathbb{E}_{(X_1,\ldots, X_m; W)\sim P^\alg} \left[ 1 - \frac{1}{m}\sum_{i=1}^m \tr[ F_W (f(X_i)) \rho(X_i)^{\otimes m_{\mathrm{test}}} ] \right],
\end{align}
Similarly, according to \Cref{definition:expected-true-risk} we obtain the expected true risk
\begin{align}
    R_{\rho}(\mathcal{A})
    &= \mathbb{E}_{\bar{X};\bar{W}}\left[ \sum_{\ell\in\{1,\ldots,k\}\setminus\{f(\bar{X})\}} \tr[ F_{\bar{W}} (\ell) \rho(\bar{X})^{\otimes m_{\mathrm{test}}} ]\right]\\
    &= \mathbb{E}_{\bar{X};\bar{W}}\left[ 1 - \tr[ F_{\bar{W}} (f(\bar{X})) \rho(\bar{X})^{\otimes m_{\mathrm{test}}} ]\right].
\end{align}
In words, $\hat{R}_{\rho}(\mathcal{A})$ is the expected average misclassification probability on the available training data, and $R_{\rho}(\mathcal{A})$ is the expected msiclassification probability on a fresh test data point.
Thus, our notions of risk are simply the expected version of those considered in \cite{banchi2021generalization}.

It remains to evaluate the guarantees proved in \Cref{section:framework-main-result} for this scenario. According to \Cref{corollary:qmi-gen-bound-qmgf-and-cmgf-subgaussian-independent}, we can focus on determining the sub-gaussianity parameter of the random variable $1 - \tr[ F_w (f(X_i)) \rho(X_i)^{\otimes m_{\mathrm{test}}} ]$ for $X_i\sim P$ and for fixed $w$. As this random variable takes values in $[0,1]$, it is $(\tfrac{1}{2})$-sub-gaussian by Hoeffding. 
So, \Cref{corollary:qmi-gen-bound-qmgf-and-cmgf-subgaussian-independent} yields the expected generalization error bound
\begin{equation}
    \lvert R_{\rho}(\mathcal{A}) - \hat{R}_{\rho}(\mathcal{A})\rvert
    \leq \sqrt{\frac{1}{2m} I(((X_i, f(X_i)))_i;W)} .
\end{equation}
We leave it as an open question whether this bound can be directly related and compared to \cite[Theorem 1]{banchi2021generalization} (see also the results in \cite{banchi2023statistical}), which depends exponentially on the $2$-Rényi mutual information between the classical input and the quantum register for a single copy.
Moreover, it will be interesting to investigate whether recent quantum generalization bounds based on (quantum) Fisher information \cite{abbas2021power,abbas2021effective, haug2023generalization} can be reinterpreted in our information-theoretic framework.
More generally, we envision that, similarly to how classical information-theoretic generalization guarantees help overcome the limitations of uniform generalization bounds pointed out in \cite{zhang2017understanding, zhang2021understanding}, a quantum information-theoretic perspective will be an important tool in remedying the drawbacks \cite{gil-fuster2023understanding} of recently established uniform generalization bounds for variational quantum machine learning \cite{caro2020pseudo, caro2021encodingdependent, chen2021expressibility, popescu2021learning, cai2021quantum, du2022efficient, caro2022generalization, gyurik2023structuralrisk}.

\subsection{Approximate quantum membership learning}

Next, we discuss a task of learning a POVM that approximately decides membership of quantum states in an a priori unknown set.
To this end, consider the data Hilbert space
\begin{equation}
    \mathcal{H}_{\mathrm{data}}
    = \mathcal{H}_{\mathrm{test}}\otimes \mathcal{H}_{\mathrm{train}}
    = ((\mathbb{C}^d)^{\otimes m_{\mathrm{test}}})^{\otimes m}\otimes ((\mathbb{C}^d)^{\otimes m_{\mathrm{train}}})^{\otimes m} .
\end{equation}
Let $P$ be an unknown probability measure over qudit states. Let $\varepsilon>0$. 
Consider the CQ data state
\begin{equation}
    \rho
    = \mathbb{E}_{\rho_1,\ldots,\rho_m\sim P^m}\left[ \left( \bigotimes_{i=1}^m \ketbra{f_{\mathcal{P},\varepsilon}(\rho_i)}{f_{\mathcal{P},\varepsilon}(\rho_i)} \right) \otimes \left(\bigotimes_{i=1}^m \rho_i^{\otimes m_\mathrm{test}}\right) \otimes \left(\bigotimes_{i=1}^m \rho_i^{\otimes m_\mathrm{train}}\right)\right] ,
\end{equation}
where $\mathcal{P}\subseteq \mathcal{S}(\mathbb{C}^d)$ is a subset of qudit states, and $f_{\mathcal{P},\varepsilon}:\mathcal{S}(\mathbb{C}^d)\to\{0,1,\perp\}$ is defined as 
\begin{equation}
    f_{\mathcal{P},\varepsilon}(\rho)
    = \begin{cases} 1 \quad&\textrm{ if } \rho\in\mathcal{P}\\ 0 &\textrm{ if } d_1 (\rho, \mathcal{P})\geq \varepsilon\\ \perp &\textrm{ else}\end{cases} .
\end{equation}
Here, we used the notation $d_1 (\rho, \mathcal{P}) = \inf_{\sigma\in\mathcal{P}} \Norm{\rho - \sigma}_1$.
Thus, given an input state $\rho$, the function value $f_{\mathcal{P},\varepsilon}(\rho)$ $\varepsilon$-approximately (and ambiguously for states with $d_1 (\rho, \mathcal{P})< \varepsilon$) decides whether $\rho$ is in $\mathcal{P}$.
If we let $Q$ denote the probability measure over $\{0,1,\perp\}\times \mathcal{S}(\mathbb{C}^d)$ induced by $P$ via $Q(z_i,\rho_i) = P(\rho_i) \delta_{z_i, f_{\mathcal{P},\varepsilon}(\rho_i)}$, then we can rewrite $\rho$ as
\begin{align}
    \rho
    &= \mathbb{E}_{(Z_1,\rho_1),\ldots , (Z_m, \rho_m)\sim Q^m}\left[ \left( \bigotimes_{i=1}^m \ketbra{Z_i}{Z_i} \right) \otimes \left(\bigotimes_{i=1}^m \rho_i^{\otimes m_\mathrm{test}}\right) \otimes \left(\bigotimes_{i=1}^m \rho_i^{\otimes m_\mathrm{train}}\right)\right]\\
    &= \mathbb{E}_{Z_1,\ldots , Z_m\sim Q_{\mathsf{Z}}^m}\left[\left( \bigotimes_{i=1}^m \ketbra{Z_i}{Z_i} \right) \otimes \left(\bigotimes_{i=1}^m \mathbb{E}_{\rho_i\sim Q_{\mathcal{S}(\mathbb{C}^d)} | Z_i}\left[ \rho_i^{\otimes m_\mathrm{test}}\otimes \rho_i^{\otimes m_\mathrm{train}}\right]\right)\right] .
\end{align}
Thus, the data state has the form of \Cref{eq:data-state}, with classical instance space $\mathsf{Z}=\{0,1,\perp\}$.
This rewriting also highlights a similarity to ambiguous state discrimination: The training data consists of a classical label (saying ``far from $\mathcal{P}$'', ``in $\mathcal{P}$'', or ''marginal case'') and a quantum part given by a conditioned average over copies of the corresponding quantum states. Given the data, the learner should essentially produce a $2$-outcome POVM that distinguishes between `` far from $\mathcal{P}$'' and ``in $\mathcal{P}$'' well on average, where the marginal cases do not matter.

More precisely, the goal of a learner is to learn a $2$-outcome POVM for deciding whether a state belongs to $\mathcal{P}$ or is $\varepsilon$-far from $\mathcal{P}$. For states that are not in $\mathcal{P}$ but less than $\varepsilon$-far from $\mathcal{P}$, any of the two outcomes is deemed acceptable.
To model such a learner, we let $\mathcal{H}_{\mathrm{hyp}}$ be trivial, and we take $\mathsf{W}$ to be some measurable hypothesis space such that each $w\in\mathsf{W}$ is associated with a POVM $\{F_w, \mathbbm{1}_{d}^{\otimes m_{\mathrm{test}}} - F_w\}\subseteq\mathcal{E}((\mathbb{C}^d)^{\otimes m_{\mathrm{test}}})$. The action of the learner is described by POVMs $\{E^\alg_{(z_i)_i} (w)\}_{w\in \mathsf{W}}\subseteq \mathcal{E}(((\mathbb{C}^d)^{\otimes m_{\mathrm{train}}})^{\otimes m})$, for $(z_i)_i\in \{0,1,\perp\}^m$.
We define the loss observables as 
\begin{align}
    L(s,w)
    &= L ((z_i)_i, w)\\
    &= \frac{1}{m}\sum_{i=1}^m (\mathbbm{1}_{d}^{\otimes m_{\mathrm{test}}})^{\otimes (i-1)}\otimes  \left( \delta_{z_i , 0} F_w + \delta_{z_i , 1} (\mathbbm{1}_{d}^{\otimes m_{\mathrm{test}}} - F_w)\right) \otimes (\mathbbm{1}_{d}^{\otimes m_{\mathrm{test}}})^{\otimes (m-i)} .
\end{align}
This leads to an expected empirical risk
\begin{equation}
    \hat{R}_\rho (\alg)
    = \mathbb{E}_{\rho_1,\ldots ,\rho_m; W} \left[ \frac{1}{m}\sum_{i=1}^m \left( \mathbf{1}_{\rho_i\in\mathcal{P}} \tr[(\mathbbm{1}_{d}^{\otimes m_{\mathrm{test}}} - F_W) \rho_i^{\otimes m_{\mathrm{test}}}] + \mathbf{1}_{d_1(\rho_i,\mathcal{P})\geq\varepsilon} \tr[F_W \rho_i^{\otimes m_{\mathrm{test}}}] \right) \right] ,
\end{equation}
where the joint distribution of $((\rho_i)_i, W)$ is given by
\begin{equation}
    P^\alg ((\rho_i)_i, W)
    = \left( \prod\limits_{i=1}^m P(\rho_i)\right)\cdot \tr\left[E^\alg_{(f_{\mathcal{P},\varepsilon}(\rho_i))_i} (W) \left(\bigotimes_{i=1}^m \rho_i^{\otimes m_{\mathrm{train}}}\right)\right]\, .
\end{equation}
The expected true risk in this case becomes
\begin{equation}
    R_\rho (\alg)
    = \mathbb{E}_{\bar{\rho}; \bar{W}} \left[ \mathbf{1}_{\bar{\rho}\in\mathcal{P}} \tr[(\mathbbm{1}_{d}^{\otimes m_{\mathrm{test}}} - F_{\bar{W}}) \bar{\rho}^{\otimes m_{\mathrm{test}}}] + \mathbf{1}_{d_1(\bar{\rho},\mathcal{P})\geq\varepsilon} \tr[F_{\bar{W}} \bar{\rho}^{\otimes m_{\mathrm{test}}}] \right] ,
\end{equation}
where $\bar{\rho}$ and $\bar{W}$ are independent random variables with joint product distribution
\begin{equation}
    P^\alg (\bar{\rho}, \bar{W})
    = P(\bar{\rho})\cdot \mathbb{E}_{\bar{\rho}_1,\ldots,\bar{\rho}_m\sim P^m}\left[\tr\left[E^\alg_{(f_{\mathcal{P},\varepsilon}(\bar{\rho}_i))_i} (\bar{W}) \left(\bigotimes_{i=1}^m \bar{\rho}_i^{\otimes m_{\mathrm{train}}}\right)\right]\right] \, .
\end{equation}
That is, the expected empirical risk is the expected average error that the learned POVM makes on the data that it was learned from. In contrast, the expected true risk is the expected average probability that the POVM makes a wrong prediction on a randomly drawn new state. (Again, the classification of marginal cases is irrelevant.)

To apply \Cref{corollary:qmi-gen-bound-qmgf-and-cmgf-subgaussian-independent}, since there is only a classical hypothesis here and since there are no initial correlations or entanglement across the test train bipartition, we study the sub-gaussianity parameter of the random variable $\tr\left[\left( \delta_{Z_i , 0} F_w + \delta_{Z_i , 1} (\mathbbm{1}_{d}^{\otimes m_{\mathrm{test}}} - F_w)\right) \mathbb{E}_{\rho_i\sim Q_{\mathcal{S}(\mathbb{C}^d)} | Z_i}\left[ \rho_i^{\otimes m_\mathrm{test}}\right]\right]$.
This random variable takes the value $0\leq \tr\left[(\mathbbm{1}_{d}^{\otimes m_{\mathrm{test}}} - F_w) \mathbb{E}_{\rho_i\sim Q_{\mathcal{S}(\mathbb{C}^d)} | 1}\left[\rho_i^{\otimes m_\mathrm{test}}\right]\right]\leq 1$ with probability $Q_{\mathsf{Z}}(1)$ and the value $0\leq \tr\left[F_w \mathbb{E}_{\rho_i\sim Q_{\mathcal{S}(\mathbb{C}^d)} | 0}\left[\rho_i^{\otimes m_\mathrm{test}}\right]\right]\leq 1$ with probability $Q_{\mathsf{Z}}(0)$.
In particular, by Hoeffding's inequality, it is $(\tfrac{1}{2})$-sub-gaussian.
Thus, \Cref{corollary:qmi-gen-bound-qmgf-and-cmgf-subgaussian-independent} implies
\begin{equation}
    \lvert R_\rho (\alg) - \hat{R}_\rho (\alg)\rvert
    \leq \sqrt{\frac{1}{2m} I((Z_i)_i;W)} .
\end{equation}
If the learner has prior knowledge indicating that membership in $\mathcal{P}$ can be (approximately) decided using only few-copy measurements and chooses the set of admissible POVMs $\{F_w, \mathbbm{1}_{d}^{\otimes m_{\mathrm{test}}} - F_w\}$ with a suitable locality structure, this is expected to lead to an improved generalization performance compared to a learner that considers general many-copy measurements as viable hypotheses (compare also the discussion in \Cref{example:quantum-state-classification-gen-bound}). 
To the best of our knowledge, we are the first to take this PAC perspective on quantum membership learning and to establish a generalization bound for it.

\begin{remark}
    The learning problem described in this section can also be interpreted as learning to solve an average-case version of quantum property testing for states, see \cite[Section 4]{montanaro2016survey}.
    From this perspective, we are asking: Given data consisting of (copies of) quantum states correctly classified according to an unknown property $\mathcal{P}$ of states and a proximity parameter $\varepsilon$, learn a POVM that tests $\mathcal{P}$ w.r.t. proximity parameter $\varepsilon$ well on average over states drawn from $P$.
    We note that formulating meaningful average-case property testing problems is subtle. For instance average-case property testing w.r.t.~uniformly random bit strings becomes trivial because of the blow-up phenomenon for Hamming distance balls \cite{goldreich2011average}.

    Complementary to the scenario discussed above, one might also consider testing for multiple properties, drawn from an unknown distribution, on a fixed (but unknown) quantum state. Here, the challenge would be to learn a mapping from a property $\mathcal{P}$ to an associated $2$-outcome POVM that classifies the unknown state $\rho_0$ according to whether it has property $\mathcal{P}$ or is $\varepsilon$-far from it.    
\end{remark}

\subsection{Learning quantum state-preparation channels from classical-quantum data}

In this section, we discuss how our framework can incorporate recent work on learning classical-to-quantum mappings \cite{chung2021sample, caro2021binary, fanizza2022learning}.
Let $\mathsf{Z}$ be some measurable instance space. Let $P$ be a probability measure over $\mathsf{Z}$.
Consider the data Hilbert space
\begin{equation}
    \mathcal{H}_{\mathrm{data}}
    = \mathcal{H}_{\mathrm{test}}\otimes \mathcal{H}_{\mathrm{train}}
    = (\mathbb{C}^d)^{\otimes m}\otimes (\mathbb{C}^d)^{\otimes m} .
\end{equation}
Take the CQ data state
\begin{equation}
    \rho
    = \E_{(Z_1,\ldots,Z_m)\sim P^m}\left[ \left(\bigotimes_{i=1}^m \ketbra{Z_i}{Z_i}\right)\otimes \left(\bigotimes_{i=1}^m \mathcal{N}(Z_i) \right)\otimes \left(\bigotimes_{i=1}^m \mathcal{N}(Z_i) \right) \right] ,
\end{equation}
where $\mathcal{N}: \mathsf{X}\to\mathcal{S}(\mathbb{C}^d)$ is an unknown qudit state-preparation channel.
The goal of a quantum learner with a hypothesis class $\{\mathcal{N}_w\}_{w\in\mathsf{W}}$ of classical descriptions of state preparation channels is to output $w$ such that performing $\mathcal{N}_w$ on inputs drawn from $P$ approximates the action of the unknown channel $\mathcal{N}$ on those inputs well in trace distance. Throughout, we assume that $\mathcal{N}(z)$ and $\mathcal{N}_w(z)$ are pure states for all $z\in\mathsf{Z}$ and $w\in\mathsf{W}$, and we therefore use notations like $\mathcal{N}(z) = \ketbra{\mathcal{N}(z)}{\mathcal{N}(z)}$ for these states. 

With our framework, we now formalize this setting for a learner that produces only a classical hypothesis, by taking $\mathcal{H}_{\mathrm{hyp}}$ to be trivial, and we define our loss observables as 
\begin{equation}
    L(s,w)
    = L((z_i)_i,w)
    = \frac{1}{m}\sum_{i=1}^m  \mathbbm{1}_d^{\otimes (i-1)}\otimes L_i(z_i,w)  \otimes \mathbbm{1}_d^{\otimes (m-i)}\, ,
\end{equation}
with local loss observables
\begin{equation}
    L_i(z_i,w) 
    = \mathbbm{1}_d - \mathcal{N}_w (z_i) \, .
\end{equation}
With these choices, \Cref{definition:expected-empirical-risk,definition:expected-true-risk} lead to the expected empirical risk
\begin{align}
    \hat{R}_\rho (\alg)
    &= \E_{(S,W)\sim P^\alg}\left[1 - \frac{1}{m}\sum_{i=1}^m\lvert\braket{\mathcal{N}_{W} (Z_i)|\mathcal{N}(Z_i)}\rvert^2\right]\\
    &= \E_{(S,W)\sim P^\alg}\left[\frac{1}{m}\sum_{i=1}^m\left(\frac{1}{2}\lVert \mathcal{N}_{W} (Z_i) - \mathcal{N} (Z_i)\rVert_1\right)^2\right]\, ,
\end{align}
and the expected true risk 
\begin{equation}
    R_\rho (\alg)
    = \E_{\bar{Z},\bar{W}}[1 - \lvert\braket{\mathcal{N}_{\bar{W}} (\bar{Z})|\mathcal{N}(\bar{Z})}\rvert^2]
    = \E_{\bar{Z},\bar{W}}\left[\left(\frac{1}{2}\lVert \mathcal{N}_{\bar{W}} (\bar{Z}) - \mathcal{N} (\bar{Z})\rVert_1\right)^2\right]\, .
\end{equation}
That is, the expected empirical risk is the expected squared trace distance between the output states of the true channel and the hypothesis channel averaged over the training data, whereas the expected true risk considers the average squared trace distance on a fresh input state.

In this scenario, we can apply \Cref{corollary:qmi-gen-bound-qmgf-and-cmgf-subgaussian-independent}. Namely, for every fixed $w\in\mathsf{W}$, the random variable $\tr[L_i (Z_i,w) \mathcal{N}(Z_i)]$ 
with $Z_i\sim P$ takes values in $[0,1]$ and thus is $(\tfrac{1}{2})$-sub-gaussian by Hoeffding. Hence, the generalization error can be bounded as
\begin{equation}
    \lvert \operatorname{gen}_\rho (\alg)\rvert
    \leq \sqrt{\frac{1}{2m} I(S;W)}\, .
\end{equation}
If $\mathsf{W}$ is finite, we can bound $I(S;W)\leq \log\lvert\mathsf{W}\rvert$, thus recovering an in-expectation version of the sample complexity bound of \cite{chung2021sample}. 
If $\mathsf{W}$ is infinite, we can resort to empirical covering net arguments similarly to \Cref{example:quantum-state-classification-gen-bound} and \Cref{section:pac-learning-quantum-states}. 
When the maps in $\mathsf{W}$ only ever output two possible quantum states, this approach, combined with standard bounds on the size of an empirical covering net via the VC-dimension \cite{vapnik1971uniform} (compare for instance \cite[Section 8.3.4]{ vershynin2018high-dimensional} and \cite[Section 3.3]{caro2022phdthesis}), leads to an in-expectation version of the guarantee proved in \cite[Section 4.1]{caro2021binary}.
More generally, using covering nets w.r.t.~empirical Schatten $q$-norms as in \cite[Definition 1]{fanizza2022learning}, we can obtain generalization bounds similar in spirit to \cite[Theorem 4]{fanizza2022learning}, which we may turn into bounds on the expected excess risk following the line of reasoning from \Cref{section:pac-learning-quantum-states}. 
Note, however, that these upper bounds on the mutual information via capacity measures are worst-case, we expect tighter data- and algorithm-dependent bounds to be possible.

Let us point out that the reasoning in this subsection was specific to state preparation channels outputting pure states, so that the overlap serves as a measurable quantity tightly related to the trace distance.
For channels outputting mixed states, other loss observables would be required to obtain risks that accurately reflect the desired average trace distance approximation to the true output states. 
In the case of only two possible known output states, one may use the Holevo-Helstrom measurement as in \cite{caro2021binary}.
However, for the general case, the ``right'' choice is not immediate. 
We believe that measurements in a random orthonormal basis as used in \cite{chung2021sample} or the quantum data analysis approach of \cite{fanizza2022learning} may serve as inspiration for how to incorporate channels with mixed output states.
Assuming purified access, an alternative route may proceed via combining the well known Fuchs-van de Graaf inequalities \cite{fuchs1999cryptographic} with a recent quantum fidelity estimation procedure for low-rank states \cite{wang2023quantum}.

\subsection{Generalization bounds for differentially private quantum learners}
Differential privacy~\cite{dwork2014algorithmic} is a robust framework that ensures the privacy of individuals in a dataset by adding controlled noise to the data or to the output of data analyses, which becomes crucial when training machine learning models on sensitive information. In machine learning, integrating differential privacy helps in mitigating the risks of data leakage and model inversion attacks, ensuring that the model's predictions do not inadvertently reveal private information about any individual in the training data. With the advent of quantum machine learning, several works tried to quantize the basic concepts, definitions and results of differential privacy~\cite{hirche2023quantum,angrisani2023unifying,angrisani2023quantum,nuradha2023quantum,aaronson2019gentle,zhou2017differential,du2021quantum}.
Classically, differentially private learners are known to satisfy mutual information stability \cite{feldman2018calibrating}, which can then be plugged into information-theoretic generalization bounds (see also \cite[Section 7.6]{hellstroem2023generalization}).
Additionally, in the case of locally differentially private (LDP) classical learners, strong data processing inequalities (DPIs) have been established (see \cite{asoodeh2023contraction, zamanlooy2023strong,angrisani2023unifying} and the references therein), which also aid in controlling the entropic quantities appearing in our bounds. 
Here, we give proof-of-principle demonstrations for how recent quantum results on contraction properties of LDP channels and measurements \cite{hirche2023quantum,angrisani2023quantum} can be used within our framework to analyze the generalization behavior of such quantum learners.

First, suppose that all channels $\Lambda^{\alg}_{s,w}$ used by the learner $\mathcal{A}$ are $\varepsilon$-LDP (see \cite[Section V]{hirche2023quantum} or \cite[Section 2.2]{angrisani2023quantum} for a definition). 
Then, combining \cite[Corollary 3.1]{angrisani2023quantum} with the Pinsker inequality, we see that
\begin{equation}\label{eq:LDP-contraction}
    I(\mathrm{test};\mathrm{hyp})_{\sigma^\alg (s,w)}\leq 2\varepsilon (1-e^{-\varepsilon}) \sqrt{2 I(\mathrm{test};\mathrm{train})_{\rho^\alg (s,w)}}\, .
\end{equation}
Using Jensen's inequality, this means that the relevant expected QMI in our generalization bounds is upper bounded as
\begin{equation}
     \E_{(S,W)\sim P^{\alg}}\left[ I(\mathrm{test};\mathrm{hyp})_{\sigma^{\alg}(S,W)}\right]
     \leq 2\sqrt{2}\varepsilon (1-e^{-\varepsilon}) \sqrt{\E_{(S,W)\sim P^\alg}[I(\mathrm{test};\mathrm{train})_{\rho^\alg (S,W)}]} \, .
\end{equation}
To further upper bound the average QMI in the post-measurement states $\rho^{\alg} (s,w)$, we can write $I(\mathrm{test};\mathrm{train})_{\rho^\alg (S,W)}$ in terms of von Neumann entropies, and use concavity of the entropy as well as the definition of the Holevo information $\chi$ to arrive at
\small
\begin{align}
    \E_{(S,W)\sim P^\alg}[I(\mathrm{test};\mathrm{train})_{\rho^\alg (S,W)}]
    &\leq I(\mathrm{test};\mathrm{train})_{\E_{(S,W)\sim P^\alg}[\rho^\alg (S,W)]} + \chi\left(\{P^\alg(s,w),\rho^\alg (s,w)\}\right)\\
    &\leq I(\mathrm{test};\mathrm{train})_{\rho} + \chi\left(\{P^\alg(s,w),\rho^\alg (s,w)\}\right)\, ,
\end{align}
\normalsize
where the last step used the data-processing inequality.
Thus, we control the QMI contribution to the generalization error in terms of the initial QMI present in the data and a proxy for the maximum accessible information about the measurement outcomes accessible from the post-measurement ensemble.
When performing a similar analysis for a learner using general (not $\varepsilon$-LDP) channels $\Lambda^\alg_{s,w}$, a direct application of DPI yields the weaker $I(\mathrm{test};\mathrm{hyp})_{\sigma^\alg (s,w)}\leq I(\mathrm{test};\mathrm{train})_{\rho^\alg (s,w)}$ instead of \Cref{eq:LDP-contraction}.
Once we note that $1-e^{-\varepsilon} = \varepsilon + \mathcal{O}(\varepsilon^2)$, we obtain the following rule of thumb: We expect the QMI contribution to the generalization error to improve by a factor of $\mathcal{O}(\varepsilon^2)$ when using $\varepsilon$-LDP quantum channels.

Next, we turn our attention to the classical MI term in our generalization bounds.
Here, we assume that the learner $\alg$ uses an overall $\varepsilon$-LDP POVM. 
As the POVM $\{\ketbra{s}{s}\otimes E_s^\alg (w)\}_{s,w}$ is not LDP even if every $\{E^\alg_s(w)\}_w$ is, we make the simplifying assumption that the learner uses an $s$-independent $\varepsilon$-LDP POVM $\{E^\alg (w)\}_w$.
Then, we can write $I(S;W)=\E_{S\sim P^m}[D(P^\alg_{\mathsf{W}|S}\| P^\alg_{\mathsf{W}})]$, where $P^\alg_{\mathsf{W}|S}$ is the outcome distribution when measuring $\{E^\alg (w)\}_w$ on $\rho(S)$, and where $P^\alg_{\mathsf{W}}$ is the outcome distribution when measuring $\{E^\alg (w)\}_w$ on $\E_{\tilde{S}\sim P^m}[\rho(\tilde{S})]$.
As we assume $\{E^\alg (w)\}_w$ to be $\varepsilon$-LDP, \cite[Lemma 3.1]{angrisani2023quantum} now implies
\begin{align}
    I(S;W)
    &\leq 2e^\varepsilon (1-e^{-\varepsilon})^2 \E_{S\sim P^m}\left[D\left(\rho(S)\Big\| \E_{\tilde{S}\sim P^m}[\rho(\tilde{S})]\right)\right]\\
    &= 2e^\varepsilon (1-e^{-\varepsilon})^2\,  \chi\left(\{P^m(s),\rho(s)\}_{s\in\mathsf{Z}^m}\right) \, ,
\end{align}
where the second step used \Cref{eq:holevo-information-alternate-expression}.
So, the classical MI contribution to the generalization error is controlled by the Holevo information of the quantum data states.
Again, compared to a general learner, we expect the classical MI contribution to the generalization error to be smaller by a factor of $\mathcal{O}(\varepsilon^2)$ when using an ($s$-independent) $\varepsilon$-LDP POVM.

In this subsection, we have used our generalization guarantees to show that requiring a quantum learner $\alg$ to be $\varepsilon$-LDP -- both in terms of the channels and the measurement used -- is expected to be beneficial for generalization performance.
While our discussion here already highlights the benefits of an LDP assumption for generalization in a broad sense, it would be interesting to instantiate this insight for specific quantum learning tasks of interest.
Moreover, while our discussion focused on local differential privacy, it does not yet apply to differentially private quantum learners. The question of whether quantum differential privacy implies a version of mutual information stability useful for quantum generalization error bounds remains open.
Finally, we have demonstrated how to use local differential privacy to control the classical and quantum mutual information terms in our generalization bounds. Investigating the effect of $\varepsilon$-LDP assumptions on the Holevo information term would give further insight into the relevance of $\varepsilon$-LDP to generalization when processing entangled quantum data.

\subsection{Generalization bounds for inductive supervised quantum learning}

\cite{monras2017inductive} considered quantum learners described by multipartite quantum channels acting on quantum training data and on the input marginals of test states.
Then, they defined the expected risk as the expectation value of a loss observable measured on the output of the learner and on the output marginals of the test states.
We can formulate this in our framework as follows: 
We take a trivial classical instance space $\mathsf{Z}$ and consider the data Hilbert space
\begin{align}
    \mathcal{H}_{\mathrm{data}}
    = \mathcal{H}_{\mathrm{test}} \otimes \mathcal{H}_{\mathrm{train}}
    &= \mathcal{H}_{\mathrm{test,out}}\otimes \left(\mathcal{H}_{\mathrm{test,in}} \otimes \mathcal{H}_{\mathrm{train},\mathrm{in}}\right)\\
    &= (\mathbb{C}^{d_\mathrm{out}})^{\otimes m_{\mathrm{test}}}\otimes ((\mathbb{C}^{d_\mathrm{in}})^{\otimes m_{\mathrm{test}}}\otimes (\mathbb{C}^{d})^{\otimes m_{\mathrm{train}}})\, .
\end{align}
Then we take a quantum data state of the form
\begin{equation}
    \rho
    = \rho_{\mathrm{test}}^{\otimes m_{\mathrm{test}}}\otimes \rho_{\mathrm{train}}\, ,
\end{equation}
with $\rho_{\mathrm{test}}\in \mathcal{S}(\mathbb{C}^{d_\mathrm{out}}\otimes \mathbb{C}^{d_\mathrm{in}})$.
The goal of the learner is to use $\rho_{\mathrm{train}}$ to predict the mapping from the input to the output parts of the test systems.

We will consider quantum learners with hypothesis space $\mathcal{H}_{\mathrm{hyp}}\cong \mathcal{H}_{\mathrm{test,out}}$ that first perform a POVM $\{\mathbbm{1}_{\mathrm{test,in}}\otimes E^\alg (w)\}_{w\in\mathsf{W}}$ that act non-trivially only on the [$\mathrm{train,in}$] subsystem, and, depending on the observed outcome, apply quantum processing of the form $(\Lambda_w^\alg)^{\otimes m_{\mathrm{test}}}\otimes \operatorname{id}_{\mathrm{train,in}}$, with each $\Lambda_w^\alg :\mathcal{T}_1(\mathbb{C}^{d_{\mathrm{in}}})\to \mathcal{T}_1(\mathbb{C}^{d_{\mathrm{out}}})$ acting only on one of the [$\mathrm{test,in}$] subsystems.
To measure the performance of such a learner, we use a local loss observable of the form
\begin{equation}\label{eq:inductive-quantum-learning-loss}
    L(w)
    = \bar{L}
    = \frac{1}{m}\sum_{i=1}^{m_{\mathrm{test}}} \mathbbm{1}_{d_{\mathrm{out}},d_{\mathrm{out}}}^{\otimes (i-1)}\otimes L_0 \otimes \mathbbm{1}_{d_{\mathrm{out}},d_{\mathrm{out}}}^{\otimes (m-i)} \, ,
\end{equation}
where $L_0\in\mathcal{B}(\mathbb{C}^{d_{\mathrm{out}}}\otimes \mathbb{C}^{d_{\mathrm{out}}})$.
With these choices, the expected empirical risk
\begin{equation}
    \E_{W\sim P^\alg_{\mathsf{W}}}\left[ \tr\left[\bar{L} \sigma^\alg(W) \right] \right]
\end{equation}
reproduces what \cite{monras2017inductive} simply call expected risk, whereas
our expected true risk
\begin{equation}
    \E_{W\sim P^\alg_{\mathsf{W}}}\left[ \tr\left[\bar{L} (\rho_{\mathrm{test},\mathrm{out}}^{\otimes m_{\mathrm{test}}}\otimes\sigma^\alg(W)_{\mathrm{hyp}})\right] \right]
\end{equation}
does not have a direct counterpart in \cite{monras2017inductive}.
Note: While the inductive (i.e., ``measure-then-process'') learners that we consider here are not the most general form of quantum learner from $\rho$, we have a motivation for this focus. Namely, formulated in our language \cite[Theorem 1]{monras2017inductive} implies that, under a non-signalling assumption, quantum learners can approximately be assumed to be inductive. 
Here, the approximation is w.r.t.~the expected empirical and true risks arising from a loss observable as in \Cref{eq:inductive-quantum-learning-loss} and improves with growing $m_{\mathrm{test}}$ because of a quantum de Finetti type behavior.

We can apply our generalization guarantees in this setting as follows: 
Notice that, since the POVM act trivially on [$\mathrm{test,in}$] and since the quantum processing is a tensor power of single-system channels, both $\sigma^{\alg}(w)$ and $\rho_{\mathrm{test,out}}\otimes\sigma^\alg(w)_{\mathrm{hyp}}$ factorize according to the tensor product structure $\mathcal{H}_{\mathrm{test,out}}\otimes\mathcal{H}_{\mathrm{hyp}} \cong  (\mathbb{C}^{d_{\mathrm{out}}}\otimes \mathbb{C}^{d_{\mathrm{out}}})^{\otimes m_{\mathrm{test}}}$.
As the loss observable is local w.r.t.~the same factorization, \Cref{corollary:qmi-gen-bound-qmgf-and-cmgf-subgaussian-independent} applies and, simply using boundedness of $L$ to get sub-gaussianity, yields the generalization bound
\begin{equation}
    \left\lvert\operatorname{gen}_\rho(\alg)\right\vert
    \leq \sqrt{\frac{C\lVert L\rVert^2}{m_{\mathrm{test}}} \E_{W\sim P^\alg_{\mathsf{W}}} \left[\sum_{i=1}^{m_{\mathrm{test}}} I(\mathrm{test,out};\mathrm{hyp})_{\sigma^\alg_{i}(W)}\right]}\, .
\end{equation}
Thus, the framework of \cite{monras2017inductive} fits naturally into our formulation, and this way our framework gives rise to a notion of generalization error that can be analyzed quantum information-theoretically.
This, to the best of our knowledge, led us to the first generalization bound that applies to arbitrary inductive quantum learners.

\section*{Acknowledgements}

The authors thank Marco Fanizza, Fredrik Hellström, and Matteo Rosati for feedback on an earlier draft.
Moreover, the authors thank Mark M.~Wilde for pointing out the connection to measured relative entropy discussed in \Cref{remark:measured-relative-entropy} and Giacomo de Palma for pointing out the reference \cite{depalma2023wassersteinspinsystems}.
Finally, the authors thank the anonymous COLT 2024 reviewers for valuable feedback that in particular helped improve the presentation.
MCC is supported by a DAAD PRIME fellowship. 
The Institute for Quantum Information and Matter is an NSF Physics Frontiers Center. 
TG is supported by UKRI Future Leaders Fellowship MR/S031545/1, EPSRC New Horizons Grant EP/X018180/1, and EPSRC Robust and Reliable Quantum Computing Grant EP/W032635/1. 
DSF is supported by France 2030 under the French National Research Agency award number “ANR-22-PNCQ-0002”. CR acknowledges financial support from the ANR project QTraj (ANR-20-CE40-0024-01) of the French National Research Agency (ANR), as well as from the Humboldt Foundation. SS is supported by a Royal Commission for the Exhibition of 1851 Research Fellowship.


\clearpage
\printbibliography

\appendix

\section{Auxiliary Results and Proofs}\label{appendix:proofs}
\numberwithin{equation}{section}

\begin{lemma}[{Restatement of \cite[Theorem 8.1]{depalma2023wassersteinspinsystems}}]\label{lemma:lipschitz-observables-subgaussian}
    Let $H$ be a Hermitian $m$-qudit observable. 
    Let $\rho = \bigotimes_{i=1}^m \rho_i$ be an $m$-fold tensor product of qudit states.
    Then, for any $\lambda\in\mathbb{R}$,
    \begin{equation}
        \tr[e^{\log (\rho) + \lambda H}]
        \leq e^{\frac{\lambda^2 m \lVert H\rVert_{\mathrm{Lip}}^2 }{2}}\, .
    \end{equation}
\end{lemma}

\begin{proof}[Proof of \Cref{corollary:qmi-gen-bound-qmgf-and-cmgf-subgaussian-independent}]
    First, note that the sub-gaussianity assumption on the $L_i(z_i, w)$ implies
    \small
    \begin{align}
        &\log \tr\left[ \left(\rho_{\mathrm{test}}(s) \otimes \sigma^{\alg}_{\mathrm{hyp}}(s,w)\right) \cdot e^{\lambda\left(L(s,w)-\tr[L(s,w) \left(\rho_{\mathrm{test}}(s) \otimes \sigma^{\alg}_{\mathrm{hyp}}(s,w)\right)]\mathds{1}_{\mathrm{test},\mathrm{hyp}}\right)}  \right]\\
        &= \sum_{i=1}^m \log \tr\left[ \left(\rho_{\mathrm{test},i}(Z_i) \otimes \sigma^{\alg}_{\mathrm{hyp},i}(z_i,w)\right) \cdot e^{\tfrac{\lambda}{m}\left(L_i(z_i,w)-\tr[L_i(z_i,w) \left(\rho_{\mathrm{test},i}(Z_i) \otimes \sigma^{\alg}_{\mathrm{hyp},i}(z_i,w)\right)]\mathds{1}_{\mathrm{test},\mathrm{hyp},i}\right)}  \right]\\
        &\leq \sum_{i=1}^m \frac{\alpha_i^2 \lambda^2}{2 m^2} .
    \end{align}
    \normalsize
    Therefore, $L(s,w)$ is $\alpha$-sub-gaussian w.r.t.~$\rho_{\mathrm{test}}(s) \otimes \sigma^{\alg}_{\mathrm{hyp}}(s,w)$ with sub-gaussianity parameter $\alpha = m^{-1}\sqrt{\sum_{i=1}^m \alpha_i^2}$, for every $(s,w)\in\mathsf{Z}^m\times \mathsf{W}$.
    
    Using the sub-gaussianity assumption on the $\tr[L_i(Z_i, w) \left( \rho_{\mathrm{test},i}(Z_i) \otimes \sigma^{\alg}_{\mathrm{hyp},i}(Z_i,w)\right)]$, we see that
    \small
    \begin{align}
        &\log\E_{S\sim P^m}\left[e^{\lambda (\tr[L(S,w)\left(\rho_{\mathrm{test}}(S) \otimes \sigma^{\alg}_{\mathrm{hyp}}(S,w)\right)] - \E_{S\sim P^m}[\tr[L(S,w)\left(\rho_{\mathrm{test}}(S) \otimes \sigma^{\alg}_{\mathrm{hyp}}(S,w)\right)]])}\right]\\
        &= \sum_{i=1}^m \log \E_{Z_i\sim P}\left[e^{\tfrac{\lambda}{m} (\tr[L_i(Z_i, w)\left(\rho_{\mathrm{test},i}(Z_i) \otimes \sigma^{\alg}_{\mathrm{hyp},i}(Z_i,w)\right)] - \E_{Z_i\sim P}[\tr[L_i(Z_i, w)\left(\rho_{\mathrm{test},i}(Z_i) \otimes \sigma^{\alg}_{\mathrm{hyp},i}(Z_i,w)\right)]])}\right]\\
        &\leq \sum_{i=1}^m \frac{\beta_i^2 \lambda^2}{2 m^2} .
    \end{align}
    \normalsize
    In other words, $\tr[L(S,w)\left(\rho_{\mathrm{test}}(S) \otimes \sigma^{\alg}_{\mathrm{hyp}}(S,w)\right)]$, with $S\sim P^m$, is $\beta$-sub-gaussian with sub-gaussianity parameter $\beta = m^{-1}\sqrt{\sum_{i=1}^m \beta^2}$, for every $w\in\mathsf{W}$.
    Therefore, we can apply \Cref{corollary:qmi-gen-bound-qmgf-and-cmgf-subgaussian} and obtain the claimed bound, once we use that $\sigma^{\alg}(s,w) = \bigotimes_{i=1}^m \sigma^{\alg}_i(z_i,w) $ implies $I(\mathrm{test};\mathrm{hyp})_{\sigma^{\alg}(s,w)} = \sum_{i=1}^m I(\mathrm{test};\mathrm{hyp})_{\sigma^{\alg}_i(z_i,w)}$.
\end{proof}

\begin{proof}[Proof of \Cref{corollary:pac-learning-states-excess-risk-bound}]
    On the one hand, we have
    \begingroup
    \allowdisplaybreaks
    \begin{align}
        &\left\lvert\E_{\bar{\hat{w}}} \E_{\bar{Z}_{m+1}} \left[ \left\lvert  \tr[E(\bar{Z}_{m+1})\rho_0] - \tr[E(\bar{Z}_{m+1})\rho_0(\bar{\hat{w}})] \right\rvert \right] - R_\rho (\alg) \right\rvert\\
        &\leq \E_{\bar{Z}_{m+1}}\E_{\bar{C}_\ell^{(m+1)}\rvert \bar{Z}_{m+1}} \left[\left\lvert   \tr[E(\bar{Z}_{m+1})\rho_0] - \frac{1}{m_{\mathrm{test}}}\sum_{\ell =1}^{m_{\mathrm{test}}} \bar{C}_\ell^{(m+1)} \right\rvert\right]\\
        &\leq \frac{C}{\sqrt{m_{\mathrm{test}}}} ,
    \end{align}
    \endgroup
    where the first step is an application of the reverse triangle inequality and the second step is via first conditioning on $\bar{Z}_{m+1}$ and then using Hoeffding-based sub-gaussianity, as already argued in \Cref{section:pac-learning-quantum-states}.
    On the other hand, we have 
    \begingroup
    \allowdisplaybreaks
    \begin{align}
        \hat{R}_\rho (\alg)
        &= \E_{S\sim P^{2m}} \E_{C_\ell^{(i)}\rvert S} \E_{\hat{W}}\left[ \hat{R}_{S_{(m+1):2m}, C_\ell^{(i)} }^{\mathrm{test}} (\hat{W}) \right]\\
        &= \E_{S\sim P^{2m}} \E_{C_\ell^{(i)}\rvert S} \E_{\hat{W}}\left[\frac{1}{m}\sum_{i=m+1}^{2m} \left\lvert \tr[ E(Z_i) \rho_0( \hat{W})] - \frac{1}{m_{\mathrm{test}}}\sum_{\ell =1}^{m_{\mathrm{test}}} C_\ell^{(i)}\right\rvert \right]\\
        &\leq \E_{S\sim P^{2m}} \E_{B_\ell^{(i)}\rvert S}\left[\frac{1}{m}\sum_{i=m+1}^{2m} \left\lvert \tr[ E(Z_i) \rho_0( \hat{W})] - \frac{1}{m_{\mathrm{train}}}\sum_{\ell =1}^{m_{\mathrm{train}}} B_\ell^{(i)}\right\rvert \right] \\
        &\hphantom{\leq~} + \E_{S\sim P^{2m}} \E_{B_\ell^{(i)}\rvert S}\left[\frac{1}{m}\sum_{i=m+1}^{2m} \left\lvert \tr[ E(Z_i) \rho_0 ] - \frac{1}{m_{\mathrm{train}}}\sum_{\ell =1}^{m_{\mathrm{train}}} B_\ell^{(i)}\right\rvert \right]\\
        &\hphantom{\leq~} + \E_{S\sim P^{2m}} \E_{C_\ell^{(i)}\rvert S}\left[\frac{1}{m}\sum_{i=m+1}^{2m} \left\lvert \tr[ E(Z_i) \rho_0] - \frac{1}{m_{\mathrm{test}}}\sum_{\ell =1}^{m_{\mathrm{test}}} C_\ell^{(i)}\right\rvert \right]\\
        &= \E_{S\sim P^{2m}} \E_{B_\ell^{(i)}\rvert S}\left[\inf_{w\in\mathsf{W}_1 } \frac{1}{m}\sum_{i=m+1}^{2m} \left\lvert \tr[ E(Z_i) \rho_0( w)] - \frac{1}{m_{\mathrm{train}}}\sum_{\ell =1}^{m_{\mathrm{train}}} B_\ell^{(i)}\right\rvert \right] \\
        &\hphantom{\leq~} + \E_{S\sim P^{2m}} \E_{B_\ell^{(i)}\rvert S}\left[\frac{1}{m}\sum_{i=m+1}^{2m} \left\lvert \tr[ E(Z_i) \rho_0 ] - \frac{1}{m_{\mathrm{train}}}\sum_{\ell =1}^{m_{\mathrm{train}}} B_\ell^{(i)}\right\rvert \right]\\
        &\hphantom{\leq~} + \E_{S\sim P^{2m}} \E_{C_\ell^{(i)}\rvert S}\left[\frac{1}{m}\sum_{i=m+1}^{2m} \left\lvert \tr[ E(Z_i) \rho_0] - \frac{1}{m_{\mathrm{test}}}\sum_{\ell =1}^{m_{\mathrm{test}}} C_\ell^{(i)}\right\rvert \right]\\
        &\leq  \E_{S\sim P^{2m}} \left[\inf_{w\in\mathsf{W}_1 } \frac{1}{m}\sum_{i=m+1}^{2m} \left\lvert \tr[ E(Z_i) \rho_0( w)] - \tr[ E(Z_i) \rho_0]\right\rvert \right]\\
        &\hphantom{\leq~} + 2\E_{S\sim P^{2m}} \E_{B_\ell^{(i)}\rvert S}\left[\frac{1}{m}\sum_{i=m+1}^{2m} \left\lvert \tr[ E(Z_i) \rho_0 ] - \frac{1}{m_{\mathrm{train}}}\sum_{\ell =1}^{m_{\mathrm{train}}} B_\ell^{(i)}\right\rvert \right]\\
        &\hphantom{\leq~} + \E_{S\sim P^{2m}} \E_{C_\ell^{(i)}\rvert S}\left[\frac{1}{m}\sum_{i=m+1}^{2m} \left\lvert \tr[ E(Z_i) \rho_0] - \frac{1}{m_{\mathrm{test}}}\sum_{\ell =1}^{m_{\mathrm{test}}} C_\ell^{(i)}\right\rvert \right]\\
        &\leq  \E_{S\sim P^{2m}} \left[\inf_{w\in\mathsf{W}_1 } \frac{1}{m}\sum_{i=m+1}^{2m} \left\lvert \tr[ E(Z_i) \rho_0( w)] - \tr[ E(Z_i) \rho_0]\right\rvert \right]\\
        &\hphantom{\leq~} + C\left( \frac{1}{\sqrt{m_{\mathrm{train}}}} + \frac{1}{\sqrt{m_{\mathrm{test}}}}\right).
    \end{align}
    \endgroup
    Here, the first step is plugging in the definition of $\hat{R}_{S_{(m+1):2m}, C_\ell^{(i)} }(\cdot)$, the second step holds by applying the triangle inequality twice, the third step uses the definition of $\hat{W}$, the fourth step is one more triangle inequality, and the final step follows from Hoeffding-type sub-gaussianity bounds.
    
    To finish the proof, we need the following fact:
    \begin{claim}\label{claim:auxiliary-covering-claim}
        With the notation introduced above,
        \begin{align}
            &\E_{S\sim P^{2m}} \left[\inf_{w\in\mathsf{W}_1 } \frac{1}{m}\sum_{i=m+1}^{2m} \left\lvert \tr[ E(Z_i) \rho_0( w)] - \tr[ E(Z_i) \rho_0]\right\rvert \right]\\ 
            &\leq \inf_{w\in\mathsf{W}} \E_{\bar{Z}_{m+1}} \left[ \left\lvert  \tr[E(\bar{Z}_{m+1})\rho_0] - \tr[E(\bar{Z}_{m+1})\rho_0(w)] \right\rvert \right] + \tilde{\varepsilon} + C\sqrt{\frac{\log (\lvert\mathsf{W}_1\rvert )}{m}} .
        \end{align}
        where $C>0$ is some positive constant.
    \end{claim}
    \begin{proof}
        See below.
    \end{proof}
    Combining \Cref{claim:auxiliary-covering-claim} with our previous upper bound on $\hat{R}_\rho (\alg)$, we have shown
    \begin{align}
        \hat{R}_\rho (\alg)
        &\leq \inf_{w\in\mathsf{W}} \E_{\bar{Z}_{m+1}} \left[ \left\lvert  \tr[E(\bar{Z}_{m+1})\rho_0] - \tr[E(\bar{Z}_{m+1})\rho_0(w)] \right\rvert \right] + \tilde{\varepsilon}\\
        &\hphantom{=~}+ C\left( \sqrt{\frac{\log (\lvert\mathsf{W}_1\rvert )}{m}} + \frac{1}{\sqrt{m_{\mathrm{train}}}} + \frac{1}{\sqrt{m_{\mathrm{test}}}}\right).
    \end{align}
    Finally, once we recall the bound $\lvert \mathsf{W}_1\rvert \leq \left(\nicefrac{2}{\tilde{\varepsilon}}\right)^{C\log(d)/\tilde{\varepsilon}^2}$ on the size of the covering net, we can bring together our upper bound on $R_\rho (\alg) $, our upper bound on $\hat{R}_\rho (\alg)$, and our generalization error bound to obtain
    \begin{align}
        \operatorname{excess}_\rho (\alg)
        &\leq \operatorname{gen}_\rho(\alg) + \tilde{\varepsilon} + \mathcal{O}\left( \sqrt{\frac{\log (\lvert\mathsf{W}_1\rvert )}{m}} + \sqrt{\frac{\log (\lvert \mathsf{W}_1\rvert)}{m_{\mathrm{train}}}} + \frac{1}{\sqrt{m_{\mathrm{test}}}}\right)\\
        &\leq \tilde{\varepsilon} + \tilde{\mathcal{O}}\left(\sqrt{\frac{\log (d)}{m \tilde{\varepsilon}^2}} + \frac{1}{\sqrt{m_{\mathrm{train}}}} + \frac{1}{\sqrt{m_{\mathrm{test}}}}\right) ,
    \end{align}
    as claimed.
\end{proof}

\begin{proof}[Proof of \Cref{claim:auxiliary-covering-claim}]
    First recall that the chosen covering net $\mathsf{W}_1$ depends only on $Z_1,\ldots,Z_m$, so that we can exchange $\E_{Z_{m+1},\ldots,Z_{2m}\sim P^{m}}$ and $\inf_{w\in\mathsf{W}_1 }$ to obtain the bound
    \begin{align}
        &\E_{S\sim P^{2m}} \left[\inf_{w\in\mathsf{W}_1 } \frac{1}{m}\sum_{i=m+1}^{2m} \left\lvert \tr[ E(Z_i) \rho_0( w)] - \tr[ E(Z_i) \rho_0]\right\rvert \right]\\
        &\leq \E_{(Z_1,\ldots,Z_m)\sim P^{m}} \left[\inf_{w\in\mathsf{W}_1 } \E_{(Z_{m+1},\ldots,Z_{2m})\sim P^{m}}\left[\frac{1}{m}\sum_{i=m+1}^{2m} \left\lvert \tr[ E(Z_i) \rho_0( w)] - \tr[ E(Z_i) \rho_0]\right\rvert \right] \right]\\
        &=\E_{(Z_1,\ldots,Z_m)\sim P^{m}}\left[\inf_{w\in\mathsf{W}_1 } \E_{\bar{Z}_{m+1}\sim P} \left[\left\lvert \tr[ E(\bar{Z}_{m+1}) \rho_0( w)] - \tr[ E(\bar{Z}_{m+1}) \rho_0]\right\rvert \right] \right].
    \end{align}
    Next, we define the following pieces of notation for the true risk with perfectly accurately evaluated quantum expectation values
    \begin{equation}
        R(w)
        :=\E_{\bar{Z}_{m+1}\sim P} \left[\left\lvert \tr[ E(\bar{Z}_{m+1}) \rho_0( w)] - \tr[ E(\bar{Z}_{m+1}) \rho_0]\right\rvert \right],
    \end{equation}
    the empirical risk with perfectly accurately evaluated quantum expectation values
    \begin{equation}
        \hat{R}_{s}(w)
        := \frac{1}{m}\sum_{j=1}^m \lvert \tr[E (z_j)\rho_0 (w)] - \tr[E (z_j)\rho_0]\rvert,
    \end{equation}
    and the corresponding true risk minimizers
    \begin{equation}
        w_{\mathsf{W}_1} \in \operatorname{argmin}_{w\in\mathsf{W}_1} R(w),~
        w_{\mathsf{W}} \in \operatorname{argmin}_{w\in\mathsf{W}} R(w),
    \end{equation}
    and empirical risk minimizers
    \begin{equation}
        \hat{w}_{\mathsf{W}_1}\in \operatorname{argmin}_{w\in\mathsf{W}_1} \hat{R}_{s}(w),~
        \hat{w}_{\mathsf{W}}\in \operatorname{argmin}_{w\in\mathsf{W}} \hat{R}_{s}(w).
    \end{equation}
    With this, we can rewrite and bound
    \begingroup
    \allowdisplaybreaks
    \begin{align}
        &\E_{(Z_1,\ldots,Z_m)\sim P^{m}}\left[\inf_{w\in\mathsf{W}_1 } \E_{\bar{Z}_{m+1}\sim P} \left[\left\lvert \tr[ E(\bar{Z}_{m+1}) \rho_0( w)] - \tr[ E(\bar{Z}_{m+1}) \rho_0]\right\rvert \right] \right] \\
        &\hphantom{=~~}- \inf_{w\in\mathsf{W}} \E_{\bar{Z}_{m+1}} \left[ \left\lvert  \tr[E(\bar{Z}_{m+1})\rho_0] - \tr[E(\bar{Z}_{m+1})\rho_0(w)] \right\rvert \right]\\
        &= \E_{(Z_1,\ldots,Z_m)\sim P^{m}}\left[ R(w_{\mathsf{W}_1})\right] - R(w_{\mathsf{W}})\\
        &= \E_{(Z_1,\ldots,Z_m)\sim P^{m}}\left[ R(w_{\mathsf{W}_1}) - R(w_{\mathsf{W}})\right]\\
        &= \E_{(Z_1,\ldots,Z_m)\sim P^{m}}\left[ R(w_{\mathsf{W}_1}) - R(w_{\mathsf{W}})\right] + \E_{(Z_1,\ldots,Z_m)\sim P^{m}}\left[ \hat{R}(\hat{w}_{\mathsf{W}_1}) - \hat{R}(w_{\mathsf{W}})\right] \\
        &\hphantom{=~}- \E_{(Z_1,\ldots,Z_m)\sim P^{m}}\left[ \hat{R}(\hat{w}_{\mathsf{W}_1}) - \hat{R}(w_{\mathsf{W}})\right]\\
        &= \E_{(Z_1,\ldots,Z_m)\sim P^{m}}\left[ R(w_{\mathsf{W}_1}) -\hat{R}(\hat{w}_{\mathsf{W}_1})\right] + \E_{(Z_1,\ldots,Z_m)\sim P^{m}}\left[ \hat{R}(\hat{w}_{\mathsf{W}_1}) - \hat{R}(w_{\mathsf{W}})\right] \\
        &\hphantom{=~}+ \underbrace{\E_{S\sim P^{m}}\left[\hat{R}(w_{\mathsf{W}}) - R(w_{\mathsf{W}})  \right]}_{=0}\\
        &= \underbrace{\E_{(Z_1,\ldots,Z_m)\sim P^{m}}\left[ R(w_{\mathsf{W}_1})-R(\hat{w}_{\mathsf{W}_1})\right]}_{\leq 0} +  \E_{(Z_1,\ldots,Z_m)\sim P^{m}}\left[R(\hat{w}_{\mathsf{W}_1}) -\hat{R}(\hat{w}_{\mathsf{W}_1})\right] \\
        &\hphantom{=~} + \underbrace{\E_{(Z_1,\ldots,Z_m)\sim P^{m}}\left[ \hat{R}(\hat{w}_{\mathsf{W}}) - \hat{R}(w_{\mathsf{W}})\right]}_{\leq 0}+ \E_{(Z_1,\ldots,Z_m)\sim P^{m}}\left[ \hat{R}(\hat{w}_{\mathsf{W}_1}) - \hat{R}(\hat{w}_{\mathsf{W}}) \right] \\
        &\leq \E_{(Z_1,\ldots,Z_m)\sim P^{m}}\left[ \sup_{w\in\mathsf{W}_1} R(w) -\hat{R}(w)\right] + \E_{(Z_1,\ldots,Z_m)\sim P^{m}}\left[ \Norm{\hat{w}_{\mathsf{W}_1} - \hat{w}_{\mathsf{W}}}_{1,\{Z_j\}_{j=1}^m} \right] \\
        &\leq \E_{(Z_1,\ldots,Z_m)\sim P^{m}}\left[ \sup_{w\in\mathsf{W}_1} R(w) -\hat{R}(w)\right] + \tilde{\varepsilon}.
    \end{align}
    \endgroup
    Here, the second-to-last step used a reverse triangle inequality, and the final step holds because $\mathsf{W}_1$ is by definition a $\tilde{\varepsilon}$-covering net for $\mathsf{W}$ w.r.t.~$\Norm{\cdot}_{2,\{Z_j\}_{j=1}^m}$ and thus also w.r.t.~$\Norm{\cdot}_{1,\{Z_j\}_{j=1}^m}$.
    Next, observe that, for any fixed $w\in\mathsf{W}_1$, the random variable $R(w) - \hat{R}(w)$ is an average of $m$ i.i.d.~centered $2$-bounded random variables and thus is $(\tfrac{C}{\sqrt{m}})$-sub-gaussian by Hoeffding's Lemma. 
    Using the equivalence of sub-gaussianity in terms of MGF bounds and tail bounds \cite[Proposition 2.5.2]{vershynin2018high-dimensional}, this can now be combined with a union bound over $\mathsf{W}_1$ to see that the random variable $\sup_{w\in\mathsf{W}_1} R(w) - \hat{R}(w)$ is $(C\sqrt{\tfrac{\log (\lvert\mathsf{W}_1\rvert )}{m}})$-sub-gaussian. 
    Therefore, using again the $L_p$ bound version of sub-gaussianity \cite[Proposition 2.5.2]{vershynin2018high-dimensional}, we conclude
    \begin{equation}
        \E_{S\sim P^{m}}\left[ \sup_{w\in\mathsf{W}_1} R(w) -\hat{R}(w)\right]
        \leq C\sqrt{\frac{\log (\lvert\mathsf{W}_1\rvert )}{m}} .
    \end{equation}
    Plugging this into our previous bound and rearranging, we get the claimed inequality.
\end{proof}

\end{document}